\documentclass[11pt, journal, onecolumn]{IEEEtran}
\pdfoutput=1
\usepackage{graphicx}
\usepackage{amsthm}
\usepackage{amsmath,amsfonts,amssymb,mathrsfs}
\usepackage{subfigure}
\usepackage{thmtools,thm-restate}
\usepackage{setspace}

\usepackage[dvipsnames]{xcolor}

\IEEEoverridecommandlockouts
\def\bigE{{\mathbb{E}}}
\def\P{{\mathbb{P}}}
\def\CF{{\mathrm{CF}}}
\def\ze{{\mathrm{ze}}}
\def\sh{{\mathrm{sh}}}
\def\ran{{\mathrm{ran}}}
\def\det{{\mathrm{det}}}
\def\S{{\mathcal{S}}}
\def\G{{\mathcal{G}}}
\def\H{{\mathcal{H}}}

\DeclareMathOperator*{\argmax}{arg\,max}

\newtheorem{thm}{Theorem}[section]
\newtheorem{defn}{Definition}[section]
\newtheorem{remark}{Remark}[section]
\newtheorem{lem}[thm]{Lemma}
\newtheorem{cor}[thm]{Corollary}
\newtheorem{example}[thm]{Example}

\let\oldexample\example
\renewcommand{\example}{\oldexample\normalfont}

\begin{document}
\title{Control Capacity}
\author{\IEEEauthorblockN{Gireeja Ranade\IEEEauthorrefmark{1}\IEEEauthorrefmark{2} and Anant Sahai\IEEEauthorrefmark{2}\thanks{This paper was presented in part at ISIT 2015~\cite{controlcapacity} and at ICC 2016~\cite{controlcapacitysideinfo}.}\\}
\IEEEauthorblockA{\IEEEauthorrefmark{1}Microsoft Research, Redmond\\}
\IEEEauthorblockA{\IEEEauthorrefmark{2}Electrical Engineering and Computer Sciences, University of California, Berkeley\\ }
gireeja@eecs.berkeley.edu, sahai@eecs.berkeley.edu
}

\maketitle
\doublespacing
\begin{abstract}
Feedback control actively dissipates uncertainty from a dynamical system by means of actuation. We develop a notion of ``control capacity'' that gives a fundamental limit (in bits) on the rate at which a controller can dissipate the uncertainty from a system, i.e. stabilize to a known fixed point. We give a computable single-letter characterization of control capacity for memoryless stationary scalar multiplicative actuation channels. Control capacity allows us to answer questions of stabilizability for scalar linear systems: a system with actuation uncertainty is stabilizable if and only if the control capacity is larger than the log of the unstable open-loop eigenvalue. 

For second-moment senses of stability, we recover the classic uncertainty threshold principle result. However, our definition of control capacity can quantify the stabilizability limits for any moment of stability. Our formulation parallels the notion of Shannon's communication capacity, and thus yields both a strong converse and a way to compute the value of side-information in control. 
The results in our paper are motivated by bit-level models for control that build on the deterministic models that are widely used to understand information flows in wireless network information theory. 

\end{abstract}


\section{Introduction}

Shannon's notion of communication capacity has been instrumental in developing communication strategies over information bottlenecks~\cite{shannon}. This powerful idea provides engineers with a language to discuss the performance of a wide-range of systems going from a single point-to-point link to complex networks. 

Just like communication systems, control systems also face many different performance bottlenecks. Some of these are explicitly informational in nature, such as rate-limited channels in networked control systems. Other performance bottlenecks come from the fact that parameters of the system might be uncertain or unpredictable. Model parameters (e.g. mass, moment-of-inertia) must be estimated from data. Other parameters (e.g. linearization constants) might be changing with time; algorithms such as iterated-LQR control require re-linearizing the system at each step. 

The impact of of model uncertainty has been investigated by robust control, and robust control provides worst-case bounds on the controllability of a system when only partial information about system parameters is available (e.g.~\cite{zhou1998essentials,el1995state}). Many works have also investigated how side-information regarding uncertain parameters can improve the performance of control systems, and this work is discussed later in the introduction. We build on these ideas and are motivated by Witsenhausen's comments in~\cite{witsenhausen1971separation}, where he points out the need for a theory of information that is compatible with control. We propose that model parameter variations/uncertainties in control systems can also be interpreted as informational bottlenecks in the system. We define a family of notions of control capacity that capture the fundamental limits of a controller's ability to stabilize a system over an unreliable actuator.

\subsection{Main results}

\begin{figure}[htbp]
\begin{center}
\includegraphics[width = .3\textwidth]{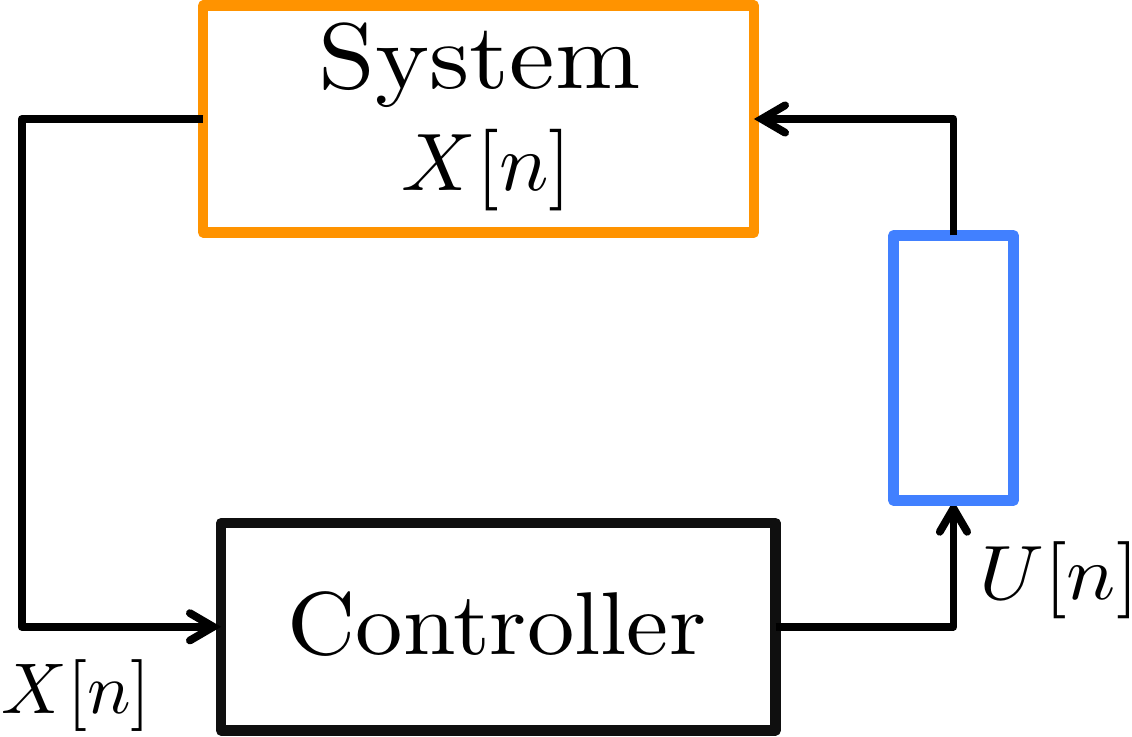}
\caption{A block diagram representation of an ``actuation'' channel, where the chosen control signal has an uncertain effect, i.e.~it is only applied after passing through an unreliable channel. Since this actuation channel models physical unreliabilities in the system (e.g.~unreliable gains) an encoder and decoder are not a part of the model.}
\label{fig:actuationchannel}
\end{center}
\end{figure}

The results in this paper focus on the stability of a scalar system with unpredictable control gains $B[n]$:
\begin{align}
\begin{split}
X[n+1] &= a \left(X[n] + B[n] U[n]\right) + W[n],\\
Y[n] &= X[n] + V[n].
\end{split}
\label{eq:sys1}
\end{align}
A block diagram for this is shown in Figure~\ref{fig:actuationchannel}.  Let the initial state $X[0] = x_{0} \neq 0$ be arbitrary, and let $V[n]$ and $W[n]$ be i.i.d. $\mathcal{N}(0,1)$ white Gaussian noise sequences. $U[n]$ is a causal control that must be computed using only the observations $Y[0]$ to $Y[n]$. The goal is to stabilize the system. While traditionally stochastic control has considered second moment stability, we consider the more general notion of $\eta$-th moment stability, i.e. to ensure that $\sup_{n}\bigE\left[\left|X\left[n\right]\right|^{\eta}\right]$ is always bounded.

$B[n]$ is a random variable that is drawn i.i.d.~from a known distribution $p_{B}$, and represents the uncertainty in how the applied control will actually influence the state. This leads us to call the block $(X[n] + B[n]U[n])$, the ``actuation channel'' in this system. 

Note that our model does not include an encoder-decoder pair around the actuation channel --- this is because we would like to model physical unreliabilities in the system as opposed to traditional communication constraints. Our objective is to develop the notion of control capacity that can be measured in bits, in order to provide an informational interpretation for these physical unreliabilities that is compatible with the language of information theory. 

The precise definition of the $\eta$-th moment control capacity ($C_{\eta}$) is deferred to Section~\ref{sec:eta}. Our first main family of results (Theorems~\ref{thm:shannoncapacitygrowth},~\ref{thm:debug},~\ref{thm:etareal}) shows that the system in~\eqref{eq:sys1} is stabilizable in the $\eta$-th moment if 
$C_{\eta} > \log |a|,$ and only if $C_{\eta} \geq \log |a|.$
The second family of results (Theorems~\ref{thm:calcshannoncap},~\ref{thm:zeroerrorcalc},~\ref{thm:calcetacap}) show that $C_{\eta}$ is actually computable, and depends only on the distribution of the i.i.d~$B$'s, i.e.,
$$C_{\eta} = \max_{d} - \frac{1}{\eta}\log \bigE\left[ \left| 1+ B \cdot d \right|^{\eta}\right].$$
It turns out the limiting cases of the notion of $\eta$-th moment stability as $\eta\to0$ and $\eta \to \infty$ are of particular interest, as the weakest and strongest potential notions of stability. The ``Shannon'' notion of control capacity, $C_{\sh},$ is what becomes relevant as $\eta \to 0$, and captures the stabilization limit for $\sup_{n}\bigE\left[\log \left|X\left[n\right]\right|\right] < \infty$. While this logarithmic sense of stability might seem artificial, Thm.~\ref{thm:tight} shows a strong converse style result for the Shannon control capacity --- without enough Shannon control capacity, it is impossible to have the state be stable in any sense. As $\eta\to\infty$ we approach the regime traditionally considered in robust control, and require stability with probability $1$. This is captured by the zero-error notion of control capacity, $C_{\ze}$. Theorem~\ref{thm:limits} establishes these limits formally.

One advantage of our information theoretic formulation is that it allows us to easily compute how side-information about a channel can improve our ability to control a system, paralleling calculations we are so familiar with in information theory. In particular, it turns out that the value of side-information about an actuation channel can be computed as a conditional expectation (Theorems~\ref{thm:sideinfocap},~\ref{thm:sideinfocapeta}). 

\subsection{Previous work}

This paper builds on many previous ideas in information theory and control. Y\"{u}ksel and Basar provide a detailed discussion and summary on the work at this intersection with a focus on understanding information structures and stabilization in their book~\cite{yuksel2013stochastic}. A recently released book by Fang et al.~\cite{fangtowards} provides some more of the history as well as recently developments on the work to integrate information theory with control, with a focus on performance-limits from the input-output perspective represented by Bode's integral formula. This section discusses some representative results and a more detailed discussion of related work can be found in~\cite{gireejaBeast}.

Stochastic unreliability in control systems has been extensively studied, starting with the uncertainty threshold principle~\cite{uncertaintyThreshold}, which considers the control of a system where both the system growth and control gain are unpredictable at each time step. The notion of control capacity provides an information-theoretic interpretation of the uncertainty threshold principle, and generalizes it to consider $\eta$-th moment stability instead of just second-moment stability. Further, control capacity can understand general uncertainty distributions (as opposed to just Gaussian distributions in~\cite{uncertaintyThreshold}). We discuss different categories of related work as well as the inspiration we have drawn from these different areas.

\subsubsection{Control with communication constraints}
Our work is strongly inspired by the family of data-rate theorems~\cite{wong1997systems, tatikonda, nairStabilization,nair2007feedback,minero2009data}. These results quantify the minimum rate required over a noiseless communication channel between the system observation and the controller to stabilize the system\footnote{One of the objectives of this line of work is to be able to consider problems with explicit rate-limits as well as parameter uncertainty in a unified framework, and this has been explored in later work in~\cite{kostina2016rate}.}. The related notion of anytime capacity considers control over noisy channels and also general notions of $\eta$-th moment stability~\cite{anytime}. Elia's work~\cite{EliaBodeShannon} uses the lens of control theory to examine the feedback capacity of channels in a way that is inspired by and extends the work of Schalkwijk and Kailath in \cite{SchalkwijkPart1, SchalkwijkPart2}. These papers ``encode'' information into the initial state of the system, and then stabilize it over a noisy channel. 

All these results allow for encoders-decoder pairs around the unreliable channels, and thus capture a traditional communication model for uncertainty. The main results have the flavor that the appropriate capacity of the bottlenecking communication channel must support a rate $R$ that is greater than a critical rate. This critical rate represents the fundamental rate at which the control system generates uncertainty --- it is typically the sum of the logs of the unstable eigenvalues of the system for linear systems. Our paper focuses on extending this aesthetic to capture the impact of physical unreliabilities. 

The data-rate theorems and anytime results consider the system plant as a ``source'' of uncertainty. This source must be communicated\footnote{In fact, the traditional data-rate theorems tend to come in pairs where the control system is paired with a pure estimation problem involving an open-loop version of the plant.} over the information bottleneck.  Our current paper focuses on how control systems can reduce uncertainty about the world by moving the world to a known point as suggested in~\cite{mitter2005information}. We refer to the dissipation of information/uncertainty as the ``sink'' nature of a control system. This difference is made salient by focusing on the control limits imposed by physical unreliabilities, which cannot be coded against. 
This represents how our perspective is different from previous data-rate theorems, while the aim of this work is to provide tools for use in conjunction with the data-rate results.

\subsubsection{Actuation channels and other parameter uncertainties}
A few other results have explicitly focused on understanding physical unreliabilities, and these have guided our explorations. Elia and co-authors were among the first to consider control actions sent over a real-erasure actuation channel~\cite{elia2004limitations,eliaFading,elia2011limitations}, with a focus on second-moment stability. They restricted the search space to consider only linear time-invariant (LTI) strategies so that the problem might become tractable and connected the problem to second-moment robust control. Related work by Schenato et al.~\cite{schenato2007foundations} and Imer et al.~\cite{imer2006optimal} also looked at this problem using dynamic programming techniques and showed that LTI strategies are in fact optimal in the infinite horizon when the effect of the control actions is available to the controller through feedback. Together these results show that the restriction to LTI strategies in~\cite{eliaFading} is in fact not a restriction at all! Our results recover some of these results for the scalar case and generalize them to $\eta$-th moment notions of stability.

Other related works by Garone et al.~\cite{garone2012lqg} and Matveev et al.~\cite{matveev2004} both consider control actions that are subject to packet drops, but allow the controller and system to use an encoder-decoder pair to treat this limitation as a traditional communication constraint. The works~\cite{schenato2007foundations, garone2012lqg} also examined the impact of control packet acknowledgements (a kind of side information); this becomes relevant for choosing between practical protocols such as TCP and UDP when building a system. In this paper, we provide a more general method to consider side-information that goes beyond the packet drop case.

Martins et al.~\cite{martinsUncertain} considered the stabilization of a system with uncertain growth rate in addition to a rate limit. Okano et al.~\cite{okano2014arxiv} also considered uncertain system growth from a robust control perspective. By contrast, the focus of our current work is on uncertainty in the control gain. Recent work~\cite{kostina2016rate} tries to provide an informational perspective that can bridge these various results.

\subsubsection{Value of information in control}
Parameter uncertainty in systems has been studied previously in control, and there has been a long quest to understand a notion of ``the value of information'' in control. One perspective on this has been provided by the idea of preview control and how it improves performance. A series of works have examined the value of information in stochastic control problems with \emph{additive} disturbances. For a standard LQG problem, Davis~\cite{davis1989anticipative} defines the value of information as the control cost reduction due to using a clairvoyant controller with non-causal knowledge. He observed that future side information about noise realizations can effectively reduce a stochastic control problem to a set of deterministic ones: one for each realization of the noise sequence. 

The area of non-anticipative control characterizes the price of not knowing the future realizations of the noise~\cite{rockafellar1976nonanticipativity,flaam1985nonanticipativity}. Rockafellar and Wets~\cite{rockafellar1976nonanticipativity} first considered this in discrete-time finite-horizon stochastic optimization problems, and then Dempster~\cite{dempster1981expected} and Flam~\cite{flaam1985nonanticipativity} extended the result for infinite horizon problems. Finally, Back and Pliska~\cite{back1987shadow} defined the shadow price of information for continuous time decision problems as well. Other related works and references for this also include~\cite{tadmor2005h,middleton2004tracking,chen2001optimal}. 

An important related result is also that of Martins et al.~\cite{martins2007fundamental}. Their paper studied how a preview of the noise can improve the frequency domain sensitivity function of the system, in the context of systems with additive noise. We are motivated by a similar spirit, however, our paper considers multiplicative parameter uncertainty in the system in addition to additive noise, and we take an information-theoretic approach where the value of information is measured in bits. 

\subsubsection{Multiplicative noise in information theory}
Because fading (channel gain) in wireless channels behaves multiplicatively and is modeled as random, such channels have been extensively studied in information theory \cite{lapidoth1998reliable}. However, even the qualitative nature of the results hinge crucially on whether the fading is known at the transmitter and/or receiver, and how fast the fades change. Most closely related to this paper is the work of Lapidoth and his collaborators, e.g.~\cite{lapidoth2002fading, lapidothLogLog}. These works consider non-coherent channels (that change too fast to predict) and show that the scaling of capacity with signal-to-noise ratio (SNR) is qualitatively different than if fading were known --- growing only as the double logarithm of the SNR rather than with the logarithm of the SNR. However, achieving even this scaling requires having an encoder and decoder around the channel and more importantly, using a specially tailored input distribution that is far from continuous.

\subsubsection{Side information in information theory}

There are two loci of uncertainty in point-to-point communication. The first is uncertainty about the system itself, i.e.~uncertainty about the channel state, and the second is uncertainty about the message bits in the system. Information theory has looked at both channel-state side information in the channel coding setting and source side information in the source coding setting. 

Wolfowitz was among the first to investigate side-information in the channel coding setting with time-varying channel state~\cite{wolfowitz1961coding}. Goldsmith and Varaiya~\cite{goldsmith1997capacity} build on this to characterize how channel-state side information at the transmitter can improve channel capacity. Caire and Shamai~\cite{caire1999capacity} provided a unified perspective to understand both causal and imperfect CSI. Lapidoth and Shamai quantify the degradation in performance due to channel-state estimation errors by the receiver~\cite{lapidoth2002fading}. Medard in~\cite{medard2000effect} examines the effect of imperfect channel knowledge on capacity for channels that are decorrelating in time and  Goldsmith and Medard~\cite{goldsmith2007capacity} further analyze causal side information for block memoryless channels. Their results recover Caire and Shamai~\cite{caire1999capacity} as a special case.


The impact of side information has also been studied in multi-terminal settings. For example, Kotagiri and Laneman~\cite{kotagiriMAC} consider a helper with non-causal knowledge of the state in a multiple access channel. Finally, there is the surprising result by Maddah-Ali and Tse~\cite{maddah2012completely}. They showed that in multi-terminal settings stale channel state information at the encoder can be useful even for a memoryless channel. It can enable a retroactive alignment of the signals. Such stale information is completely useless in a point-to-point setting.

On the source-coding side there are the classic Wyner-Ziv and Slepian-Wolf results for source coding with source side information that are found in standard textbooks~\cite{coverandthomas}. Pradhan et al.~\cite{pradhan2003duality} showed that even with side information, the duality between source and channel coding continues to hold: this is particularly interesting given the well-known parallel between source coding and portfolio theory. Source coding with fixed-delay side information can be thought of as the dual problem to channel coding with feedback~\cite{martinian2004source}. 

Another related body of work looks at the uncertainty in the distortion function as opposed to uncertainty of the source. Uncertainty regarding the distortion function is a way to model uncertainty of meaning. In this vein, Martinian et al. quantified the value of side information regarding the distortion function used to evaluate the decoder~\cite{martinian2004dist,martinian2008source}. 


Moving beyond communication, the MMSE dimension looks at the value of side information in an estimation setting. In a system with only additive noise, Wu and Verdu show that a finite number of bits of side information regarding the additive noise cannot generically change the high-SNR scaling behavior of the MMSE~\cite{wu2011mmse}. 

Finally, the ideas in this paper were inspired by the change in doubling rate with side information in portfolio theory~\cite{kelly1956new}. Permuter et al.~\cite{permuter2008directed,permuter2011interpretations} showed that directed mutual information is the gain in the doubling rate for a gambler due to causal side information.

\subsubsection{Bit-level models}
Our core results were first obtained for simplified bit-level carry-free models for uncertain control systems~\cite{controlcapacity}. These models build on previous bit-level models developed in wireless network information theory such as the deterministic models developed by Avestimehr, Diggavi and Tse (ADT models)~\cite{tseDetmodel}, and lower-triangular or carry-free models developed by Niesen and Maddah-Ali~\cite{lowerTriangular}.  Our previous work on carry-free models~\cite{carryfree} generalized these models to understand communication with noncoherent fading. The results in~\cite{carryfree} form the basis of the bit-level models for the control of systems with uncertain parameters. Some of the results in this paper were previously summarized in~\cite{controlcapacity}, which appeared at ISIT 2015.

\subsection{Outline}
The next section, Section~\ref{sec:setup} introduces the problem formulation. Subsequently, Sections~\ref{sec:shannon},~\ref{sec:zeroerror},~\ref{sec:eta} introduce the Shannon notion, the zero-error notion and the general $\eta$-th moment notions of control capacity for the simple case of systems with only multiplicative noise on the actuation channel (no additive noise). These are discussed in Section~\ref{sec:discussion}, and extended to the case of additive noise in Section~\ref{sec:additive}. Section~\ref{sec:sideinfo} then discusses how this notion of control capacity can be used to understand the value of side-information in systems. The discussion of the visual intuition from the bit-level carry-free models that inspired the work is deferred to the Appendix~\ref{sec:carryfreeappendix}.

\section{Problem setup and definitions} \label{sec:setup}
We focus our attention on a scalar system without additive system disturbances or observation noise. Section~\ref{sec:additive} will then extend these ideas to systems with additive disturbances.

First we consider a system $\S$ with system gain $a = 1$. 
\begin{align}
\begin{split}
X[n+1] &= X[n] + B[n] U[n],\\
Y[n] &= X[n]. 
\label{eq:Ssystem}
\end{split}
\end{align}
We call the basic setup in~\eqref{eq:Ssystem} the actuation channel as
it captures the basic multiplicative bottleneck in the system. The
control signal $U[n]$ can be any causal function of $Y[i]$ for $0\leq
i \leq n$. The random variables $B[k], 0 \leq i \leq n,$ are
independent. $B[k] \sim p_{B[k]}$, and these distributions are known
to the controller beforehand. Let $X[0] = x_{0} \neq 0$ be an
arbitrary but fixed known nonzero initial state. In the case where the $B[k]$'s are distributed i.i.d.~according to a distribution $p_{B}(\cdot)$, we will parameterize the system by this distribution as $\S(p_{B})$.

Our objective is to use the control capacity of the actuation channel to understand the stabilizability of the related system ${\S}_{a}$ with non-trivial intrinsic growth $a > 1$, 
\begin{align}
\begin{split}
X_{a}[n+1] &= a \left(X_{a}[n] + B[n] U_{a}[n]\right)\\
Y_{a}[n] &= X_{a}[n].\label{eq:SAsystem}
\end{split}
\end{align}
We fix the initial condition of this system to be the same as system $\S$, i.e. $X_{a}[0] = X[0] = x_{0}$. The random variables $B[i]$ in the system are the same as those of system $\mathcal{S}$ defined in \eqref{eq:Ssystem}.

We will introduce a few different notions of stability that are intimately related to each other. For each of these notions of stability we will later define a notion of control capacity of the actuation channel, which will be the maximum growth rate that can be tolerated while maintaining stability. First, we consider a notion of $\eta$-th moment stability, as has been considered in the past in~\cite{anytime} and other related works.

\begin{restatable}{defn}{etastable}
Consider $\eta > 0$. A system (e.g.~\eqref{eq:sys1}, \eqref{eq:Ssystem}, and~\eqref{eq:SAsystem}) is said to be stabilizable in the $\eta$-th moment sense if there exists a causal control strategy $U[0], U[1], \cdots$, i.e. a strategy such that each $U[k]$ is a function of the observations $Y[0]$ to $Y[k]$, such that for some $M < \infty$,
\begin{align*}
\limsup_{n\to\infty}\bigE \left[|X[n]|^{\eta}\bigr.\right] < M.
\end{align*}
\end{restatable}
$\eta$-th moment stability captures a family of stability notions as
$\eta$ varies from zero to infinity, and there is clearly an order to
this notion of stability: a system is stabilizable in an $(\eta +
\epsilon)$-th moment sense $(\epsilon > 0)$, only if it is
stabilizable in an $\eta$-th moment sense. We are particularly
interested in the limits $\eta \to 0$ and $\eta \to \infty$, and
define two more notions of stability that capture those limits in an
interpretable way.

The first of these is a notion of logarithmic stability, a sense of stability that corresponds to that of the ``zeroth'' moment; if a system is not logarithmically stabilizable, it is not $\eta$-th moment stabilizable for any $\eta$.

\begin{restatable}{defn}{logstable}
A system (e.g.~\eqref{eq:sys1}, \eqref{eq:Ssystem}, and~\eqref{eq:SAsystem}) is said to be logarithmically stabilizable if there exists a causal control strategy $U[0], U[1], \cdots$ such that for some $M < \infty$,
\begin{align*}
\limsup_{n\to\infty}\bigE\left[\log |X[n]|\right] < M.
\end{align*}
\end{restatable}

We will see later in Thm.~\ref{thm:limits} and illustrated by examples in Section~\ref{sec:discussion} that the $\eta$-th moment control capacity converges to the ``logarithmic'' control capacity as $\eta \to 0$. We are motivated to call the notion of control capacity associated with logarithmic stability as ``Shannon'' control capacity because of this convergence. It is reminiscent of the convergence of R\`{e}nyi $\alpha$-entropy to Shannon entropy as $\alpha \to 1$. 

The next notion of stability corresponds to a worst-case or traditional robust control perspective. The name ``zero-error'' stability comes from an analogy to the notion of zero-error communication capacity in information theory, which is the rate at which a communication channel can transmit bits with probability $1$. Correspondingly, in control, we require that the system be bounded by a finite box with probability 1.

\begin{restatable}{defn}{zeroerrorstable}
A system (e.g.~\eqref{eq:sys1}, \eqref{eq:Ssystem}, and~\eqref{eq:SAsystem}) is said to be
stabilizable in the zero-error sense if there exists $M < \infty$, an
$N > 0$ and a causal control strategy $U[0], U[1], \cdots$ such that
for all $n > N$, 
\begin{align}
\P\left( |X[n]| < M \right) = 1.
\end{align}
\end{restatable}

The last definition for stability we introduce here is ostensibly the weakest notion of stability for a system, and builds on the concept of tightness of measure. This notion requires that all the probability mass of the system state remain bounded, even if we are not requiring any moment to remain bounded.
\begin{defn}
We say the controller can keep the system $\S$ in~\eqref{eq:Ssystem}
tight there exists a causal control strategy $U[0], U[1], \cdots$, such that for every $\epsilon > 0$, there exist $M_{\epsilon},
N_{\epsilon}<\infty$ such that for $n > N_{\epsilon}$,
$$\mathbb{P}(|X_{n}| < M_{\epsilon}) \geq 1 - \epsilon$$
\label{def:tight}
\end{defn}
Logarithmic stability implies that the system can be kept tight (by Markov's inequality), but the reverse is not necessarily true. A further connection between these notions will be developed in Thm.~\ref{thm:tight}, which gives a control counterpart to the strong-converse in the information theory of communication channels.

With these definitions, we move to understanding the corresponding notions of control capacity. We start with logarithmic stability and the corresponding notion of control capacity as the simplest case.

\section{``Shannon'' control capacity} \label{sec:shannon}
In this section, we introduce the Shannon notion of control capacity in Def.~\ref{def:shannoncc}. After this definition, we first state and prove Thm.~\ref{thm:shannoncapacitygrowth}, which connects the control capacity of the actuation channel $\S$ to the growth and decay of the system $\S_{a}$. Then, Thm.~\ref{thm:calcshannoncap} discusses explicitly calculating the capacity through a single-letterization whenever the random variables $B[n]$ are distributed i.i.d.. Theorem~\ref{thm:tight}, the last in this section, connects logarithmic stability to the tightness of a system.

We define ``Shannon'' control capacity in the context of the system,
$\S$, with no system gain $(a = 1)$. Once we understand the decay rate
of this simple system, it can be translated to understand the stabilizability of system $\S_{a}$.

\begin{defn} \label{def:shannoncc}
The Shannon control capacity of the system $\S$ in~\eqref{eq:Ssystem} is defined as 
\begin{align}
C_{\sh}(\mathcal{S}) = \liminf_{n\rightarrow\infty}~\underset{U[0], \cdots, U[n-1]}{\max}-\frac{1}{n}\bigE\left[\log \frac{|X[n]|}{|X[0]|}\right].
\end{align}
\end{defn}

\begin{thm}
The system $\S_{a}$ in~\eqref{eq:SAsystem} is stabilizable in a
logarithmic sense if the Shannon-control capacity of the associated
system $\S$ in~\eqref{eq:Ssystem} $C_{\sh}\left(\S\right)  > \log
|a|$. Conversely, if the system $\S_{a}$ is stabilizable in a
logarithmic sense, then $C_{\sh}\left(\S\right)  \geq \log |a|$. 
\label{thm:shannoncapacitygrowth}
\end{thm}

The following Lemma, which shows that $\S$ and $\S_{a}$ can be made to track each other, is used to prove the theorem.

{\begin{restatable}{lem}{lemwhichcontrol}
Let $U[0], U[1], \cdots $ be a control strategy applied to $\S$
in~\eqref{eq:Ssystem}. Set $U_{a}[k] = a^{k}U[k]$ as the controls
applied to $\S_{a}$ in~\eqref{eq:SAsystem}. Then, $U_{a}[k]$ is computable as a function of observations $Y_{a}[0], \cdots, Y_{a}[k]$
for system $\S_{a}$, and for all $k\geq 0$ we have that $X_{a}[k] = a^{k}X[k]$. 

Similarly, if we start with $U_{a}[0], U_{a}[1], \cdots$ as a 
control strategy applied to $\S_{a}$, and we set $U[k] = a^{-k}
U_{a}[k]$ as the controls to be applied to $\S$ in \eqref{eq:Ssystem},
then each $U[k]$ is computable as a function of the observations
$Y[0], \cdots, Y[k]$, and for all $k\geq 0$ we have that $X[k] =
a^{-k}X_{a}[k]$.  
\label{lem:whichcontrol} \end{restatable}}

\begin{proof}
The proof is a consequence of linearity and proceeds by
induction. $X[0] = X_{a}[0] = x_{0}$ serves as the base case. Further,
$U_{a}[0]$ is computable as a function of $Y_{a}[0]$, since $Y_{a}[0] = x_{0}$. Now, using the controls $U_{a}[k]$ and applying the induction hypothesis gives:
\begin{align*}
X_{a}[k+1] &= a (X_{a}[k] + B[k]U_{a}[k])\\
&= a  (a^{k}X[k] + a^{k}B[k]U[k]). \\
&= a^{k+1}(X[k] + B[k]U[k]) = a^{k+1} X[k+1].
\end{align*}
Furthermore, since $Y_{a}[j] = a^{j}Y[j]$ for all $j \leq k$, we know
that $U_{a}[k]$ is function of $Y_{a}[0]$ to $Y_{a}[k]$. 
The reverse direction follows by a similar argument. 
\end{proof}

\begin{proof}[Proof of Thm.~\ref{thm:shannoncapacitygrowth}]
We first use Lemma~\ref{lem:whichcontrol} to construct an achievable scheme and show sufficiency of the Shannon control capacity. Since $C_{\sh}\left(\S\right) > \log |a|$, we know that there exists a control strategy $U[0], U[1], \cdots$ and an $N \geq 0$ such that for all $n > N$,
\begin{align*}
-\frac{1}{n}\bigE\left[\log\frac{|X[n]|}{|X[0]|} \right] \geq \log |a|.
\end{align*}
Since $|a|>1$, this can be re-written as:
\[
\bigE\left[\log\frac{|a^{n}X[n]|}{|X[0]|} \right]  \leq 0.
\]

Now, choose $U_{a}[k] = a^{k}U[k]$. Then we know from Lemma~\ref{lem:whichcontrol} that $a^{n} X[n] = X_{a}[n]$, and we can write:
\[  
\bigE\left[\log|X_{a}[n]| \right]  \leq \bigE\left[\log|x_0| \right] < \infty.
\]
Hence $\S_{a}$ is logarithmically stabilizable.

Now to show the necessity of Shannon control capacity, assume there
exists an $N, M$ and control law $U_{a}[0], U_{a}[1] \cdots $ such that $\bigE\left[\log|X_{a}[n]| \right] < M$ for all $n > N$. Hence
\[
\bigE\left[\log\frac{|X_{a}[n]|}{x_{0}} \right] < M -  \bigE\left[\log|x_{0}| \right].
\]
Applying Lemma~\ref{lem:whichcontrol} ($|X_{a}[n]| = |a^{n} X[n]|$), and dividing by $n$ gives:
\[
\log|a| + \frac{1}{n}\bigE\left[\log\frac{|X[n]|}{|x_{0}|} \right] < \frac{M -  \bigE\left[\log|x_{0}| \right]}{n},
\]
or 
\[
-\frac{1}{n}\bigE\left[\log\frac{|X[n]|}{|x_{0}|} \right] > \log |a| - \frac{M -  \bigE\left[\log|x_0| \right]}{n}.
\]
Thus, taking a limit, since $\bigE\left[\log|x_0|\right]$ is a constant, 
\[
\liminf_{n\rightarrow\infty} - \frac{1}{n}\bigE\left[\log\frac{|X[n]|}{|X[0]|} \right] \geq \log |a|, 
\]
which implies that the Shannon control capacity $C_{\sh}(\S)\geq \log|a|$.
\end{proof}

 Thm.~\ref{thm:shannoncapacitygrowth} provides us with an operational
 meaning for Shannon control capacity. However, this notion of control capacity is only valuable if we can actually compute it. The definition of control capacity involves an optimization over an infinite sequence of potential control laws, which could potentially be hard to evaluate. However, Thm.~\ref{thm:calcshannoncap} shows that this reduces to a single-letter optimization in the case where the $B[n]$'s are distributed i.i.d.~with distribution $p_{B}$. We parameterize the system with this distribution to indicate this, and denote it by $\S(p_{B})$.

Before we get to the main result, we take care of the trivial case. 

\begin{thm}
The Shannon control capacity is $\infty$ for the system $\S(p_{B})$
in~\eqref{eq:Ssystem}, with the~$B[n]$'s distributed i.i.d.~according
to $p_{B}$ if the $p_{B}$ has an atom not at $0$.
\end{thm}

\begin{proof}
Let $p_{B}$ have an atom at $\beta\neq 0$. Then consider the strategy $U[k] = -\frac{1}{\beta} X[k]$. In this case, $\frac{X[n]}{X[0]}$ can be $0$ with positive probability, and hence the negative log can be infinite, which implies that the Shannon control capacity is infinite. This captures the idea that betting on the atom will eventually pay off if we wait long enough. 
\end{proof}

\begin{thm}
The Shannon control capacity of the system $\S(p_{B})$
in~\eqref{eq:Ssystem}, where $p_{B}$ is a distribution with no atoms
(except possibly at zero), is given as:
\begin{align}
C_{\sh}(\S(p_{B})) =
  \underset{d\in\mathbb{R}}{\max}~\bigE\left[-\log|1 + B \cdot
  d|\right], \label{eq:computingShannonCapacity}
\end{align}
\noindent where $B \sim p_{B}$.  \label{thm:calcshannoncap}
\end{thm}

The proof of Theorem~\ref{thm:calcshannoncap} relies on the following lemma:
\begin{lem}
Suppose the system state at time $n$ is $x[n] \in \mathbb{R}, x[n]
\neq 0$. Then, for the system $\mathcal{S}$ in~\eqref{eq:Ssystem}
define the one-step Shannon control capacity
\begin{align*}
C_{1,\sh}(x[n]) =\underset{U[n]}{\max}~ - \bigE\left[\log \frac{|X[n+1]|}{|x[n]|}\right].
\end{align*}
where $U[n]$ is any function of $x[n]$.
\noindent Here, the expectation is taken over the random realization of $B[n] \sim p_{B}$. Hence, $X[n+1]$ is a random variable even though $X[n] =x[n]$ has been realized and fixed. Then, $C_{1,\sh}(x[n])$ does not depend on $x[n]$ or $n$, and is given by
\begin{align*}
C_{1,\sh}(x[n]) = C_{1,\sh} = \underset{d\in\mathbb{R}}{\max}~\bigE\left[ -\log |1 + B \cdot d| \right],
\end{align*}
where $B \sim p_{B}$.
Hence, there exists a scalar $d$ so that $U[n] = d \cdot Y[n] = d \cdot x[n]$ such that 
\begin{align*}
-\bigE\left[\log \biggr|\frac {X[n+1]}{x[n]}\biggr|\right] = C_{1,\sh}.
\end{align*}
\label{lem:shcclem}
\end{lem}
\begin{proof}
We can write
\begin{align*}
C_{1,\sh}(x[n]) &=\underset{U[n]}{\max} - \bigE\left[\log \frac{|X[n+1]|}{|x[n]|} \right]\\
&=\underset{U[n]}{\max}~ - \bigE\left[\log \biggr|1 + B[n] \cdot \frac{U[n]}{x[n]}\biggr|\right].
\end{align*}
Now, since $U[n]$ is a function of $x[n]$, we can replace $\frac{U[n]}{x[n]}$ by the parameter $d \in \mathbb{R}$ and optimize over that instead. Hence,
$C_{1,\sh}=\underset{d}{\max}~  \bigE\left[-\log \left|1 + B[n] \cdot d\right|\right].$ The proportionality constant $d$ in the lemma statement is just the $d$ we have found in the optimization.
\end{proof}

\begin{lem}
If the distribution $p_{B}$ has no atoms (except possibly at zero), then the distribution of $X[n]$, $p_{X[n]}$, cannot have an atom at $0$. \label{lem:nozero}
\end{lem}
\begin{proof}
We use induction to show this. Since $X[0] = x_{0} \neq 0$, it serves
as the base case. Assume the statement is true for $X[n]$,
i.e.~$p_{X[n]}$ has no atoms at zero. First, consider the case $U[n]=
0$ is applied. Then, $X[n+1] = X[n]$, and $p_{X[n+1]}$ cannot have an
atom at $0$ by the induction hypothesis.  

If $U[n] \neq 0$ is applied, then:
\begin{align*}
\P(X[n+1] = 0 ) =&~\P\left(X[n] + B[n] U[n] = 0 \right)\\
=&~\P\left(B[n] = -\frac{X[n]}{U[n]} \right).
\end{align*}
But this probability is equal to zero since $p_{B}$ has no atoms
except at zero and $\frac{X[n]}{U[n]}$ is a constant that is not equal
to zero by the induction hypothesis.
\end{proof}

This brings us to the proof of Thm.~\ref{thm:calcshannoncap}.

\begin{proof}[Proof of Thm.~\ref{thm:calcshannoncap}]

\textbf{Achievability:} We know from Lemma~\ref{lem:shcclem} that
there exists $U[n] = d Y[n]$ such that 
\begin{align}
-\bigE\left[\log \biggr|\frac{X[n+1]}{x[n]}\biggr|\right] = \underset{d\in\mathbb{R}}{\max}~\bigE\left[ -\log|1 + B[n] \cdot d| \right]
\end{align}
\noindent for every $x[n]$. Starting with $X[0] = x_{0}$ we apply the sequence of controls generated by applying Lemma~\ref{lem:shcclem} at each time step $0\leq k \leq n$. 

\noindent We can rewrite the expression for control capacity by using
a telescoping sum (division by $X[k]$ is permitted due to Lemma~\ref{lem:nozero}) and linearity of expectation as:
\begin{align}
-\frac{1}{n}\bigE\left[\log \frac{|X[n]|}{|X[0]|}\right] = -\frac{1}{n}\sum_{i=0}^{n-1} \bigE\left[\log \frac{|X[k+1]|}{|X[k]|}\right].\label{eq:linexp}
\end{align}

Plugging in the control law $U[n] = d Y[n]$ tells us
that \eqref{eq:linexp} is equal to:
\begin{align} 
-\frac{1}{n}\sum_{i=0}^{n-1} \bigE\left[\log \frac{|X[k](1+d
  B[k]|}{|X[k]|}\right] = \frac{1}{n}\sum_{i=0}^{n-1} \bigE\left[- \log
  |1+d B[k]| \right].
\end{align}

Since the $B[k]$ are i.i.d., the terms inside are identical and by
the choice of $d$ in Lemma~\ref{lem:shcclem} we have the desired
result. 

\noindent \textbf{Converse:}
We will prove this using induction. 
Recall that we are allowed to divide by $X[k]$ below due to the Lemma~\ref{lem:nozero} above. 
We can bound the term of interest as:
\begin{align}
\underset{U[0], \cdots, U[n-1]}{\max}~-\frac{1}{n}\bigE\left[\log \frac{|X[n]|}{|X[0]|}\right] 
\leq \underset{U[0], \cdots,U[n-1]}{\max}~-\frac{1}{n}\bigE\left[\log \frac{|X[n]|}{|X[n-1]|}\right] + \underset{U[0], \cdots, U[n-2]}{\max}~-\frac{1}{n}\bigE\left[\log \frac{|X[n-1]|}{|X[0]|}\right]. \label{eq:inductionsplit}
\end{align}
We condition the first term on the RHS in~\eqref{eq:inductionsplit} on $X[n-1]$ so that the inner expectation is over $B[n-1]$, and the outer is over $X[n-1]$. This gives:
\[
\bigE\left[\log \frac{|X[n]|}{|X[n-1]|}\right] = \bigE\left[\bigE\left[\log \frac{|X[n]|}{|X[n-1]|}~\biggr|~X[n-1]\right]\right]. 
\]
We can now bound this by looking at the maximizing realization of $X[n-1] = x[n-1]$. Hence,
\[
\underset{U[0], \cdots,U[n-1]}{\max}~ -\bigE\left[\bigE\left[\log
    \frac{|X[n]|}{|X[n-1]|}~\bigr|~X[n-1]\right]\right] \leq
\max_{x[n-1]\neq 0}  \underset{U[0], \cdots,U[n-1]}{\max}~ \bigE\left[- \log \frac{|X[n]|}{|x[n-1]|} \right]. 
\]
Clearly, on the RHS, the control laws before time $n-1$ no longer
matter and so by Lemma~\ref{lem:shcclem}, we can write
\[
\max_{x[n-1]\neq 0} \max_{U[n-1]} \bigE\left[- \log \frac{|X[n]|}{|x[n-1]|} \right] = \underset{d}{\max}~ \bigE\left[- \log |1 + B[n-1] \cdot d|\right].
\]
Now, using the induction hypothesis for the second term in~\eqref{eq:inductionsplit} gives the result. 
\end{proof}

\begin{remark}
As a consequence of the proof we see that linear memoryless strategies are optimal for calculating the Shannon control capacity. 
\end{remark}

Theorem~\ref{thm:calcshannoncap} allows us to relate this notion of Shannon control capacity to the tightness of systems as in Definition~\ref{def:tight} through a strong converse style result. To show this we first prove a lemma that bounds relevant random variables. 

{\begin{restatable}{lem}{lemvarbound}
Consider a random variable $B$ with a bounded density $p_{B}$ (except
possibly with an atom at $B=0$) such that  there exist 
$\gamma, \xi \geq 1$ so that the density $p_B(b) \leq \gamma \min\left(1,
\frac{\xi}{|b|}\right)$. Then, there exists a universal exponential tail
bound for the left tail of the random variable $Z_d = \ln|1 + Bd|$
for all possible values of $d$, i.e.~there exists $K_{B}$ such that for all $
t>0$, we know $\P(Z_d < -t) \leq K_{B} e^{-t}$.

If $B$ furthermore has finite first and second moments $\bigE\left[|B|\right]$
and $\bigE\left[|B|^2\right]$, then there also exists a universal upper bound 
$\overline{\sigma}_{B}^2$ that bounds the variance
$\mbox{Var}\left[Z_d\right] \leq \overline{\sigma}_{B}^2$.

\label{lem:varbound} 
\end{restatable}}

\begin{proof}
  We first establish the tail bound on $Z_{d}$. Let $Z_d = Z_d^+ + Z_d^-$ where $Z_d^+ = Z_d$ whenever $Z_d \geq
  0$ and $Z_d^- = Z_d$ whenever $Z_d < 0$. Consequently $Z_d^+ Z_d^- =
  0$ and as a result, the second moment $\bigE[Z_d^2] =
  \bigE[(Z_d^+)^2] + \bigE[(Z_d^-)^2]$.

  The event $Z_d^{-} \leq -t$ is the same as $|1+Bd|
  \leq e^{-t}$. This implies that $B$ must belong to an interval of 
  length $\frac{2}{|d|}e^{-t}$, and that if $d>0$, then $\frac{1}{d}
  (1 - e^{-t}) \leq B \leq \frac{1}{d}(1 + e^{-t})$. If $d < 0$, then $\frac{1}{d}
  (1 - e^{-t}) \geq B \geq \frac{1}{d}(1 + e^{-t})$. Notice that if $t
  > 
  0$, such an interval cannot include $B = 0$. Now consider $t \geq \ln 2$ so that all points in the interval are at a distance of at least $\frac{1}{2|d|}$ from the origin. 

  Because of this, we know from the tail bound on the density $p_{B}$, we know that the probability
  \begin{align}
   \P(Z_d^- \leq t) &\leq \frac{2}{|d|}e^{-t}  \gamma \min\left(1,
   {2 |d|\xi}\right) \nonumber \\
   &\leq \gamma 4 \xi e^{-t}, \label{eq:exptailboundharmonic}
  \end{align}
where we use $\min\left(1, 2 |d|\xi \right) \leq  2 |d|\xi.$ Putting everything together, we know 
   \begin{align}
     \P(Z_d^{-} \leq -t) &\leq
     \begin{cases}
       1 & \mbox{if }t < \ln 2 \\
       4 \gamma \xi e^{-t} & \mbox{if }t > \ln 2.
     \end{cases} \label{eq:fullexptailbound}
   \end{align}
  This establishes the tail bound on $Z_{d}^{-}$ since we can choose
  $K_{B} = 4 \gamma \xi$ as $4 e^{-\ln 2} = 2 > 1$ is a trivially
  valid bound. 

To show the variance bound, we consider the cases $d \leq 1$ and $d > 1$ separately. If $d \leq 1$, we bound the variance by the second moment.
   We integrate~\eqref{eq:fullexptailbound} to get the bound
   \begin{align}
     \bigE[(Z_d^{-})^2] &=
     \int_0^\infty \P((Z_d^{-})^2 \geq x)~dx \nonumber \\
     &= \int_0^\infty \P(Z_d^{-} \leq -\sqrt{x})~dx \nonumber \\
     &\leq (\ln 2)^2 + \int_{(\ln 2)^2}^\infty 4 \gamma \xi e^{-\sqrt{x}}~dx
       \nonumber \\
       & = (\ln 2)^2 + 4(1 + \ln 2) \gamma \xi. \label{eq:boundonnegativepartsmalld}
   \end{align}
It remains to bound $\bigE\left[ (Z_{d}^{+})^{2}\right]$. We know that  $Z_d > 0$ implies that $|1 + Bd| > 1$. Now for $x > 1$, we know that $\ln(x) \leq x -1$ and hence 
$$Z_d \leq |1 + Bd| - 1 \leq |Bd| \leq |B|,$$
where the second inequality second inequality follows by applying the triangle inequality. The last inequality follows form the assumption $|d| \leq 1$. Consequently $\bigE\left[(Z_d^+)^2\right] \leq \bigE\left[|B|^2\right]$ which is bounded by assumption.
Combing this with~\eqref{eq:boundonnegativepartsmalld} implies that that there is a universal upper bound $k_1$ such that $\mbox{Var}[Z_d] \leq k_{1}$ for all $|d| \leq 1$.

Now consider the case when $|d| > 1$. We have that
   \begin{align}
     \mbox{Var}[Z_d] &=  \mbox{Var}[\ln|d| + \ln |\frac{1}{d} + B|]
     \nonumber\\
     &= \mbox{Var}[\ln |\frac{1}{d} + B|] 
   \end{align}
   Let $\widetilde{Z}_d = \ln |\frac{1}{d} + B|$. We can essentially
   repeat the same style of arguments as in the case $d \leq 1$, except with
   some minor variations.

  Split $\widetilde{Z}_d = \widetilde{Z}_d^+ + \widetilde{Z}_d^-$
  where $\widetilde{Z}_d^+ = \widetilde{Z}_d$ whenever
  $\widetilde{Z}_d \geq 0$ and $\widetilde{Z}_d^- = \widetilde{Z}_d$
  whenever $\widetilde{Z}_d < 0$. As before, $\bigE[\widetilde{Z}_d^2]
  = \bigE[(\widetilde{Z}_d^+)^2] + \bigE[(\widetilde{Z}_d^-)^2]$.  

  Consider the negative case first. If $\widetilde{Z}_d^{-} \leq -t$,
  then we must have $|\frac{1}{d}+ B| \leq e^{-t}$ and this implies that $B$ is in
  an interval of length $2 e^{-t}$ that begins at $-\frac{1}{d} -
  e^{-t}$ and extends to $-\frac{1}{d} + e^{-t}$. As above,
  the upper bound $\gamma$ on the density of $B$ tells us that
  $\P(\widetilde{Z}_d^{-} \leq -t) \leq 2 \gamma e^{-t}$ for all
  $t>0$. Integrating this bound gives $\bigE[(\widetilde{Z}_d^{-})^2]
  \leq 4 \gamma$.

  For the positive side, if $\widetilde{Z}_d > 0$, we know that $|\frac{1}{d} + B| > 1$. Again, using the fact that $\ln (x) \leq x-1$ for $x \geq 1$, we know that $\widetilde{Z}_d \leq |\frac{1}{d} + B| - 1 \leq
   |B|$ since $d > 1$ by assumption here. Consequently
   $\bigE[(\widetilde{Z}_d^+)^2] \leq \bigE[|B|^2]$ which is 
   bounded by assumption.

   Combining the bounds on $\bigE[(\widetilde{Z}_d^{-})^2]
$ and $\bigE[(\widetilde{Z}_d^{+})^2]$ we obtain the desired bound on $\mbox{Var}[Z_{d}] \leq k_{2}$ when $d > 1$. Thus the maximum over $k_{1}$ and $k_{2}$ gives a universal
   bound on the variance of $Z_d$. 
\end{proof}

With that lemma in hand, we are ready to prove a counterpart of the
strong converse in channel coding for control capacity. If there is not sufficient Shannon capacity, then the system state
eventually blows up with probability $1$. 

Notice that the technical condition imposed on the bounded density is very mild
since the density has to integrate to $1$ anyway and so the condition
that $p_{B}(b) \leq \gamma\min(1,\frac{\xi}{|b|})$ is just ruling out
densities with ever shortening bursts of wild oscillations.

\begin{thm} \label{thm:tight}
Let $p_{B}(\cdot)$ have a bounded density (except
possibly for an atom at zero) 
such that there exists $\gamma, \xi > 1$ so that the density $p_B(b) \leq \gamma \min\left(1, \frac{\xi}{|b|}\right).$ 

If $C_{\sh}\left(\S (p_{B})\right) > \log |a|$, then the system
$\S_{a}(p_{B})$ in~\eqref{eq:SAsystem} with the $B[k]$'s
i.i.d.~according to $p_{B}$ can be kept tight. Furthermore, if
$C_{\sh}\left(\S (p_{B})\right) < \log |a|$, then for all causal
control strategies, and all bounds $M > 0$, the limiting probability
\begin{align}
  \lim_{n \rightarrow \infty} \P(|X_{a}| \geq M) =
  1. \label{eq:blowsup}
\end{align}

\end{thm}

\begin{proof}
If $C_{\sh}\left(\S (p_{B})\right) > \log |a|$, then we know that the
system $\S_{a}\left(p_{B}\right)$ can be logarithmically stabilized,
which implies that it can be kept tight.

Now, consider an actuation channel whose Shannon control capacity is
not big enough, $C_{\sh}\left(\S (p_{B})\right) = \log |a| - \epsilon$
where $\epsilon > 0$. For convenience, throughout this proof we will
take all logarithms to be natural logarithms and thus work in base $e$
instead of $2$. 

Let $U$ be an arbitrary control law. For $n>1$, if $X[n] =0$ then by
choosing $U[n] = 0$, the controller can ensure the minimum possible
$|X[n+1]| = 0$, which is the optimal action since it will keep the
state at $0$ forever. Consequently, without loss of generality we
restrict attention to control strategies that apply a $0$ control when
faced with a $0$ state. These are sample-path by sample-path as good
as or better than other strategies when it comes to keeping the state
within bounds. 

Because the initial condition $X_{a}[0] = x_{0}$ is assumed to be known and
because the controller can recall all past observations and controls,
the control law $U[n]$ can be interpreted as being a function of all
the random gains $B[0], B[1], \ldots, B[n-1]$. Let ${\mathcal
  F}_{n-1}$ be the sigma field generated by $B[0], B[1], \ldots,
B[n-1]$. Because the controller applies a zero control to a zero
state, we can re-interpret the control law as being $U[k] = D[k] X_{a}[k]$
where $D[k] = \frac{U[k]}{X_{a}[k]}$ and hence a deterministic function of
$B[0], B[1], \ldots, B[n-1]$. 

Using this to expand \eqref{eq:SAsystem}, we see that
\begin{align}
|X_{a}[n]| &= \left|a \left(1 + D[n-1] \cdot B[n-1]\right) X_{a}[n-1] \right| \nonumber\\
&= \left| x_{0}\right| \prod_{k=0}^{n-1} |a|\left| 1 + D[k] \cdot B[k]
\right| \nonumber \\
&= \left| x_{0}\right| e^{\epsilon n} \prod_{k=0}^{n-1}
|a|e^{-\epsilon}\left| 1 + D[k] \cdot B[k]
\right| \label{eq:prod1martingale}
\end{align}

Take the natural log to get
\begin{align}
  \ln |X_{a}[n]| = \ln|x_0| + \epsilon n + \sum_{k=0}^{n-1} \ln(|a|e^{-\epsilon}\left| 1 + D[k]
  \cdot B[k] \right|). \label{eq:logstate}
\end{align}

Let $J[k] = \ln(|a|e^{-\epsilon}\left| 1 + D[k] \cdot B[k] \right|) =
\ln|a| - \epsilon + \ln \left| 1 + D[k] \cdot B[k] \right|$.

It turns out that $\{(J[k],{\mathcal F}_k)\}_{k=0}^{\infty}$ is a
sub-martingale difference sequence.
\begin{align}
  \bigE\left[J[k] | {\mathcal F}_{k-1}\right] &= \ln|a| - \epsilon +
  \bigE\left[\ln | 1 + D[k] \cdot B[k] | \big| {\mathcal F}_{k-1} \right] \nonumber \\
  &\geq \ln|a| - \epsilon -(\ln |a| -
  \epsilon) \label{eq:substituteinmax} \\
  &\geq 0 \nonumber
\end{align}
  where \eqref{eq:substituteinmax} comes from the fact that $D[k]$ is
  a deterministic function of $B[0], B[1], \ldots, B[k-1]$ and hence a
  constant in that expectation, and Shannon control capacity is the
  maximum that $-\bigE[\ln|1 + d B|]$ can be over the choice of
  constants $d$.

Define a process of strictly non-negative random variables $H[k] =
\bigE\left[J[k] | {\mathcal F}_{k-1}\right] \geq 0$, and consider $J[k] =
H[k] + (J[k] - H[k])$. Clearly, $\{((J[k] - H[k]),{\mathcal
  F}_k)\}_{k=0}^{\infty}$ is a martingale difference sequence.

Furthermore, because $H[k]$ is a constant relative to ${\mathcal
  F}_{k-1}$ as are $\ln |a|$ and $\epsilon$, we know that
\begin{align}
  \mbox{Var}[(J[k] - H[k]) | {\mathcal F}_{k-1}] &= 
  \mbox{Var}[J[k] | {\mathcal F}_{k-1}] \nonumber \\
  &= \mbox{Var}\left[ \ln | 1 + D[k] B[k] |  \big| {\mathcal
      F}_{k-1}\right] \nonumber \\
  &\leq \overline{\sigma}_{B}^2. \label{eq:martingalevariance}
\end{align}
where \eqref{eq:martingalevariance} comes from Lemma~\ref{lem:varbound}. Taking expectations on both sides gives us that
\begin{align}
  \mbox{Var}[(J[k] - H[k])] \leq  \overline{\sigma}_{B}^2.\label{eq:martingalevariance1}
\end{align}

Now, we can use $J[k]$ and $H[k] \geq 0$ to rewrite
\eqref{eq:logstate} as 
\begin{align}
  \ln |X_{a}[n]| = \ln|x_0| + \epsilon n + \sum_{k=0}^{n-1} H[k] +
  \sum_{k=0}^{n-1} (J[k] - H[k]). \label{eq:logstate2}
\end{align}
The first  term is a constant and the second and third terms are non-negative and growing at least linearly. The
martingale $\sum_{k=0}^{n-1} (J[k] - H[k])$ is the only one that could
possibly be negative. However, a simple application of Chebyshev's
inequality tells us that this is not likely.
\begin{align}
  \P\left(\sum_{k=0}^{n-1} (J[k] - H[k]\right) \leq -\frac{\epsilon n}{2})
  &\leq \P\left(|\sum_{k=0}^{n-1} (J[k] - H[k])| \geq \frac{\epsilon n}{2}\right)  \nonumber \\
  &\leq \frac{4\mbox{Var}\left[\sum_{k=0}^{n-1} (J[k] - H[k])\right]}{\epsilon^2
    n^2} \nonumber \\
  &= \frac{4   \sum_{k=0}^{n-1} \mbox{Var}\left[(J[k] - H[k]) \right]}{\epsilon^2
    n^2} \label{eq:applymartingalevariance} \\ 
  &\leq \frac{4  \overline{\sigma}_{B}^2}{\epsilon^{2}n}. \label{eq:finalchebyshevpart}
\end{align}
where \eqref{eq:applymartingalevariance} follows inductively from the
property of martingales and in particular, the uncorrelatedness of the
successive martingale difference terms and from~\eqref{eq:martingalevariance1}. The above bound in
\eqref{eq:finalchebyshevpart} goes to zero as $n \rightarrow \infty$.

Consequently,  we know that with a probability tending to $1$ that
$\sum_{i=0}^{n-1} (J[k] - H[k]) >  -\frac{n \epsilon}{2}$. Hence from
\eqref{eq:logstate2} and the fact that the $H[k] \geq 0$, we know that
\begin{align*}
  \lim_{n \rightarrow \infty} \P\left(\ln |X_{a}[n]| > \ln|x_0| +
  \frac{\epsilon}{2} n\right) &\geq
  \lim_{n \rightarrow \infty} \P\left(\ln |X_{a}[n]| > \ln|x_0| +
  \frac{\epsilon}{2} n + \sum_{k=0}^{n-1} H[k] \right) \nonumber \\
  &= \lim_{n \rightarrow \infty} \P\left(\sum_{k=0}^{n-1} (J[k] -  H[k])
  > -\frac{\epsilon n}{2}\right) \nonumber \\
  &= 1.
\end{align*}

Since the log of the absolute value of the state is growing
unboundedly with probability approaching $1$, so is the absolute
value of the state itself and the theorem is proved. 
\end{proof}

\begin{remark}
The theorem also holds trivially true in the case where $p_{B}$ has
atoms not at zero, since $C_{\sh}\left(p_{B}\right) = \infty$ in that
case. 
\end{remark}

It is interesting to recall that Burnashev's converse argument for the
reliability function of a communication channel with feedback also had
a martingale argument embedded in it to deal with the entropy
reduction in the posterior of the message \cite{burnashev1976data,
  berlin2009simple}. Here, the role of the entropy is played by the
log of the state itself. 

\section{Zero-error control capacity} \label{sec:zeroerror}
We now move to understanding control capacity for the strictest sense of stability: zero-error control capacity. Theorem~\ref{thm:debug} connects the control capacity of the actuation channel $\S$ to stability of the system $\S_{a}$. Theorem~\ref{thm:zeroerrorcalc} calculates zero-error control capacity for the system $\S(p_{B})$ with i.i.d. $B[n]$'s.

\begin{defn}
The zero-error control capacity of the system $\S$ is defined as 
\begin{align}
C_{\ze}\left(\S\right) := \max \left\{C~\bigr|~ \textrm{ exists } N \textrm{ such that for all } n> N, \underset{U[0], \cdots, U[n-1]}{\max}~\mathbb{P}\left(-\frac{1}{n} \log \frac{|X[n]|}{|X[0]|} \geq C\right) = 1 \right\}.
\end{align}
This is well defined since $X[0] \neq 0$ and $\log 0 = -\infty$ by definition.
\end{defn}

\begin{thm}\label{thm:debug}
The system $\S_{a}$ in~\eqref{eq:SAsystem} is stabilizable in a zero-error sense if and only if the zero-error control capacity of the associated system $\S$ in~\eqref{eq:Ssystem} is $C_{\ze}\left(\S\right) \geq \log |a|$.
\end{thm}

As in the case of Shannon control capacity, we use Lemma~\ref{lem:whichcontrol} to prove Thm.~\ref{thm:debug}.
\begin{proof}[Proof of Thm.~\ref{thm:debug}]
\textbf{Achievability:}

First we show that $\S_{a}$ is stabilizable in a zero-error sense if $C_{\ze}(\S) \geq \log |a|$. Hence, there exists $N$ and a causal control strategy $U[0], U[1], \cdots $ such that for all $n > N$ we have that
\begin{align}
1=&~\P\left(-\frac{1}{n} \log \frac{|X[n]|}{|X[0]|}  \geq \log|a|\right) \nonumber\\
=&~\P\left(|a|^{n} \cdot |X[n]|  \leq  x_{0}\right). \label{eq:acze}
\end{align}
Choose $U_{a}[k] = a^{k} U[k]$. Then, Lemma~\ref{lem:whichcontrol} implies that $|X_{a}[n]| = |a|^{n} \cdot |X[n]|$, which combined with~\eqref{eq:acze} completes the proof.

\textbf{Converse:}
To show necessity, we assume that $\S_{a}$ is zero-error stabilizable. Let $U_{a}[0], \cdots, U_{a}[n]$ be a causal control law that stabilizes the system. There exist $N$ and $M$ for $n > N$, such that:
\begin{align*}
\P\left(|X_{a}[n]| < M\right) = 1.
\end{align*}
Choosing $U[k] = a^{-k} U_{a}[k]$, and using Lemma~\ref{lem:whichcontrol} we have that $|X_{a}[n]| = |a|^{n} \cdot |X[n]|$, and hence  
\begin{align*}
\P\left(|X[n]|  < |a|^{-n} M\right) = 1.
\end{align*}
Dividing by $|X[0]| = x_{0} \neq 0$ and rearranging this gives:
\begin{align*}
\P\left(-\frac{1}{n} \log \frac{|X[n]|}{|X[0]|}  > \log|a| - \frac{1}{n}\frac{M}{|X[0]|} \right) = 1.
\end{align*}
As $n\to\infty$, we see that the lower bound on $-\frac{1}{n} \log \tfrac{|X[n]|}{|X[0]|}$ will get arbitrarily close to $\log|a|$, which gives that:
\begin{align*}
C_{\ze}(\S) \geq \log|a|. 
\end{align*}\qedhere
\end{proof}

The operational definition for zero-error control capacity involves an optimization over an infinite sequence of potential control laws. In this section, we show that in fact this quantity can be easily computed for the system $\S(p_{B})$ where the $B[n]$'s are distributed i.i.d.~with distribution $p_{B}$.  

We first focus on the case where $p_{B}$ has bounded support. We will later show (Theorem~\ref{thm:zeroerrorzero}) that when the support of $p_{B}$ is unbounded, the zero-error control capacity is zero. 

\begin{thm}
Consider the system $\S(p_{B})$ in~\eqref{eq:Ssystem} with $B[n]$'s
drawn from $p_{B_n}$, each with essential support on $[b_{1}, b_{2}]$. 
Then the zero-error control capacity of the system is
\begin{align}
C_{\ze}(\mathcal{S}(p_{B})) &= - \log \min_d \max_{b \in [b_1,b_2]} |1+b\cdot d| \label{eq:optimizationcharacterizationzeroerror}\\
&= 
\begin{cases}
\log \frac{|b_{2}-b_{1}|}{|b_{2}+b_{1}|}~&\textrm{if}~~0 \not\in [b_{1}, b_{2}],\\\
0~&\textrm{if}~~0 \in [b_{1},b_{2}]. 
\end{cases} \label{eq:explicitzeroerror}
\end{align}
\label{thm:zeroerrorcalc}
\end{thm}

As in the Shannon control capacity case, the proof relies on showing that a simple greedy strategy is optimal. This is captured in the following lemma.

\begin{restatable}{lem}{lemzeroerroronestep}
Suppose the system state at time $n$ is $x[n] \in \mathbb{R}, x[n]
\neq 0$. For the system $\S(p_{B})$ in Thm.~\ref{thm:zeroerrorcalc},
define the one-step zero-error control capacity as 
\begin{align*}
C_{1,\ze}(x[n]) = \sup \left\{C ~\biggr|~ \underset{U[n]}{\max}~P\left(-\log \tfrac{|X[n+1]|}{|x[n]|} \geq C\right) = 1\right\}.
\end{align*}
Then, $C_{1,\ze}(x[n])$ does not depend on $x[n]$, and is given by $-
\log \min_d \max_{b \in [b_1,b_2]} |1+b\cdot d|$ which simplifies to
\begin{align*}
C_{1,\ze}(x[n]) =
\begin{cases}
\log \frac{|b_{2}+b_{1}|}{|b_{2}-b_{1}|}~\textrm{if}~~0 \not\in [b_{1}, b_{2}]\\
0~\textrm{if}~0 \in [b_{1},b_{2}]. 
\end{cases}
\end{align*}
\label{lem:zeroerroronestep}
\end{restatable}

The achievability part of this lemma implies that there exists a linear memoryless stationary control law such that 
\begin{align*}
P\left(-\log \frac{|X[n+1]|}{|x[n]|} \geq C_{\ze}\right) = 1.
\end{align*}
The converse implies that for any $\delta > 0$, for all $U[n]$, 
\begin{align*}
P\left(-\log \frac{|X[n+1]|}{|x[n]|} > C_{\ze} + \delta \right) < 1.
\end{align*}

As was the case for Shannon control capacity, the key observation in proving Lemma~\ref{lem:zeroerroronestep} is to notice that $\tfrac{X[n+1]}{x[n]} = 1 + B[n] \tfrac{U[n]}{x[n]}.$
Then, we observe that $U[n]$ is a function of $x[n]$, and what remains to be understood is the quantity:
$$\underset{d \in \mathbb{R}}{\min}~\underset{B \in [b_{1}, b_{2}]}{\max} \bigr| 1 + B\cdot d \bigr|,$$
where $B\sim p_{B}$. The full proof is deferred to Appendix~\ref{sec:appendixzeroerror}.

\begin{proof}[Proof of Thm.~\ref{thm:zeroerrorcalc}]

\noindent \textbf{Achievability:}
First, note that if for any $k$, $X[k]=0$ for $0 \leq k \leq n$ we are done, since with $X[k]=0$ we can guarantee $X[k+1]=0$ with the choice $U[k] = 0$. (Here, we define $\frac{0}{0}=1$.) Consequently, we focus on the case where $X[k] \neq 0$ for any $0 \leq k \leq n$. Now consider:
\begin{align}\label{eq:probintersection}
\P\left(-\tfrac{1}{n} \log \frac{|X[n]|}{|X[0]|}  \geq C_{\ze}\right) 
=&~\P \left(\sum_{k=0}^{n-1} - \log \frac{|X[k+1]|}{|X[k]|}  \geq n \cdot C_{\ze}\right),\\ 
\geq&~\P\left(\left\{-\log \frac{|X[1]|}{|X[0]|}  \geq C_{\ze}\right\},  \cdots , \left\{-\log \frac{|X[n]|}{|X[n-1]|}  > C_{\ze}\right\}\right).\label{eq:pi2}
\end{align}
Hence, we can focus on each of the events $\left\{-\log \frac{|X[k+1]|}{|X[k]|}  \geq C_{\ze}\right\}$ to give us a bound on the probability on the LHS of~\eqref{eq:probintersection}. For every realization $x[k]$ of $X[k]$, Lemma~\ref{lem:zeroerroronestep} provides a $U[k]$ that ensures that 

$$\P\left(-\log \frac{|X[k+1]|}{|X[k]|} \geq C_{\ze}\mid X[k] = x[k]\right) = 1.$$ 

Since the controller has access to a perfect observation $Y[k] = X[k]$
to generate $U[k]$, we can guarantee that the unconditional
probability, $\P\left(-\log \tfrac{|X[k+1]|}{|X[k]|} \geq
  C_{\ze}\right)$ is also equal to $1$. Hence, the controller can
causally generate a sequence of controls $U[0]$ to $U[n]$ such that the probability in~\eqref{eq:pi2}  is equal to $1$, which completes the proof.

\noindent \textbf{Converse:}
To prove the converse, we would like to show that for all $\delta > 0$,
$\P\left(-\frac{1}{n} \log \frac{|X[n]|}{|X[0]|}  > C_{\ze} + \delta \right) < 1.$
This is equivalent to showing that for all sequences of applied controls up to time $n-1$, 
\begin{align}
\P\left(-\frac{1}{n} \log \frac{|X[n]|}{|X[0]|}  \leq C_{\ze} + \delta \right) > 0. \label{eq:con1}
\end{align}
We will use induction to show this. $n = 1$ is the base case. We know from Lemma~\ref{lem:zeroerroronestep} that~\eqref{eq:con1} is true for $n=1.$ Now,
\begin{align}
\log \frac{|X[n]|}{|X[0]|} =\log \frac{|X[n-1]|}{|X[0]|} +  \log \frac{|X[n]|}{|X[n-1]|}. \label{eq:con2}
\end{align}
Let $\G = \left\{ -\frac{1}{n-1}\log \frac{|X[n-1]|}{|X[0]|} \leq C_{\ze} +\delta\right\}$, and the event $\H = \left\{-\log \frac{|X[n]|}{|X[n-1]|} \leq C_{\ze} +\delta\right\}$. 
Then, using~\eqref{eq:con2}, we can lower bound the probability in~\eqref{eq:con1} by $\P(\G \cap \H)$. The induction hypothesis implies that $\P(\G) > 0$.

Now, consider $\P(\H \mid \G)$. If the event $\G$ occurred then we can infer that $X[n-1] \neq 0$, since $-\frac{1}{n-1} \log \frac{|X[n-1]|}{|X[0]|} = \infty$ if $X[n-1]=0$. Hence, we can apply Lemma~\ref{lem:zeroerroronestep} to get that $\P(\H \mid \G) > 0$. Thus, $\P(\G \cap \H) = \P(\H \mid \G) \P(\G) > 0$, which completes the proof. 

\end{proof}

\begin{remark} As in the case of Shannon control capacity, this proof
  reveals that there is nothing lost by going to linear memoryless
  strategies for the zero-error control capacity problem. Furthermore,
  this immediately generalizes to the non-stationary case as well. 
\end{remark}

Finally, the last theorem of this section shows that when $p_{B}$ is unbounded, the zero-error control capacity is zero.
\begin{thm}
The zero-error control capacity of the system $\S$ in~\eqref{eq:Ssystem} with $B[i]$ distributed i.i.d according to $p_{B}$, where $p_{B}$ has unbounded support, is zero.
\label{thm:zeroerrorzero}
\end{thm}

\begin{proof}
First, we note that a trivial ``do nothing'' strategy is useless since if $U[i] = 0$ for $0\leq i\leq n$, then $X[0] = X[n]$. 

Now, consider any strategy $U[0], \cdots, U[n-1]$ with not all
controls zero. Say $U[n-1] = u \neq 0$ and without loss of generality
suppose $u > 0$. Then, for every value of $X[n]$ and chosen $U[n]\neq
0$, whatever bound $M < \infty$ we may try to claim, $\P\left(|X[n] + B[n]
  u| > M\right) = \P\left(B[n] \in \left(-\infty, \frac{-M -
      X[n]}{u}\right) \bigcup \left(\frac{M -
      X[n]}{u},+\infty\right)\right) > 0$. The same argument works for
$u < 0$. 
\end{proof}

\section{$\eta$-th moment control capacity} \label{sec:eta}
Finally, we consider $\eta$-th moment stability. Theorem~\ref{thm:etareal} shows that $\S_{a}$ is stabilizable in the $\eta$-th moment sense if and only if the $\eta$-th moment control capacity of $\S$ is greater than $\log|a|$. Theorem~\ref{thm:calcetacap} shows how to single-letterize the expression for $\eta$-th  moment control capacity. As $\eta$ ranges from $0$ to $\infty$ it captures a range of stabilities from the weaker ``Shannon'' sense as $\eta \rightarrow 0$, to the zero-error notions of stability as $\eta \rightarrow \infty$ --- this is justified by Theorem~\ref{thm:limits}.

\begin{defn}
The $\eta$-th moment control capacity of the system $\S$ as in~\eqref{eq:Ssystem} is defined as
\begin{align}
C_{\eta}(\mathcal{S}) = \liminf_{n\rightarrow\infty}~~\underset{U[0], \cdots, U[n-1]}{\max}-\frac{1}{n}~\frac{1}{\eta}~\log\bigE\left[\biggr|\frac{X[n]}{X[0]}\biggr|^{\eta}\Biggr.\right].
\end{align}
\end{defn}

\begin{thm}
Consider the system $\S_{a}$ from~\eqref{eq:SAsystem}, and the
associated system $\S$ as in~\eqref{eq:Ssystem}. $\S_{a}$ is $\eta$-th
moment stabilizable if $C_{\eta}\left(\mathcal{S}\right) > \log
|a|$. Conversely, if the system $\S_{a}$ is $\eta$-th moment
stabilizable, then $C_{\eta}\left(\mathcal{S}\right) \geq \log
|a|$.
\label{thm:etareal}
\end{thm}

\begin{proof}[Proof of Thm.~\ref{thm:etareal}]
\textbf{Achievability}:
Since $C_{\eta}\left(\mathcal{S}\right) > \log |a|$, we know that for the system $\S$ there exists a control strategy $U[0], U[1], \cdots$ and $N < \infty$ such that for $n>N$:
\begin{align*}
\log |a| \leq -\frac{1}{n}~\frac{1}{\eta}~\log\bigE\left[\biggr|\frac{X[n]}{X[0]}\biggr|^{\eta}\Biggr.\right] 
&= -\frac{1}{n}~\frac{1}{\eta}~\log\bigE\left[\biggr|\frac{a^{-n}X_{a}[n]}{X[0]}\biggr|^{\eta}\Biggr.\right].
\end{align*}
The equality follows from Lemma~\ref{lem:whichcontrol}, using $U_{a}[k] = a^{k}U[k]$. This can be rewritten as:
\begin{align*}
\bigE\left[\biggr|\frac{X_{a}[n]}{X[0]}\biggr|^{\eta}\Biggr.\right] \leq 1,
\end{align*}
which gives the required bound $\bigE\left[|X_{a}[n]|^{\eta}\bigr.\right] \leq x_{0}^{\eta}$ to show that $\S_{a}$ is $\eta$-th moment stabilizable.

\textbf{Converse:}
There exists a control strategy $U_{a}[0], U_{a}[1], \cdots $ and $N, M < \infty$ such that for $\forall n>N$, we have that
$\bigE[|X_{a}[n]|^{\eta}] < M$. We rewrite this using Lemma~\ref{lem:whichcontrol} and divide by $|X[0]| = x_{0}$, to get:
\[
\bigE\left[\biggr|\frac{a^{n}X[n]}{X[0]}\biggr|^{\eta}\Biggr.\right] < \frac{M}{x_{0}^{\eta}}.
\]
This can be further manipulated to give:
\[
-\frac{1}{n} \frac{1}{\eta} \log \bigE \left[ \biggr|\frac{X[n]}{X[0]}\biggr|^{\eta}\Biggr.\right] > \log |a| - \frac{1}{n} \frac{1}{\eta} \log \frac{M}{x_{0}^{\eta}}.
\]
Hence, 
\[
\liminf_{n\to\infty} -\frac{1}{n} \frac{1}{\eta} \log \bigE \left[ \biggr|\frac{X[n]}{X[0]}\biggr|^{\eta}\Biggr.\right] \geq \log |a|,
\]
which gives the desired result.
\end{proof}

The next theorem shows how to calculate the $\eta$-th moment control capacity in the case of i.i.d. $B[n]$'s.
\begin{thm}
The $\eta$-th moment control capacity of the system $\S(p_{B})$, parameterized with a single distribution $p_{B}$ with $B[n]$s are i.i.d.~is given by:
\begin{align}
C_{\eta}(\S(p_{B})) =
  \underset{d\in\mathbb{R}}{\max}~-\frac{1}{\eta}\log\bigE\bigr[|1 + B
  \cdot d|^{\eta}\bigr] \label{eq:computingetacapacity}
\end{align}
\noindent where $B \sim p_{B}$. 
\label{thm:calcetacap}
\end{thm}

The proof uses a one-step lemma, just as in the Shannon and zero-error cases. 

\begin{lem}
Suppose the system state at time $n$ is $x[n] \in \mathbb{R}, x[n]
\neq 0$. Then, for system $\S$ in~\eqref{eq:Ssystem} define the
one-step $\eta$-control capacity as 
\begin{align*}
C_{1,\eta}(x[n]) =\underset{U[n]}{\max}~-\frac{1}{\eta}\log\bigE\Biggr[ \biggr|\frac{X[n+1]}{x[n]}\biggr|^{\eta}\Biggr].
\end{align*}
where $U[n]$ is any function of $x[n]$.
\noindent Here, the expectation is taken over the random realization of $B[n] \sim p_{B}$ i.i.d., and $X[n] =x[n] \neq 0$ has been realized and fixed. Then, $C_{1,\eta}(x[n])$ does not depend on $x[n]$ or $n$, and is given by
\begin{align*}
C_{1,\eta}(x[n]) = C_{1,\eta} =  \underset{d\in\mathbb{R}}{\max}~-\frac{1}{\eta} \log \bigE \bigr[|1 + B[n] \cdot d|^{\eta} \bigr],
\end{align*}
where $B[n] \sim p_{B}$.
Hence, there exists a constant $d$ such that if $U[n] = d \cdot X[n]$, then  
\begin{align*}
- \log \bigE\Biggr[\biggr|\frac{X[n+1]}{x[n]}\biggr|^{\eta}\Biggr] = C_{1,\eta}. 
\end{align*}
\label{lem:etalem}
\end{lem}
\begin{proof}
This proof is essentially identical to the proof of Lemma~\ref{lem:shcclem} in the Shannon case and is omitted here. 
\end{proof}

\begin{proof}[Proof of Thm.~\ref{thm:calcetacap}]
\textbf{Achievability:}
The argument here is very similar to that of the Shannon case. We use the linear law from Lemma~\ref{lem:etalem} and take advantage of the independence of the $B[k]$'s to turn an expectation of a product into a product of expectations. Since each of the terms is identical, the result follows. 

\textbf{Converse:}
We will use induction to prove the converse. The base case for $n=1$ follows immediately. 

For $n>1$, if $X[n] =0$ then by choosing $U[n] = 0$, the controller can ensure the minimum possible $|X[n+1]| = 0$, which is the optimal action. Consequently, we restrict attention to control strategies that apply a $0$ control when faced with a $0$ state. These are sample-path by sample-path as good as or better than other strategies when it comes to minimizing $|X[n]|$. 

Let $Q_n := 0$ if $X[n-1] = 0$ and $Q_n := \frac{X[n]}{X[n-1]}$ otherwise. If $X[n-1] \neq 0$, then we can write 

\begin{equation}
\left| \frac{X[n]}{x_{0}}\right|^{\eta} = \left|\frac{X[n]}{X[n-1]}\right|^{\eta} \left| \frac{X[n-1]}{x_{0}}\right|^{\eta}. \label{eq:xnozero}
\end{equation}
Using~\eqref{eq:xnozero} and the definition of $Q_{n}$ we can write:
\[
\bigE \left[\left| \frac{X[n]}{x_{0}}\right|^{\eta} \biggr. \right] = \bigE\left[\left|Q_n\right|^{\eta} \left| \frac{X[n-1]}{x_{0}}\right|^{\eta} \biggr. \right].
\]

Ideally we would like to separate the two terms above to use
induction, but since the terms are not necessarily independent, this
is not directly possible. $B[n-1]$ is the new term in $Q_{n}$, and we condition on $B[0]$ to $B[n-2]$ to focus on this.  We have that
\begin{align}
\bigE \Biggr[\bigE\left[\left|Q_n \right|^{\eta} ~ \left|\frac{X[n-1]}{x_{0}}\right|^{\eta}~\biggr|~B[0], \ldots, B[n-2] \right] \Biggr] 
= \bigE \Biggr[ \left|\frac{X[n-1]}{x_{0}}\right|^{\eta}~ \bigE\left[\left|Q_n \right|^{\eta} ~\biggr|~B[0], \ldots, B[n-2] \right]\Biggr]. \label{eq:induc2}
\end{align}
The equality in~\eqref{eq:induc2} follows because $\tfrac{X[n-1]}{x_0}$ is a constant when conditioned on $B[0], \ldots, B[n-2]$.  

Let $Q := \bigE\left[|Q_n|^\eta \bigr| B[0], \ldots, B[n-2]\right]$ and $R := \left|\frac{X[n-1]}{x_0}\right|^{\eta}$. If $X[n-1] = 0$, then $Q = 0$ and $R = 0$.

If $X[n-1] = x[n-1]\neq 0$, then by Lemma~\ref{lem:etalem} we have that
\[
\min_{U[n-1]} \bigE\left[\left| \frac{X[n]}{x[n-1]}\right|^\eta\right] = \min_{d} \bigE\left[\left|1  + B[n-1]d\right|^\eta\right]. 
\]
Let $Q^{'} := \min_{d} \bigE\left[\left|1  + B[n-1]d\right|^\eta\right].$ Hence $Q \geq Q^{'}$ if $X[n-1] \neq 0$, i.e., whenever $Q$ and $R$ are non-zero themselves. Hence, $ R Q \geq R Q'$ whenever $R \neq 0$.
Thus~\eqref{eq:induc2} is  lower-bounded by 
\begin{align}
\bigE \Biggr[ \left|\frac{X[n-1]}{x_{0}}\right|^{\eta}~ \min_{d} \bigE\left[\left|1  + B[n-1]d\right|^\eta\right]\Biggr]
= \left(\min_{d} \bigE\left[\left|1  + B[n-1]d\right|^\eta\right]\right) ~\bigE \Biggr[ \left|\frac{X[n-1]}{x_{0}}\right|^{\eta}\Biggr].\label{eq:induc3}
\end{align}

From \eqref{eq:induc3}, induction and Lemma~\ref{lem:etalem} give the desired converse. 
\end{proof}

Once again we see that there is no loss of optimality in restricting to linear memoryless stationary strategies for the purposes of calculating the $\eta$-th moment control capacity of the system $\S(p_{B})$.

\begin{cor}
Consider the system $\S(p_{B})$ from~\eqref{eq:Ssystem}, with $B[n] \sim p_{B}$ i.i.d.~with mean $\mu_{B}$ and variance $\sigma_{B}^{2}$. Then
\begin{align}
C_{2}(\S(p_{B})) = \frac{1}{2} \log \left(1 + \frac{\mu_{B}^{2}}{\sigma_{B}^{2}} \right).
\end{align} \label{cor:meansq}
\end{cor}

\begin{proof}
We know that
\begin{align*}
C_{2}(\S(p_{B})) = \underset{d\in\mathbb{R}}{\max}~ -\frac{1}{2} \log\bigE\left[|1 + B \cdot d|^{2}\right]
\end{align*}

We can compute the optimum $d = -\frac{\mu_{B}}{\mu_{B}^{2} +
  \sigma_{B}^{2}}$ by taking derivatives. 
                               
Substituting $d$ back into the equation gives the desired result. This recovers the second-moment result that was known from~\cite{uncertaintyThreshold}.  \end{proof}

\begin{remark}
This optimality of linear strategies for the second moment case was seen in the classical uncertainty threshold principle, but we can now say they are optimal from a control capacity perspective for all moments. 
\end{remark}

\begin{remark}
The similarity of this expression to the $\frac{1}{2}\log(1 + SNR)$ formula for communication capacity in the AWGN case is notable. Previous work by Elia~\cite{eliaFading} as well as Martins and co-authors~\cite{martinsUncertain,martins2007fundamental} noticed similar patterns in related systems. Our goal here is to develop a unifying theory for these observations. \end{remark}

The last theorem of this section shows how the $\eta$-th moment control capacity is really the broadest sense of capacity. As $\eta\to0$ this tends to the Shannon notion of capacity and as $\eta\to\infty$ this tends to the zero error notion of capacity.

\begin{thm}
\label{thm:limits}
Consider the system $\S$ in~\eqref{eq:Ssystem}. If the $B[k]$ are i.i.d~continuous random variables with no
atoms, except possibly at $B=0$, with a bounded density $p_{B}$ such that there exist $ \gamma, \xi > 1$ so that
$p_B(b) \leq \gamma \min\left(1, \frac{\xi}{|b|}\right)$, and have finite first and second moments $\bigE[|B|]$ and $\bigE[|B|^2]$, then if $\lim_{\eta \rightarrow 0} C_{\eta}\left(\S(p_{B})\right)$ exists and is finite, $\lim_{\eta \rightarrow 0} C_{\eta}(\S) = C_{\sh}\left(\S(p_{B})\right).$ If $\lim_{\eta \rightarrow 0} C_{\eta}(\S) = \infty$, then $C_{\sh}(\S) = \infty$ as well.

Similarly, if $B[k]$ are i.i.d.~with essential support on $[b_1,b_2]$ and $C_{\ze}\left(\S(p_{B}) \right) > 0$, then $\lim_{\eta \rightarrow \infty} C_{\eta}\left(\S(p_{B})\right) = C_{\ze}\left(\S(p_{B})\right)$.  If $\lim_{\eta \rightarrow \infty} C_{\eta}\left(\S(p_{B})\right) = 0$, then $C_{\ze}\left(\S(p_{B})\right) = 0$ as well. 
\end{thm}
 
\begin{proof}

    The case $\lim_{\eta \rightarrow 0} C_{\eta}(\S) = \infty$ follows
    immediately from the operational meaning of the control
    capacities. Since the Shannon sense is weaker than any $\eta$-th
    moment sense of stability, the Shannon control capacity will at least be $\lim_{\eta \rightarrow 0} C_{\eta}(\S)$. Similarly, the case $\lim_{\eta
      \rightarrow \infty} C_{\eta}(\S) = 0$, also follows since the zero-error sense is operationally stronger than any $\eta$-th moment sense of stability.

For $\eta \to \infty$ we see that the inner expectation in $- \log \sqrt[\eta]{\bigE \left[ |1 + Bd|^\eta \right]}$ will be dominated by the maximum possible values that $|1 + Bd|$ can take. This means that
\begin{align*} 
\lim_{\eta \rightarrow \infty} g_{\eta}(u)  &= - \lim_{\eta \rightarrow \infty} \log \sqrt[\eta]{\max_{b \in [b_1,b_2]} |1 + bd|^\eta } \\
&= - \log \max_{b \in [b_1,b_2]} |1 + bd|,
\end{align*}
which agrees with the optimization characterization \eqref{eq:optimizationcharacterizationzeroerror} of the zero-error control capacity, proving the result.

The nontrivial case is $\eta \rightarrow 0$ when the limits are finite. For convenience of differentiating, in this section we will use nats (base e) instead of bits (base 2) to measure control information.

Let $Z_{d} := \ln |1 + B\cdot d|$ be a family of random variables parametrized by the scalar $d$. The dependence on the random variable $B$ is suppressed for notational convenience. 

Consider $Z_{d}$'s log-moment-generating function $F(d,s) = \ln \bigE
\left[ e^{Z_{d} s}  \right] = \ln \bigE \left[ |1 + B d|^{s}
  \right]$. By the standard properties of log-moment-generating
functions, this is a convex function of $s$. Since the log-moment-generating function is also the cumulant generating
function, the first two terms of the Taylor expansion around $s=0$ are given by the first two cumulants of $Z_{d}$, the
mean and the variance. Thus 
\begin{align}
  F(d,s) &= s \bigE[Z_{d}] + \frac{s^2}{2} \mbox{Var}[Z_{d}] + \cdots
  \label{eq:taylorexpand}
\end{align}

Now, consider the expression $G(d,s) = \frac{-F(d,s)}{s}$. From \eqref{eq:taylorexpand}, we know 
$$G(d,s) = - \bigE[Z_{d}] - \frac{s}{2} \mbox{Var}[Z_{d}] + \cdots.$$

Computing the derivative of $G(d,s)$ we have that:
  \begin{align}
    \frac{\partial}{\partial s} G(d,s) &=
    \frac{(\ln \bigE \left[ |1 + Bd|^s \right])
      \bigE\left[|1+Bd|^s\right]
      -s \bigE\left[|1 + Bd|^s \ln|1+Bd|\right]}
         {s^2 \bigE\left[|1+Bd|^s\right]} \label{eq:derivativeOfD}
  \end{align}
We note that $f(x) = x \ln (x)$ is convex, and hence Jensen's inequality implies that $f\left( \bigE[x] \right) \leq \bigE\left[ f(x)\right].$ Choosing $x = |1 + Bd|$ gives us that the numerator in~\eqref{eq:derivativeOfD} must be $\leq 0$, and hence $\frac{\partial}{\partial s} G(d,s) \leq 0$, which implies that $G(d,s)$ must be a decreasing function of $s$. 

Now, from \eqref{eq:computingetacapacity} we know that the control capacity $C_\eta(\S) = G(d^*_{\eta},\eta)$ where $d^*_{\eta} := \argmax_d G(d,\eta)$. 
Furthermore, looking at the Shannon control capacity expression in \eqref{eq:computingShannonCapacity}, we can define
$J(d) := -\bigE \left[ Z_{d} \right] = G(d,0)$. Let $\argmax_{d} J(d) := d^*_{\sh}$, then we have $J(d^{*}_{\sh}) = C_{\sh}(\S)$. Finally, since $G(d,s)$ is a decreasing function of $s$, $J(d)$ is an upper bound on $G(d,s)$ for all $s\geq 0$.

We note now that 
\begin{align}
\lim_{\eta \rightarrow 0} C_{\eta}(\S) = \lim_{\eta \rightarrow 0} \sup_d G(d,\eta). \label{eq:interchangegoal1}
\end{align}

Since, $\sup_d  \lim_{\eta \rightarrow 0} G(d,\eta) = \sup_{d} J(d) = C_{\sh}(\S)$, we are done if we can show that the limit and sup can be interchanged in~\eqref{eq:interchangegoal1}. To show this,
  we need to establish that $G(d,\eta)$ is uniformly continuous in
  $\eta$. However, this is not clear when $d$ is unbounded. Hence, we
  will show that the optimizing $d^{*}_{\eta}$ and $d^{*}_{\sh}$ are
  attained in a bounded interval around $0$, and establish uniform continuity in that interval.

We show that there is a bound $d_m(B)$
  that depends on the distribution of $B$ for which we know that
  $-d_m(B) \leq d^*(\eta) \leq +d_m(B)$ and also $-d_m(B) \leq
  d^*_{\sh} \leq +d_m(B)$. Note that $G(0,\eta) = J(0) = 0$. Hence it suffices to show that the functions $G(d,\eta)$ and $J(d)$ are both
  bounded above by zero outside this interval for $d$. 
  
  Since $J(d) > G(d,\eta)$ for all $\eta > 0$, we focus on $J(d)$ and
  show that it is bounded above by $0$ outside an interval $[-d_m(B), +d_m(B)]$. The
  function $L_{k}(x)$, parametrized by $k > 1$ serves as the requisite
  lower bound on the natural logarithm function $\ln$ for this
  purpose. For $x> 0$:  
  \begin{align}
    L_k(x) &:= \begin{cases}
      \ln(x) & \mbox{if }x < 1 \\
      0 & \mbox{if }1 \leq x < k \\
      \ln k & \mbox{if }x \geq k
    \end{cases}. \label{eq:logboundingfunction}
  \end{align}
  Because $\ln$ is an increasing function, $L_k(x)$ as defined in
  \eqref{eq:logboundingfunction} is valid lower-bound to $\ln(x)$ by inspection. 

  Applying this lower bound by expanding the definitions of $J(d)$ and
  $Z_{d}$, we know that 
  \begin{align*}
    J(d) = -\bigE \left[ \ln |1 + Bd| \right]
    \leq \bigE \left[-L_k(|1 + Bd|) \right].
  \end{align*}

  We split the random variable $-L_k(|1 + Bd|)$ into its positive
  component $J^+ \geq 0$ and its negative component $J^- < 0$ so
  that $J^+ + J^- = -L_k\left(|1 + Bd|\right)$ and $J^{+} J^{-} = 0$ ensuring that
  both cannot simultaneously be nonzero. 
  
  First we upper bound  the expectation 
  $$\bigE[J^+] = \int_0^\infty \P(J^+ > t)~dt.$$ 
  As in Lemma~\ref{lem:varbound} let $Z_{d} = \ln |1 + B d|$. Since $L_k(\cdot) = \ln(\cdot)$ in the interval
  $(0,1)$, we know that $J^+ = -Z_{d}^{-}$ from Lemma~\ref{lem:varbound}. Applying the lemma, we know that for all $t > 0$,
   $$\P(J^+ > t) = \P(Z_{d} < -t) \leq K_{B} e^{-t}.$$ 
  Integrating this gives us that for all $d$:
  \begin{align}
    \bigE[J^+] \leq K_{B}.
    \label{eq:boundonpositivepart}
  \end{align}

Now we must bound $\bigE[J^{-}]$. From~\eqref{eq:logboundingfunction} we have that $L_{k}(|1 + Bd|)$ is an
indicator random variable that equals $-\ln k$ whenever $|1 + Bd| \geq e^k$. Consequently we know that 
  $$\bigE[ J^-] \leq -\left(\ln k\right) \P\left( |1 + Bd| \geq e^k\right).$$ 
  We know by an argument parallel to the one for Lemma~\ref{lem:varbound} that
  since $\gamma$ bounds the density, \mbox{$\P\left(|1 + Bd| \leq  e^k\right) \leq \frac{2\gamma}{|d|}e^{k}$} and so 
$\P\left( |1 + Bd| \geq e^k\right) \geq 1 - \frac{2\gamma}{|d|}e^{k}.$  
Hence we have that
$$\bigE[ J^-] \leq -\left(\ln k\right) \left(1 - \frac{2\gamma}{|d|}e^{k}\right).$$
  
  Choose parameter $k = e^{2 K_{B}}$ so $-\ln k = -2 K_B$. Then, let $d_m = 4 \gamma
  e^{k}$ be our bound on $|d|$ so that $\bigE[ J^-] \leq -K_{B}$. Combining this with 
  \eqref{eq:boundonpositivepart} tells us that for all $|d| > d_m$,
  the function $J(d)$ is upper-bounded by zero. Since $J(0) = 0$, we focus our attention on the interval $[-d_m, +d_m]$.

Now it remains to show show uniform continuity of $G(d,s)$ within the interval $[-d_m, +d_m]$ for $s$ in the neighborhood of $s=0$. For this, we notice that $G(d,s)$ is monotone in $s$, and both $G(d,s)$ and $J(d)$ are continuous in $d$. Hence, by Dini's theorem, the convergence must be uniform, which completes the proof.

\end{proof}

\section{Computing Control Capacity} \label{sec:discussion}

\begin{figure}[htbp]
\begin{center}
\includegraphics[width = 0.55\textwidth]{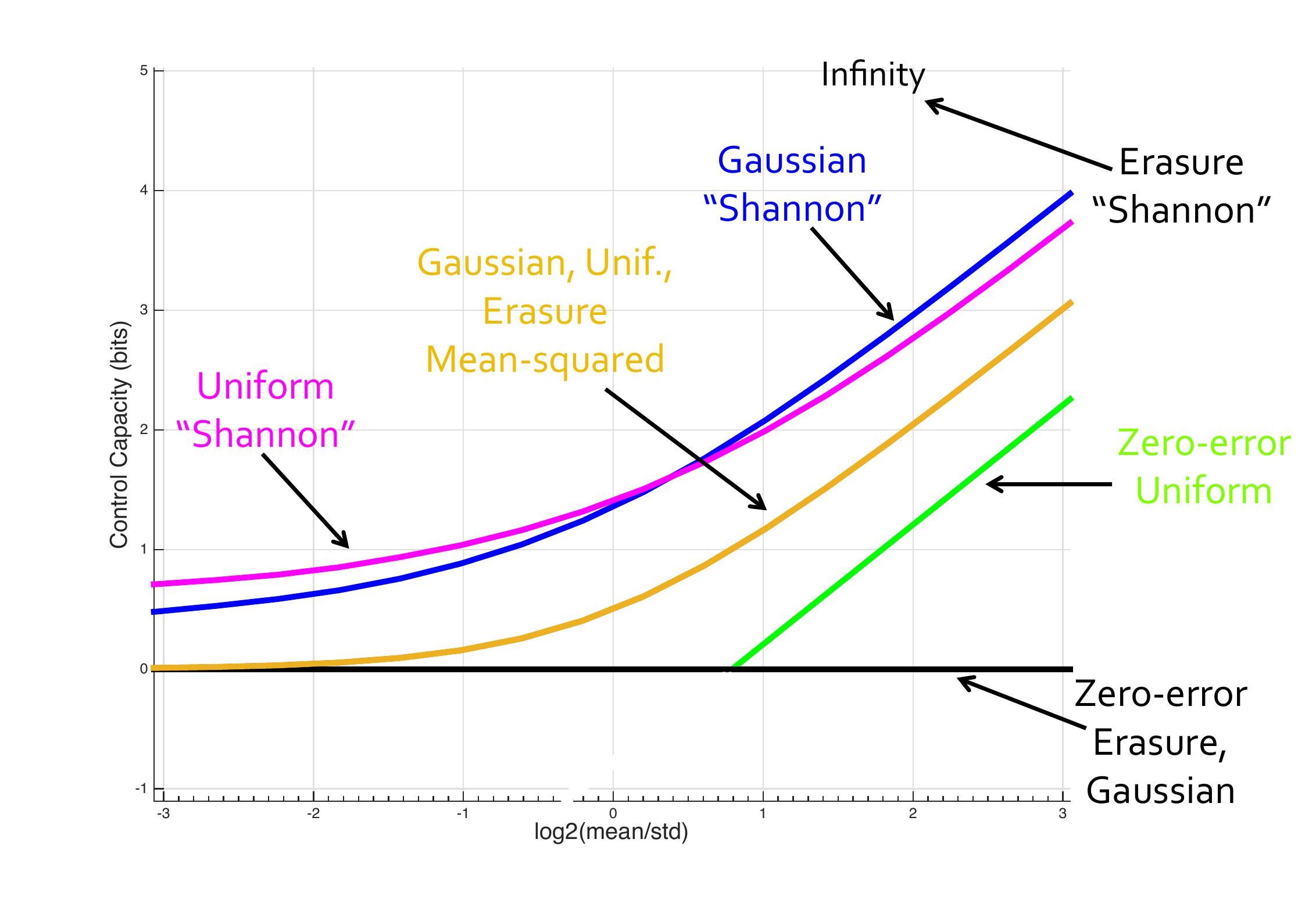}
\caption{Examples of control capacity as we vary different
  distributions for $B$ parametrized by the ratio of mean to standard deviation.}
\label{fig:cc1}
\end{center}
\end{figure}

Fig.~\ref{fig:cc1} plots the zero-error, Shannon and second-moment
control capacities for an actuation channel with a $B$ (the multiplicative noise) having a Gaussian distribution, a Bernoulli-$(p)$
distribution (erasure channel) and a Uniform distribution. These
distributions are normalized so that they all have the same ratio of
the mean to the standard deviation. The x-axis is the log of this
ratio, which is all that matters for the second-moment control capacity as seen in
Corollary~\ref{cor:meansq}. Consequently, the second moment control capacities for all three
distributions line up exactly. We see that the Shannon sense control capacity for both the
Gaussian and the Uniform are larger than the second-moment capacity as
expected. The Shannon capacity for the Bernoulli actuation channel is
infinity since it has an atom at $1$, while the zero-error capacity is
zero because it has an atom at $0$. The zero-error capacity for the Gaussian channel is zero because it is unbounded. The Uniform distribution follows the zero-error capacity line for bounded distributions, and has slope $1$. 

Notice that as the ratio of the mean to the standard deviation goes to infinity, all of the lines approach slope 1. We conjecture that in this ``high SNR'' regime, this ratio is essentially what dictates the scaling of control capacity. This is predicted by the carry-free models discussed in the Appendix since the capacity in both the zero-error and Shannon senses depends only on the number of deterministic bits in the control channel gain $g_\det-g_\ran$.

\begin{figure}[hbtp]
\begin{center}
\includegraphics[width = 0.5\textwidth]{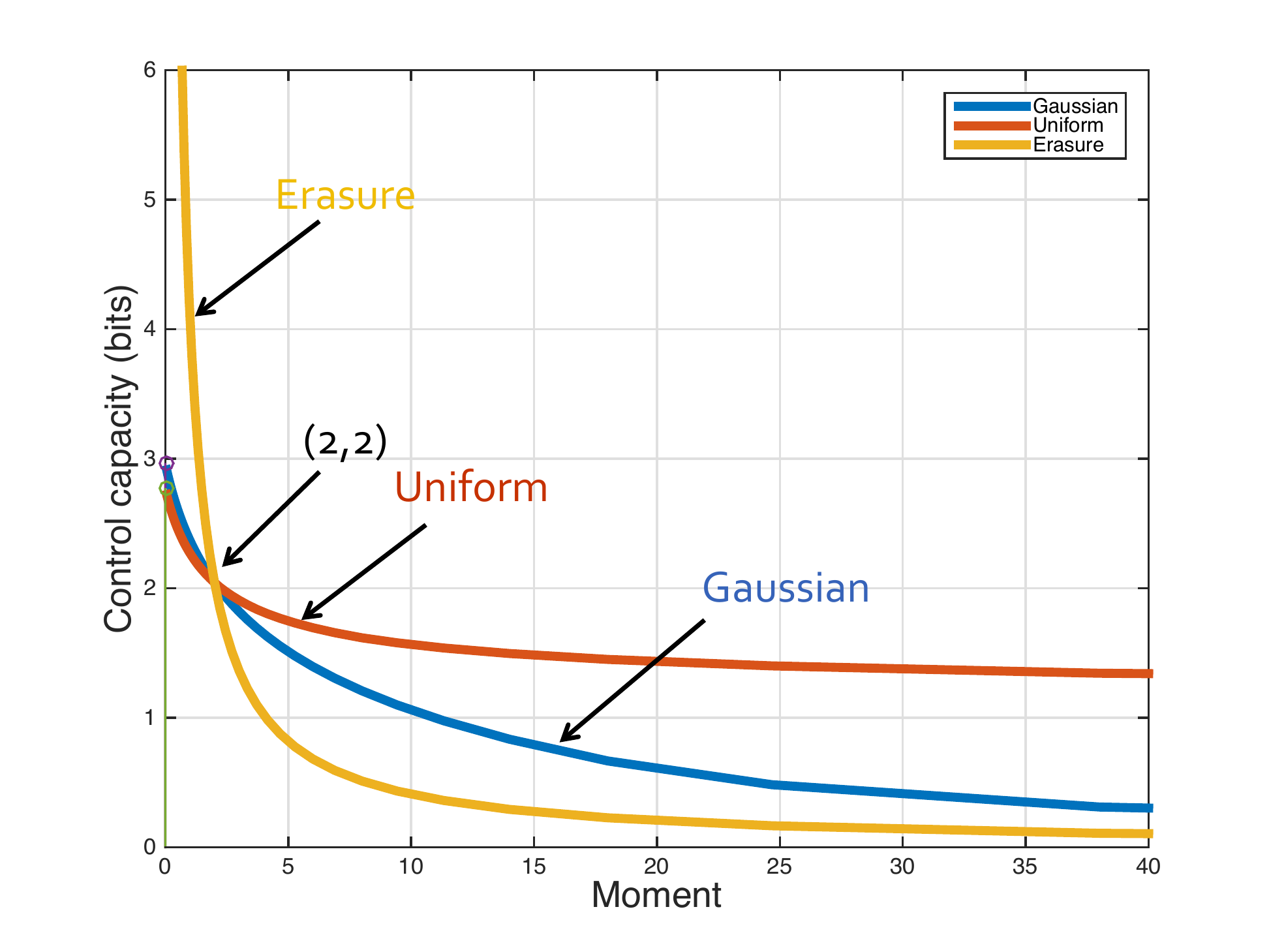}
\caption{The relationship between the different moment-senses of control capacity. For the Uniform, the zero-error control capacity is $1.2075$ which is the asymptote as $\eta \rightarrow \infty$. As $\eta\rightarrow0$ the $\eta$-th moment capacity converges to the Shannon sense. Here the Shannon control capacity for the Uniform is $2.7635$ and for the Gaussian is $2.9586$, which are the two small points seen on the extreme left (i.e.~green and purple on the y-axis).
}
\label{fig:cc2}
\end{center}
\end{figure}

Fig.~\ref{fig:cc2} allows us to explore the behavior of $\eta$-th moment capacities for the same three channels. The plot presents the $\eta$-th moment capacities for Gaussian, Erasure and Uniform control channels. We chose the three distributions such that their second-moment capacities are $2$, and all three curves intersect there. As expected, from Thm.~\ref{thm:limits}, as $\eta \rightarrow 0$, the curves approach the Shannon capacities and as $\eta \rightarrow \infty$ the curves asymptote at the zero-error control capacity. The results in this paper help us characterize this entire space, while previously only the $(2,2)$ point was really known.

\section{Additive system noise} \label{sec:additive}

The development of control capacity in the previous sections ignored additive noise and focused on the multiplicative uncertainty in the actuation channel. The results in this section show that nothing was lost by this focus. 

Consider the system $\widetilde{S}_{a}$, with additive observation noise $V[n]$ and additive system disturbance $W[n]$. 
\begin{align}
\begin{split}
\widetilde{X}[n+1] &= a(\widetilde{X}[n] + B[n]\widetilde{U}[n]) + W[n], \\
\widetilde{Y}[n] &= \widetilde{X}[n] + V[n]. \label{eq:noise}
\end{split}
\end{align}

The multiplicative noise $B[n]$ in~\eqref{eq:noise} is distributed according to $p_{B}$ as in system $\S_{a}$ in~\eqref{eq:SAsystem}, with finite $\eta$-th moment.
$V[n]$ and $W[n]$ are independent random variables at each time $n$, with finite $\eta$-th moments. Let $M_{\eta} < \infty$ be such that $\bigE\left[|V[n]|^{\eta}\right] \leq  M_{\eta}$, $\bigE\left[|W[n]|^{\eta}\right] \leq  M_{\eta}$, $\bigE\left[|B[n]|^{\eta}\right] \leq M_{\eta} < \infty$. 

Further, in this section we allow $\widetilde{X}[0]$ to be a random variable such that $\bigE\left[|\widetilde{X}[0]|^{\eta}\right] \leq M_{\eta} < \infty$.

Theorem~\ref{thm:additive} will show that this system is indeed $\eta$-th moment stabilizable if the $\eta$-control capacity is large enough --- the same condition that tells us that the system $\S_{a}$ in~\eqref{eq:SAsystem} is $\eta$-th moment stabilizable. 

\begin{thm} \label{thm:additive}
Suppose that $\S_{a}$ in~\eqref{eq:SAsystem} is $\eta$-th  moment stabilizable, and that the $\eta$-th moment control capacity of the actuation channel $\S$ in~\eqref{eq:Ssystem}, $C_{\eta}(\S) > \log |a|$. Let $U[n] = d\cdot Y[n]$ be the linear memoryless stationary strategy that achieves this control capacity, and also $\eta$-th moment stabilizes the system $\S_{a}$. Then, the control strategy $\widetilde{U}[n] = d\cdot \widetilde{Y}[n]$ also stabilizes system $\widetilde{\S}_{a}$~\eqref{eq:noise} in the $\eta$-th moment sense.
\end{thm}

\begin{proof}

We know that when we apply the control strategy $U[n] = d\cdot Y[n]$ to system~\eqref{eq:SAsystem} we get:
\begin{align}
\bigE\left[|X[n+1]|^{\eta}\right] &= \bigE\left[|a(1 + d \cdot B) |^{\eta}\right] \bigE\left[|X[n]|^{\eta}\right] \nonumber \\
&= \bigE\left[|a(1 + d \cdot B)|^{\eta}\right]^{n+1} \bigE\left[ |X[0]|^{\eta}\right], \label{eq:base}
\end{align}
where $B\sim p_{B}(\cdot)$. We will use~\eqref{eq:base} to prove the theorem. 

First, consider the case when $\eta > 1$.

Let $Q_{k} := a(1 + d B[k])$. Also, let $\bigE\left[|Q_{k}|^{\eta}\right] = L$, where $L$ is not indexed by $k$ since the $Q_{k}$'s are i.i.d.. Since the $d$ achieves $C_{\eta} > \log |a|$, we have that $L<1$.

Now, consider the evolution of the system~\eqref{eq:noise}, under the control strategy $\widetilde{U}[n] = d\cdot \widetilde{Y}[n]$.
\begin{align*}
\widetilde{X}\left[n+1\right] &= \left(a(1 + d B[n])\right) \widetilde{X}[n] + a d B[n]V[n]+ W[n] \\
&= Q_{n} \left(Q_{n-1} \widetilde{X}[n-1] + a d B[n-1]V[n-1] + W[n-1]\right) + a d B[n]V[n] +W[n]\\
&=\widetilde{X}[0] \prod_{k=0}^{n} Q_{k} +  a d \sum_{k=0}^{n} B[k]V[k] \left(\prod_{j=k+1}^{n} Q_{j}\right) + \sum_{k=0}^{n} W[k] \left(\prod_{j=k+1}^{n} Q_{j}\right). 
\end{align*}
Consider the $\eta$-th norm of $\widetilde{X}[n+1]$, 
\begin{align*}
\bigE[|\widetilde{X}[n+1]|^{\eta}]^{\frac{1}{\eta}} 
&= \bigE\left[\biggr|\widetilde{X}[0]\left(\prod_{k=0}^{n} Q_{k}\right) +  ad\sum_{k=0}^{n} B[k]V[k] \left(\prod_{j=k+1}^{n} Q_{j}\right) + \sum_{k=0}^{n} W[k] \left(\prod_{j=k+1}^{n} Q_{j}\right) \biggr|^{\eta}\right]^{\frac{1}{\eta}}.
\end{align*}

Minkowski's inequality states that for $a_{k}, b_{k} \in \mathbb{R}$, 
$$\left(\sum_{k=0}^{n} \left|a_{k} + b_{k}\right|^{\eta}\right)^{\frac{1}{\eta}} \leq \left(\sum_{k=0}^{n} \left|a_{k} \right|^{\eta}\right)^{\frac{1}{\eta}}  + \left(\sum_{k=0}^{n} \left|b_{k}\right|^{\eta}\right)^{\frac{1}{\eta}}.$$
Applying this gives:
\begin{align*}
\bigE[|\widetilde{X}[n+1]|^{\eta}]^{\frac{1}{\eta}} 
&\leq \bigE\left[\biggr|\widetilde{X}[0] \left(\prod_{k=0}^{n} Q_{k}\right) \biggr|^{\eta}\right]^{\frac{1}{\eta}} + a d\sum_{k=0}^{n} \bigE\left[\biggr|B[k]V[k] \left(\prod_{j=k+1}^{n} Q_{j}\right) \biggr|^{\eta}\bigr.\right]^{\frac{1}{\eta}} + \sum_{k=0}^{n} \bigE\left[\biggr|W[k] \left(\prod_{j=k+1}^{n} Q_{j}\right) \biggr|^{\eta}\bigr.\right]^{\frac{1}{\eta}}\\
&= L^{\frac{n+1}{\eta}} \bigE\left[ |X[0]|^{\eta} \right]^{\frac{1}{\eta}} +  a d\sum_{k=0}^{n} L^{\frac{n-k}{\eta}}\bigE\left[|B[k]V[k]|^{\eta}\bigr.\right]^{\frac{1}{\eta}} +  \sum_{k=0}^{n} L^{\frac{n-k}{\eta}}\bigE\left[|W[k]|^{\eta}\bigr.\right]^{\frac{1}{\eta}}\\
&\leq \left(\sum_{k=-1}^{n} L^{\frac{n-k}{\eta}}\right) M_{\eta}^{\frac{1}{\eta}} + \left(\sum_{k=0}^{n} L^{\frac{n-k}{\eta}}\right) \left(ad M_{\eta}^{\frac{2}{\eta}}\right)\\
&< \frac{1}{1-L^{\frac{1}{\eta}}} ~M_{\eta}^{\frac{1}{\eta}}\left(1 + ad M_{\eta}^{\frac{1}{\eta}}\right),
\end{align*}
where we use $L<1$ in the last step above.
Hence for all $n$,
\begin{equation}
\bigE\left[|\widetilde{X}[n+1]|^{\eta}\right] <  \frac{1}{(1-L^{1/\eta})^{\eta}} M_{\eta} \left(1 + ad M_{\eta}^{\frac{1}{\eta}}\right)^{\eta} < \infty.
\end{equation}

Now we consider the case where $\eta \leq 1$. 
From above, we know that
\begin{align*}
\widetilde{X}[n+1] &= \widetilde{X}[0] \prod_{k=0}^{n} Q_{k} + a d\sum_{k=0}^{n} B[k]V[k] \left(\prod_{j=k+1}^{n} Q_{j}\right) + \sum_{k=0}^{n} W[k] \left(\prod_{j=k+1}^{n} Q_{j}\right).
\end{align*}
Notice that for $\eta \leq 1$, concavity tells us that we can upperbound the $\eta$-th power of a sum by the sum of the $\eta$-th powers of the individual terms:
\begin{align*}
\bigE[|\widetilde{X}[n+1]|^{\eta}] &\leq \bigE\left[\biggr|\widetilde{X}[0] \left(\prod_{k=0}^{n} Q_{k}\right) \biggr|^{\eta}\right] + ad\sum_{k=0}^{n} \bigE\left[\biggr|B[k]V[k] \left(\prod_{j=k+1}^{n} Q_{j}\right) \biggr|^{\eta}\right]+\sum_{k=0}^{n} \bigE\left[\biggr|W[k]\left(\prod_{j=i+1}^{n} Q_{j}\right) \biggr|^{\eta}\right]\\
&\leq L^{n+1}\bigE\left[| X[0] |^{\eta}\right] + ad \sum_{k=0}^{n} L^{n-k}\bigE\left[|B[k]V[k]|^{\eta}\right]+\sum_{k=0}^{n} L^{n-k}\bigE\left[|W[k]|^{\eta}\right]\\
&< \frac{1}{1-L}~M_{\eta}~\left(1 + a d M_{\eta}\right) < \infty.
\end{align*}

Thus, in both cases the system $\widetilde{\S}_{a}$ is $\eta$-th moment stabilizable using the same memoryless linear stationary strategy that stabilized $\S_{a}$. Note the controls applied are not the same, because they are based on the observations $\widetilde{Y}[\cdot]$, but the control gain $d$ is the same. 

Although this section has talked exclusively about $\eta$-moment stability, Theorem~\ref{thm:limits} tells us that we get essentially the same result for the Shannon sense of stability as well. This is because if $C_{\sh}(\S(p_{B})) > \log |a|$, we know since $\lim_{\eta \rightarrow 0} C_{\eta}(\S(p_{B})) = C_{\sh}$ that there must exist an $\eta > 0$ for which $C_{\eta} > \log |a|$ as well. The corresponding control law gives the desired result. To understand zero-error control capacity with additive noise, a proof that exactly parallels the proof above can be given. Instead of expectations, maximizations can be used along with assuming bounds on all the additive disturbances as well as the initial condition.

\end{proof}

\section{Control capacity with side information} \label{sec:sideinfo}
This final section allows us to take advantage of the informational perspective on uncertainty in control systems developed in the earlier sections, and we can understand the impact of side information in systems. We provide a definition for the notion of control capacity with side information. Theorem~\ref{thm:sideinfo}  provides an operational meaning for the definition. Theorem~\ref{thm:sideinfocap} and~\ref{thm:sideinfocapeta} allow us to calculate the control capacity with side information in the i.i.d.~case.

We consider the same system $\S$ as in~\eqref{eq:Ssystem}, however, consider that the controller has access to additional side information $T[n]$ in addition to the observations $Y[n]$ at time $n$. 
\begin{align*}
X[n+1] &= X[n] + B[n] U[n]\\
Y[n] &= X[n], \label{eq:sideinfosystem}
\end{align*}

The pair $(B[k], T[k])$ for  $0 \leq k \leq n$ is drawn from a joint distribution $p_{B,T}(\cdot,\cdot)$ at each time. The applied control signal $U[n]$ can causally depend on $Y[k], 0\leq k \leq n$ as well as on the side information $T[k], 0\leq k \leq n$.

Now, we can naturally extend the definition in~\cite{controlcapacity} to define control capacity with side information.
\begin{defn}
The Shannon control capacity of the system $\mathcal{S}$ with side information $T$ is defined as
\begin{align}
C_{\sh}(\mathcal{S}|T) = \liminf_{n\to\infty}\underset{U[0], \cdots, U[n]}{\max}~\frac{1}{n} \bigE\left[\log \frac{|X[0]|}{|X[n]|} \right],
\end{align}
where each $U[n]$ is a causal function of $(Y[k], T[k])$ for $0\leq k \leq n$.
\end{defn}

\noindent The control capacity with side information is the maximum uncertainty (in bits) that can be dissipated on average from the state using both the observation and the side information. Parallel to Thm~\ref{thm:shannoncapacitygrowth} we can immediately characterize the logarithmic stabilizability of the system $\S_{a}$ when given access to the same side information.

\begin{thm} \label{thm:sideinfo}
Consider the system $\S_{a}$ as in~\eqref{eq:SAsystem} but with access to the additional side information $T[n]$ at time $n$. 
Then, system $\mathcal{S}_{a}$ is logarithmically stabilizable with side information $T[k]$ received by the controller at time $k$ if
$C_{\sh}(\mathcal{S}|T) > \log |a|.$ Conversely, if the system $\mathcal{S}_{a}$ is logarithmically stabilizable with
side information $T[k]$ received by the controller at time $k$, then
$C_{\sh}(\mathcal{S}|T) \geq \log |a|.$
\end{thm}
The proof of this theorem follows that of Thm.~\ref{thm:shannoncapacitygrowth} and is omitted here. 

The next theorem shows that the value of the side information is computable and can be thought of as a conditional expectation when $(B[n], T[n])$ are distributed i.i.d. according to a joint distribution $p_{B,T}.$
\begin{thm}
The Shannon control capacity of the system $\mathcal{S}(p_{B, T})$ with side information $T[n]$ at time $n$ is given by
\begin{align}
C_{\sh}(\mathcal{S}|T) = \bigE \left[\underset{~d(T)}{\max}~\bigE\left[-\log \bigr|1 + B \cdot d(T)\bigr|~\biggr|~T \right]\right].
\end{align}
The maximization allows $d$ to depend on the side-information $T$. 
\label{thm:sideinfocap}
\end{thm}
The proof of this theorem also follows the proof of Thm.~\ref{thm:calcshannoncap} and is not provided. It is discussed in~\cite{controlcapacitysideinfo, gireejaBeast}.

An $\eta$-th moment control capacity with side information for the system $\mathcal{S}$ also makes sense.
\begin{defn}
The $\eta$-th moment control capacity of the system $\mathcal{S}$ with side information $T$ is defined as
\begin{align}
C_{\eta}(\mathcal{S}|T) = \lim_{n\rightarrow\infty}\underset{U[0], \cdots U[n]}{\max}~\frac{1}{n}\frac{1}{\eta} \log \bigE\left[\biggr|\frac{X[0]}{X[n]}\biggr|^{\eta} \right],
\end{align}
where each $U[n]$ is a causal function of $(Y[k], T[k])$ for $0\leq k \leq n$.
\end{defn}

\begin{thm}
Let  $(B[n], T[n])$ be distributed i.i.d. according to a joint distribution $p_{B,T}.$ Then, the $\eta$-th moment control capacity of the system $\mathcal{S}(p_{B,T})$ is given by
\begin{align*}
C_{\eta}(\mathcal{S}|T) = \frac{-1}{\eta}\log~\bigE \left[\underset{d(T)}{\min}~\bigE \left[\bigr|1 + B \cdot d(T)\bigr|^{\eta}~\biggr|~T \right]\right].
\end{align*}
\label{thm:sideinfocapeta}
\end{thm}
The proof of this theorem follows the proof of the corresponding theorem without side information and is omitted.

\subsection{Control capacity with side information: an example}

\begin{figure}[htbp]
\begin{center}
\includegraphics[width = 0.5\textwidth]{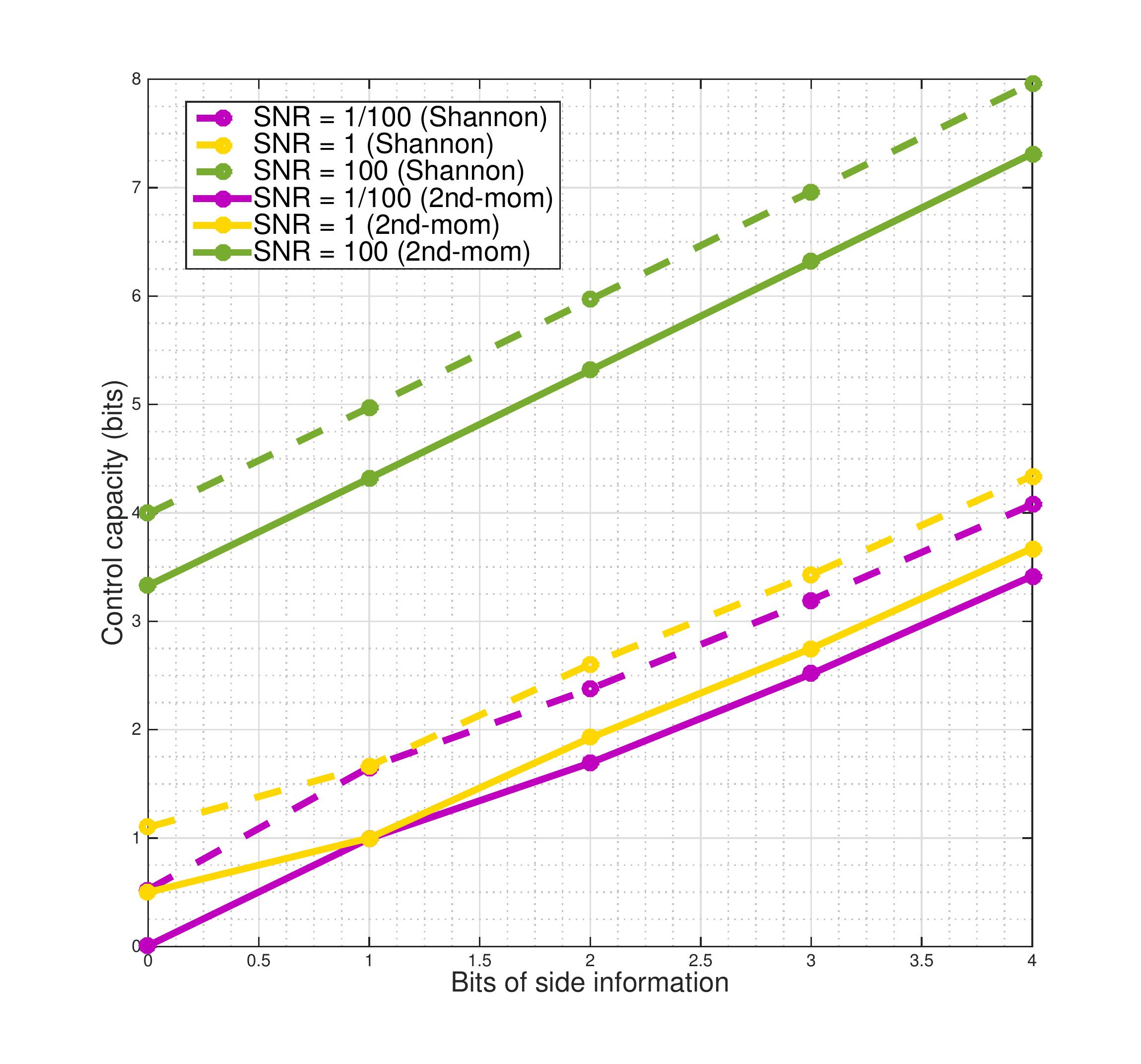}
\caption{This plot shows the increase in control capacity with side information. The actuation channels considered have a uniform distribution with different mean/standard deviation ratios (SNR) as in the legend.}
\label{fig:sideinfoplot}
\end{center}
\end{figure}

As an example, we plot the change in Shannon and second-moment control
capacities with zero to four bits of side information for a set of
actuation channels in Fig.~\ref{fig:sideinfoplot}. As in earlier figures, we focus on the ``SNR'' of the actuation channel, i.e. $\frac{\mu_{B}^{2}}{\sigma_{B}^{2}}$, as we know this is the critical parameter to compute second-moment control capacity from Corollary~\ref{cor:meansq}. We plot the control capacities for actuation channels with base ``SNR'' $\frac{1}{100}, 1$ and $100$, and thus mean to standard deviation ratios of $\frac{1}{10}, 1$ and $10$.

We consider a uniform distribution on the unreliability in the actuation channel. $B\sim\textrm{Uniform}\left[b_{1}, b_{2}\right]$. The controller is provided one bit of side information in the form of knowledge of the half-interval into which the realization of $B$ falls. i.e. the controller is told whether the realization of $B$ is in $\left[b_{1}, \frac{b_{1}+b_{2}}{2}\right]$ or in $\left(\frac{b_{1}+b_{2}}{2}, b_{2}\right]$. Two bits of side information resolves the interval into four equal-sized subintervals, and so on.

The green curves represent the second-moment (solid) and Shannon (dashed) control capacity for the uniform distribution with SNR = $100$. For both these curves, as the number of bits of side information increases the slope of both these curves approaches $1$ but are a shade below $1$ in the region close to $0$. 

On the other hand, consider the pink lines that represent the second
moment (solid) and Shannon (dashed) control capacities when the SNR $=
\frac{1}{100}$. The slope of the dashed line (Shannon) between 0 and 1
is actually slightly greater than 1! In this case, half the time, the
first bit of side information reveals perfectly the sign of the
distribution and can increase the control capacity by more than one
bit. This can been seen in the carry-free models in the Appendix, where the value of a bit of side-information can be more than a bit!  Of course, as the side information increases the capacity steadily increases and eventually, one bit of side information only increases the control capacity by one bit --- we can see that the slope of the curve tends to $1$. 

The slope of the second-moment control capacity (for SNR $=1/100$) between $0$ and $1$ is still less than $1$, but we see here that the value of the first bit of side information (that reveals the sign) is still more valuable than the second bit of side information. This curve also converges to slope $1$ as the controller gets more side information. 

Finally, we come to the control capacities of the distributions with SNR $= 1$ with the yellow curves. (Note here that the values for both the Shannon and second-moment control capacities with one bit of side information (i.e. the points corresponding to x-coordinate $1$) are slightly higher than the points for SNR $= 1/100$ even though it is not apparent in this figure). These curves shows a very intriguing phenomenon --- the first bit of side information is actually worth less than a bit, and the first bit of side information is worth less than the second bit that is received. This is certainly something we plan to investigate further.

\section{Acknowledgements}
The authors would like to thank the NSF for a GRF and CNS-0932410, CNS-1321155, and ECCS-1343398. Thanks also to Yuval Peres and Miklos Racz for helpful discussions.

\appendices

\section{Bit-level models for uncertainty in control}\label{sec:carryfreeappendix}
This appendix describes bit-level models for unknown dynamical systems. These simple models motivated the definitions and theorems in the paper, and this appendix is included to share the insights from these models with the reader. 

The carry-free bit level models described here build on previous bit-level models developed in wireless network information theory, i.e.~the deterministic models developed by Avestimehr, Diggavi and Tse (ADT models)~\cite{tseDetmodel}, and lower-triangular or carry-free models developed by Niesen and Maddah-Ali~\cite{lowerTriangular}. We previously used these models~\cite{carryfree} to understand the $\log \log SNR$ result for communication over channels with unknown fading~\cite{lapidothLogLog} and then to explore noncoherent relay networks. We call these models ``carry-free'' to indicate that the addition operation is defined without carry-overs from one bit level to the next. We will discuss this in more detail in below and in Figure~\ref{fig:multiplication}.

\subsection{Bit-level models for rate-limited control}
First, we will describe how the data-rate theorems~\cite{wong1997systems,nair2000stabilization,tatikonda1998control} can be understood using bit-level models.
Consider the system:
\begin{align}
X[n+1] &= a \cdot X[n] + U[n] + W[n], \label{eq:tat}
\end{align}
where $a$ is a fixed scalar, and the additive noise $W[n]$ is drawn i.i.d.~Unif$[0,1]$.
The controller must generate $U[n]$ based on observations over an $R$-bit channel. The data-rate theorems show that a rate of $R > \log |a|$ is necessary and sufficient to stabilize the system.

It turns out we can understand this result pictorially through bit-level models. Let us represent the system state $X[n]$ by its binary expansion as:
\[
x_{m}[n]x_{m-1}[n]\cdots x_{1}[n]x_{0}[n].x_{-1}[n]x_{-2}[n]\cdots,
\]
where $x_{i}[n] \in \{0,1\}$. The index $m\in\mathbb{Z}$ represents the highest non-zero bit level of the state. To recover the value of the state we can consider the polynomial-like formal series:
\begin{align*}
x[n](z) = x_{m}[n]z^{m} + x_{m-1}[n] z^{m-1} + \cdots + x_{0}[n] + x_{-1}z^{-1} \cdots.
\end{align*} 
Substituting $z=2$ will give back $X[n]$.

Let us also consider also the system gain $a > 0$ as expressed in binary. For simplicity, let us assume that $a\geq 1$ is a power of two, and hence we can write it as a monomial of degree 
$g_{a} = \log a$:
\[ 
a[n](z) = a(z) = 1 \cdot z^{g_{a}}.
\]

These bit-level models are particularly conducive to modeling explicit rate constraints, since a rate limit simply caps the number of levels that are visible to the estimator or controller at any given time. We can construct a bit-level model of the system in~\eqref{eq:tat} as below:
\begin{equation} \label{eq:cfratelimited}
x[n+1](z) = a(z) x[n](z) +  u[n](z) + w[n](z),
\end{equation}
where $u[n](z)$ is the control signal in binary that is based on observations received over an $R$ bit channel at each time. $w[n](z)$ is an additive binary noise sequence. We restrict this to be below the decimal level and so the highest power in the formal series representation is $-1$.
\begin{align*}
w[n](z) = w_{-1}[n]z^{-1} + x_{-2}[n] z^{-2} + \cdots,
\end{align*} 
Each $w_{-i}$ is drawn i.i.d. Bernoulli-$\left(\frac{1}{2}\right)$.

\begin{figure}[thbp]
\begin{center}
\includegraphics[width=.3\textwidth]{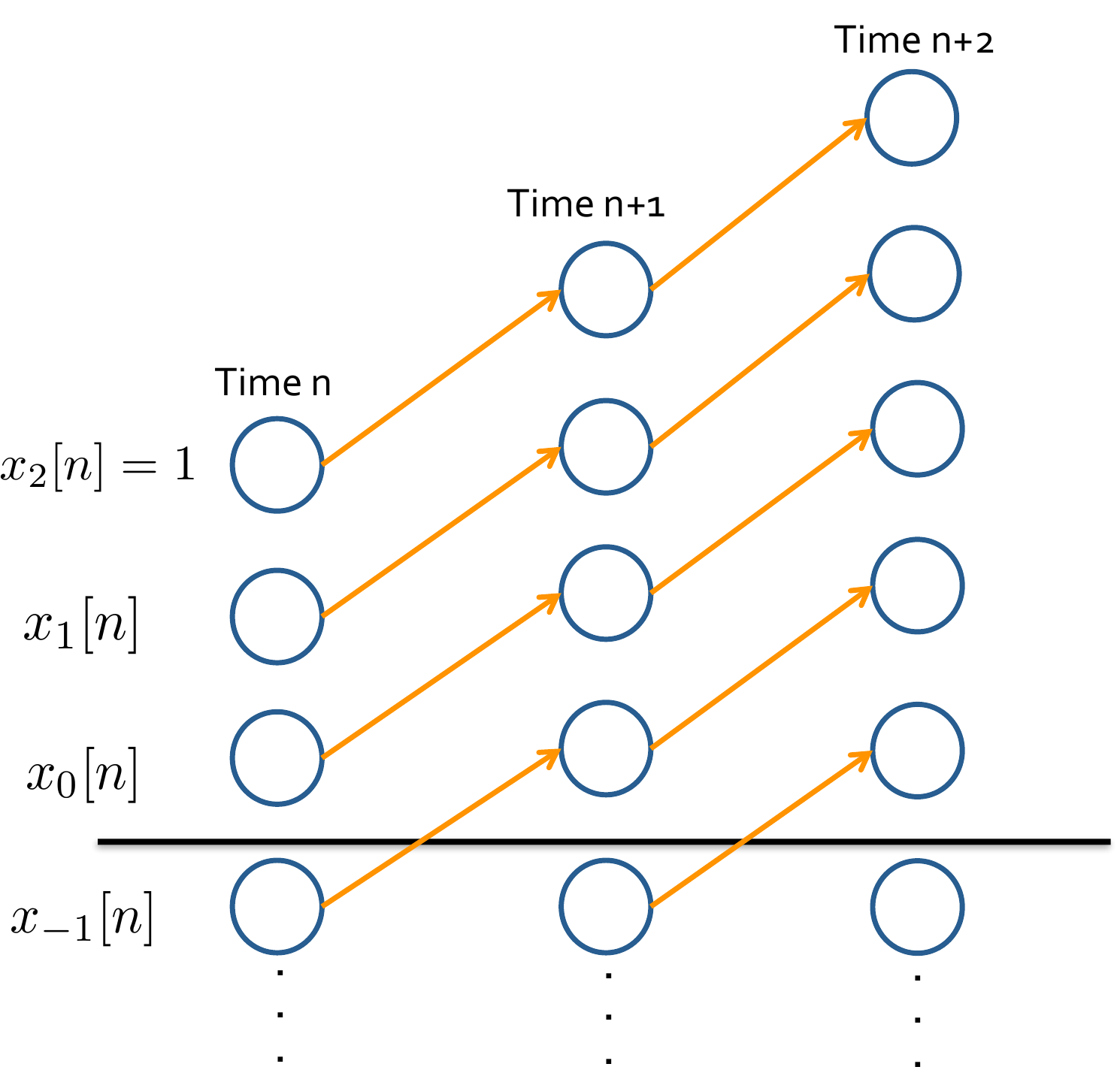}
\caption{The open-loop system state can be thought of as a stack of bits marching upward with the gain $a$.}
\label{fig:bitsmarching}
\end{center}
\end{figure}

Figure~\ref{fig:bitsmarching} represents how the system state in~\eqref{eq:cfratelimited} is growing. Consider the bits that represent the state arranged as a vertical stack, with the most significant bit at the top. Multiplication by the gain $a(z)$ causes the stack to increase in height by $g_{a}$ levels each time. As the bit-levels rise, the bits that are below the decimal point at the noise level rise above the noise level and bring added uncertainty to the system. To avoid this stack growing unboundedly, the controller must cancel at least $g_{a}$ bits at each time step, and to do this it must know their value. Hence, the minimum communication rate required for estimation is $\log a = g_{a}$. 

\subsection{Carry-free models}
Carry-free models generalize the idea of bit-level multiplication in the previous subsection to the case where the gain might not be a power-of-two. Our primary interest is in modeling the impact of randomness in system parameters, and thus we want to capture multiplication by random binary bit strings. Before introducing randomness into the picture, we first generalize to the case when $a(z)$ is not a power of two. First, we define carry-free addition and multiplication between two binary strings in a manner that parallels those operations for formal power series.

\begin{defn}
Let $a(z) = a_{m}...a_{1}a_{0}a_{-1}...$ and $b(z) = b_{n}...b_{1}b_{0}b_{-1}...$ be two binary strings. Then, their carry-free sum is defined as $c_{m}...c_{1}c_{0}c_{-1}... := a_{n}...a_{1}a_{0}a_{-1}... \oplus_{c} b_{n}...b_{1}b_{0}b_{-1}...$ where $c_{i} = a_{i} + b_{i}~\left(\textrm{mod}~2\right)$. 
\end{defn}

The addition operation involves no carryovers unlike in real addition. Bit interactions at one level do not affect higher level bits. We derive the name ``carry-free'' from this property. 

\begin{defn}
Let $a(z) = a_{m}...a_{1}a_{0}a_{-1}...$ and $b(z) = b_{n}...b_{1}b_{0}b_{-1}...$ be two binary strings.  Then, their carry-free multiplication is defined as $c_{2m}...c_{1}c_{0}c_{-1}... := a_{m}...a_{1}a_{0}a_{-1}... \otimes_{c} b_{m}...b_{1}b_{0}b_{-1}...,$ where $c_{i} = \underset{k}{\sum} a_{k}b_{i-k}~\left(\textrm{mod}~2\right)$. 
\end{defn}

Thus, carry-free multiplication of the bit-levels is like convolution of the signals represented by the bit levels in time, where the bit-level corresponds to the time index (Figure~\ref{fig:multiplication})~\cite{lowerTriangular, carryfree}. We note that it is commutative and associative.

\begin{figure}[htbp]
\begin{center}
\includegraphics[width=0.5\textwidth]{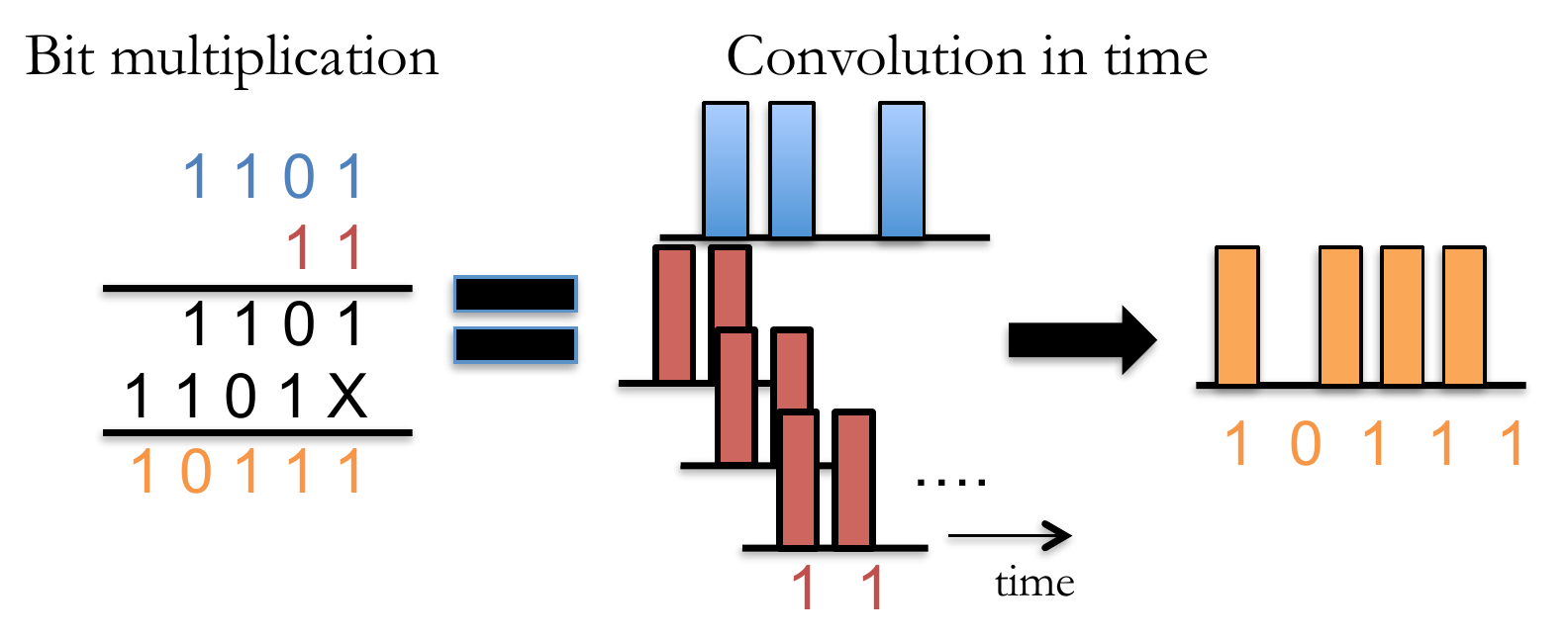}
\caption{Carry-free multiplication as ``convolution'' in time of the signal represented by the bit-levels.}
\label{fig:multiplication}
\end{center}
\end{figure}

Let $a(z)$ be $a(z) = 1 \cdot z^{g_{a}}+a_{g_{a}-1}z^{g_{a}-1}+\cdots,$ so $a$ in~\eqref{eq:tat} is not restricted to being a power of two. 
Now, we can model the same bit-level state evolution as in~\eqref{eq:cfratelimited}, but using carry-free multiplication and addition. Figure~\ref{fig:carryfreefixed} shows a bubble-picture for the rate-limited bit-level system. This figure captures the effect of growth similar to Figure~\ref{fig:bitsmarching} for one time-step, except with carry-free multiplication. 

\begin{figure}[htbp]
\begin{center}
\includegraphics[width=.35\textwidth]{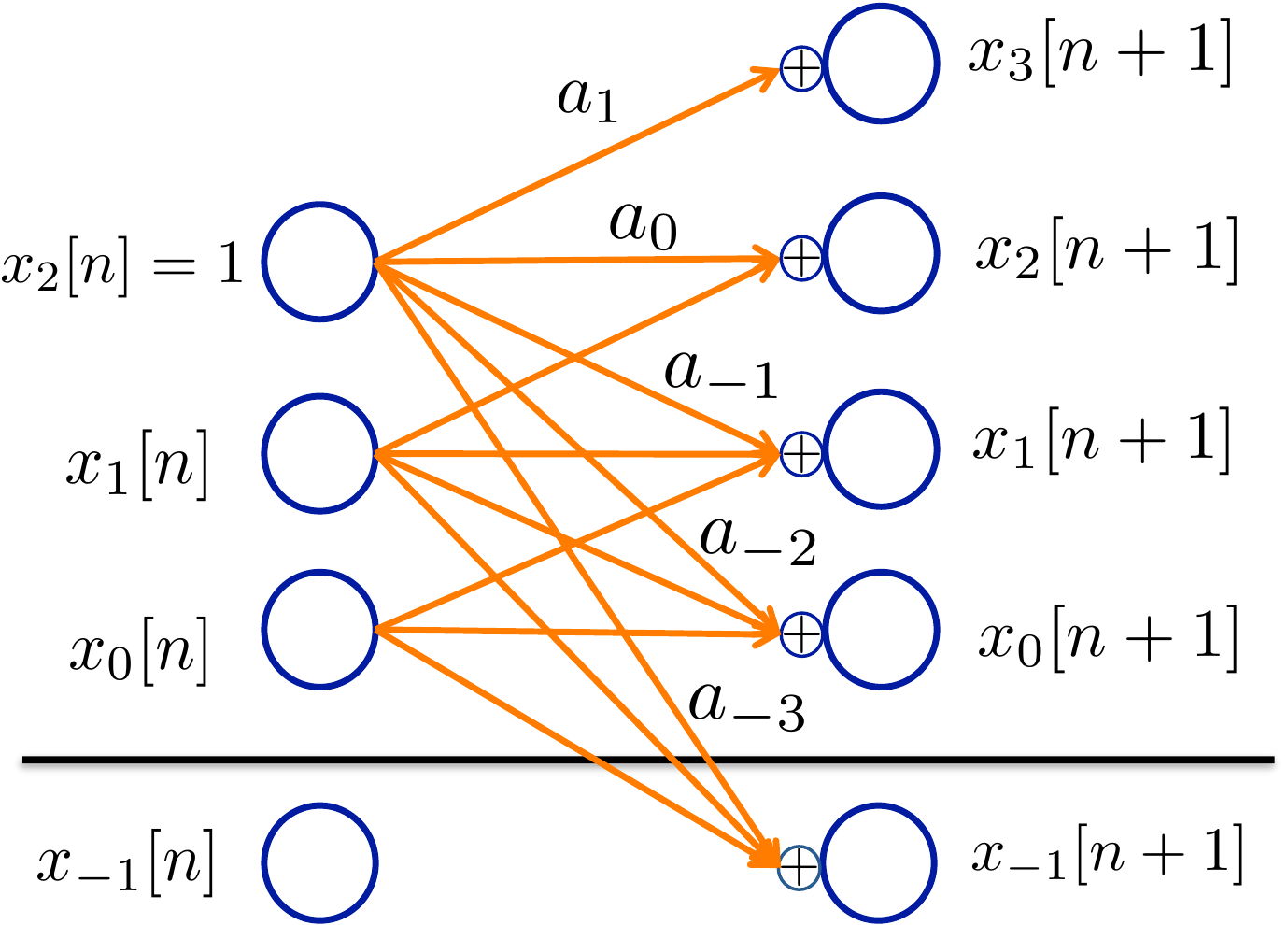}
\caption{Carry-free for models with highest bit at level $3$, and power $g_{a}=1$.}
\label{fig:carryfreefixed}
\end{center}
\end{figure}

\subsection{Carry-free actuation channels}
We build on these ideas to model the system in~\eqref{eq:SAsystem}, with a random actuation gain $B[n]$. We will use these models to understand the zero-error sense of stability as well the Shannon notion of stability (stability in expectation). For this, we introduce the notion of carry-free multiplication by a random gain to capture the i.i.d. nature of the $B[n]$'s.

We consider the binary expansion for a random actuation gain $B[n]$. $g_{b}$ is the highest non-zero bit level. The high-order bits are deterministic, and we define $g_{\det}$ as the highest deterministic level. There are a total of $g_{\det} - g_{\ran}$ deterministic bits, with $g_{\ran}$ are the first random (Bernoulli$-\left(\frac{1}{2}\right)$) bit level. 
Thus,
\begin{align}
b[n](z) &=  b_{g_{b}}[n] \cdot z^{g_{b}}+b_{g_{b}-1}[n]\cdot z^{g_{b}-1}+\cdots  \nonumber \\
&=1 \cdot z^{g_\det}+0\cdot z^{g_\det-1}+0\cdot z^{g_\det-2}+\cdots+b _{g_\ran}[n] \cdot z^{g_\ran} + b_{g_\ran-1}[n] \cdot z^{g_\ran-1} + \cdots \label{fig:b},
\end{align}
and we have
\begin{align*}
g_\ran = \max\left\{i|b_i \sim \mbox{Bernoulli}\left(\tfrac{1}{2}\right),  g_{b} \geq i \geq -\infty\right\}.
\end{align*}
Since the $B[n]$ are identically distributed, $g_{\det}$ and $g_{\ran}$ do not vary with $n$, but the realizations of random bits are drawn identically at each time. $b_{i}[n] \sim\mbox{Bernoulli}(\frac{1}{2})$ for $i \leq g_\ran$. The realizations of these bits are unknown to the controller. Also we can write 
\begin{align*}
g_\det = 
\begin{cases}
\max\{i \mid b_i = 1\}~\mbox{if}~ g_\ran < g_{b} \\
g_\ran,~\mathrm{otherwise}.\\
\end{cases}
\end{align*}

Without loss of generality, we fix all the bits from $b_{g_\det-1}$ to $b_{g_\ran+1}$ to be $0$, and these bits are known to the controller. Our arguments extend to any other set of deterministic bits, with a leading $1$. 
This is illustrated in Figure~\ref{fig:cfut}.

\begin{figure}[htbp]
\begin{center}
\includegraphics[width=.6\textwidth]{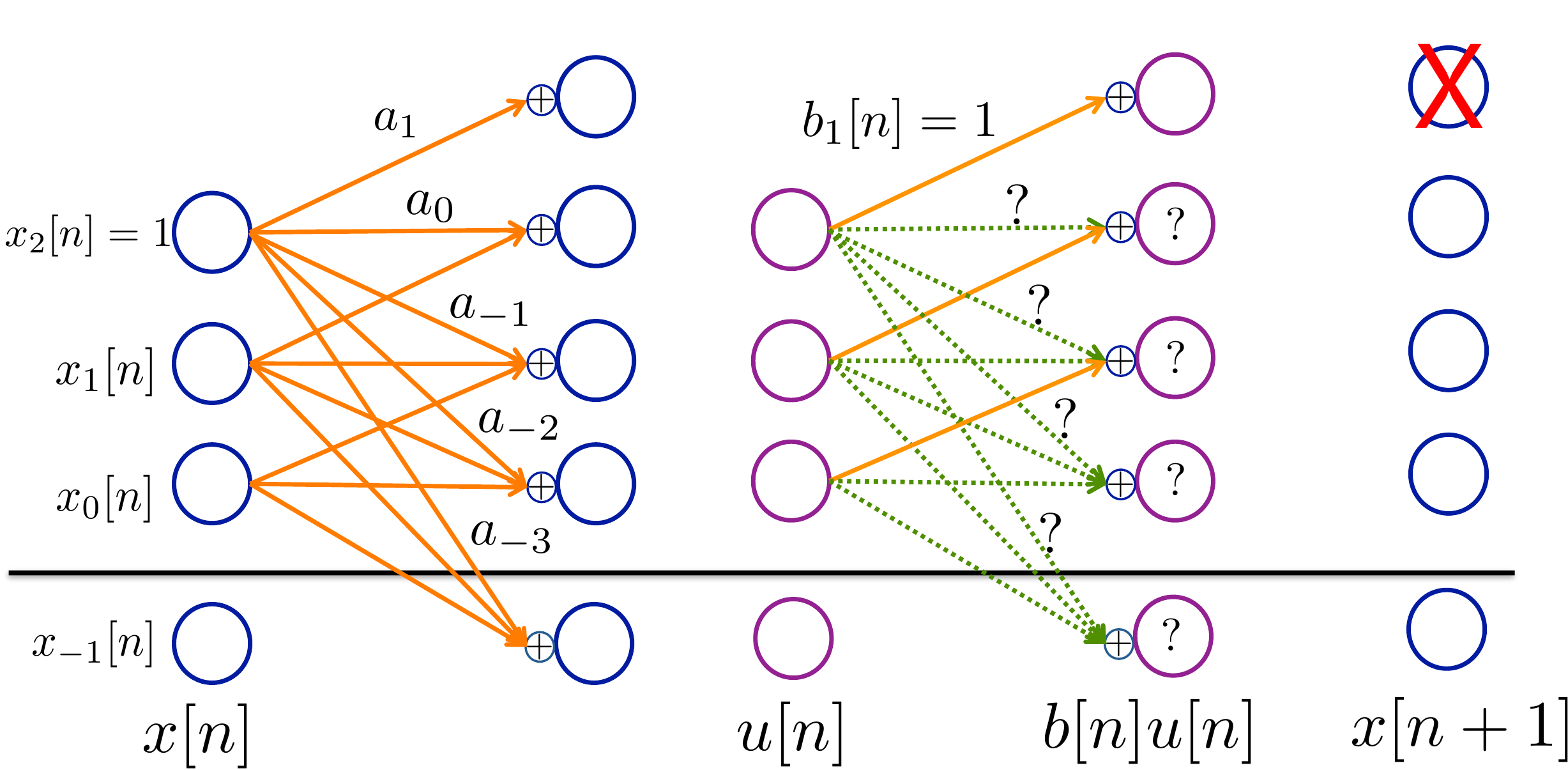}
\caption{This figure shows a carry-free model for system~\eqref{eq:cfwithb}. The solid orange lines represent deterministic bits and the dotted green lines represent random bits. The system gain has $g_{a}=1$, $g_{\det} = 1$ and $g_{ran} = 0$. So $b_{1}[n]=1$ is a deterministic value that is the same at each time step. $g_{ran}[n] = 1$, so bits $b_{0}[n], b_{-1}[n],b_{-2}[n], \cdots$ are all random Bernoulli-$\left(\frac{1}{2}\right)$ random bits. As a result the controller can only influence the top bit of the control going in.}
\label{fig:cfut}
\end{center}
\end{figure}

Now we introduce the carry-free system model for system~\eqref{eq:SAsystem}. We restrict attention to the case where the gain on the state is a known constant $a(z) = 1\cdot z^{g_{a}}$ for all $n$. Consider the system $\mathcal{S}^{\CF}_{a}$:
\begin{align}
x[n+1](z) = a(z) \cdot x[n](z) + b[n](z)\cdot u[n](z) + w[n](z)\label{eq:cfwithb}
\end{align}
Let $d_{n}$ be the degree of $x[n](z)$. Our aim is to understand the stability of this system, which is captured by the behavior of the degree $d_{n}$. 

Pictorially, the illustration in Figure~\ref{fig:cfut} shows us that $d_{n}$ will be bounded with probability $1$ only when $g_{\det} -  g_{\ran} \geq  g_{a}$ for $g_{a} >0$. (The system is self-stabilizing when $g_{a} \leq 0$.) In all of the figures, the solid orange lines represent deterministic bits and the dotted green lines represent random bits. Since at every time step the magnitude of the system state increases by exactly $g_{a}$ bits, as long as the controller can dissipate $g_{a}$ bits, it can stabilize the system. 

{\remark{Unlike the case with deterministic system gains, ADT models would not suffice to understand systems with random control gains, since they only capture bit shifts. The loss of information due to multiplicative scrambling by the random gains is essential to understand the bottleneck due to the uncertainty.}}

We now formalize some notions of stability for carry-free models and define control capacity.

\subsection{Zero-error stability}

For the zero-error stability of a carry-free system, we require that the degree of the state be bounded with probability $1$.

\begin{defn}
The system~\eqref{eq:cfwithb1} is stablizable in the zero-error sense if there exists a control strategy $u[\cdot](z)$ such that there exists $M< \infty$ and $N > 0$ such that for all $n \geq N$, we know $\P(d_{n} < M) = 1$.
\end{defn}

Now let us consider the system $\mathcal{S}^{\CF}$ below, with $a=1$ and $g_{a} = 0$. This parallels the system in~\eqref{eq:Ssystem}, and is illustrated in Figure~\ref{fig:cfccnoa}.
\begin{align}
x[n+1](z) = x[n](z) + b[n](z)\cdot u[n](z) + w[n](z)\label{eq:cfwithb1}
\end{align}
The maximum gain $g_{a}$ that can be tolerated for the system~\eqref{eq:cfwithb} is related to the maximum rate at which uncertainty can be dissipated in~\eqref{eq:cfwithb}. Thus, we define as the zero-error control capacity, $C_{\ze}$, as the maximum possible decay in the degree of the system per unit time.

\begin{defn}
The zero-error control capacity of the system $\mathcal{S}^{\CF}$ from~\eqref{eq:cfwithb1} is defined as the largest constant $C_{\ze}\left(\mathcal{S}^{\CF}\right)$ for which there exists a control strategy $u[0](z), \ldots, u[n](z)$ such that 
\begin{align*}
\P\left(\frac{1}{n}\bigr(d_{0} - d_{n}\bigr) \geq C_{\ze}\bigr(\mathcal{S}^{\CF}\bigr)\right) = 1.
\end{align*}
for all time steps $n$.
\end{defn}

\begin{figure}[htbp]
\begin{center}
\includegraphics[width = .5\textwidth]{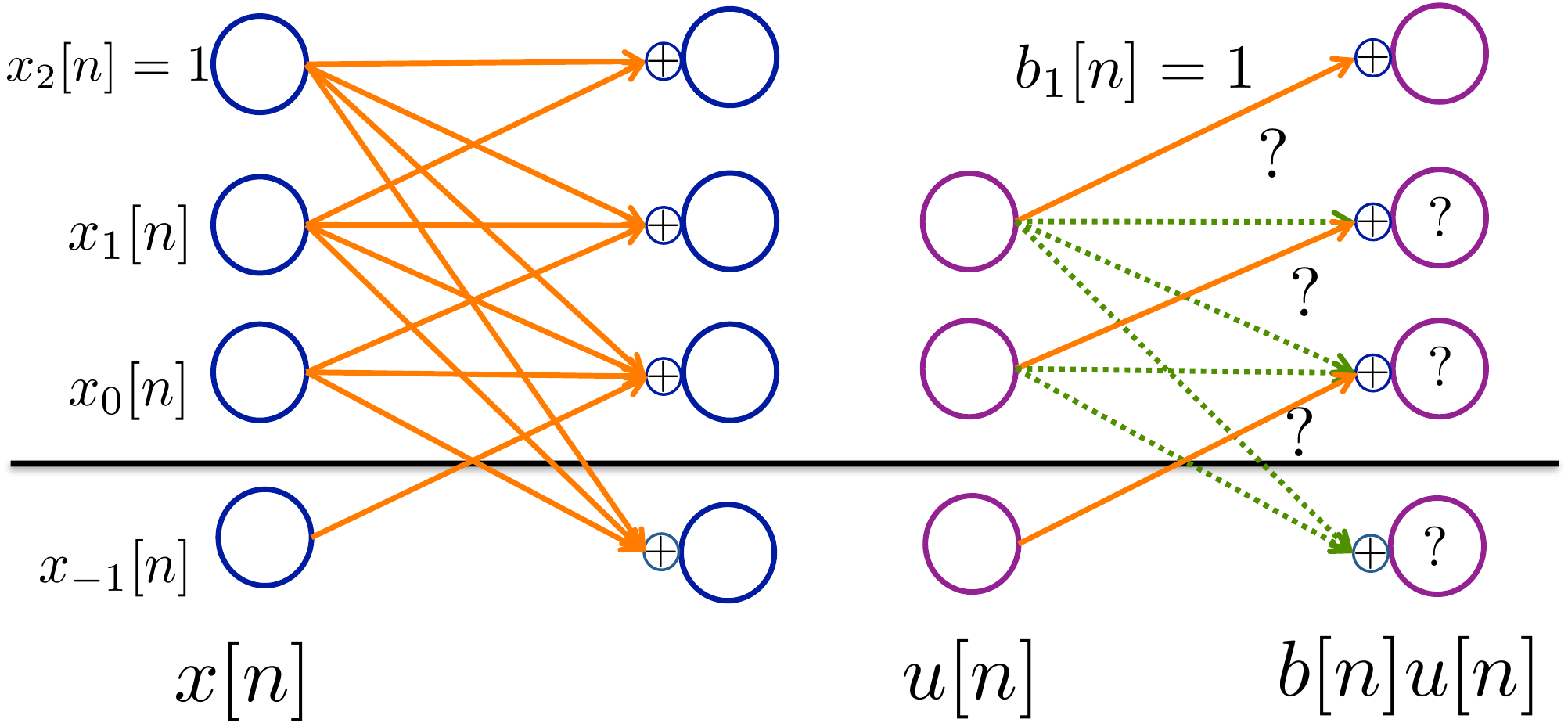}
\caption{This figure depicts the system in~\eqref{eq:cfwithb1}), with random control gain $b[n]$ but system growth $a$ fixed at $1$, so the system is not growing. The orange solid lines represent deterministic bits and the green dotted lines represent random bits in the control gain.}
\label{fig:cfccnoa}
\end{center}
\end{figure}

\begin{thm}
Consider the system $\S^{\CF}_{a}$ from~\eqref{eq:cfwithb} and the affiliated system $\mathcal{S}^{\CF}$ from~\eqref{eq:cfwithb1}, such that the actuation gains $b[n](z)$ are drawn identically in both cases. Then, $\S^{\CF}_{a}$ is stabilizable in a zero-error sense if and only if $C_{\ze}(\mathcal{S}^{\CF}) \geq g_{a}$.
\label{thm:cfzemeaning}
\end{thm}
\begin{proof}
This theorem follows naturally from the definition. 
\end{proof}

We can calculate the zero-error control capacity this using the following theorem that is intuitively illustrated in Figure~\ref{fig:cfccnoa}.
\begin{thm}
The zero-error control capacity for the system $\mathcal{S}^{\CF}$ in~\eqref{eq:cfwithb1} is equal to 
\begin{align*}
C_{\ze}\left(\mathcal{S}\right) = g_\det - g_\ran.
\end{align*} 
\label{thm:cfzecapacity}
\end{thm}
\vspace{-.8cm}
The heart of argument lies in the illustrations in Figure~\ref{fig:cfccnoa}. Once the pictorial representation is clear, the formalism of the proof is just counting. Before we prove this theorem, here is a key lemma that bounds the decay that can happen in one step, regardless of the system state $x[n]$.
\begin{lem}
For the system defined by eq.~\eqref{eq:cfwithb1}, for \textbf{any} state $x[n]$, the largest constant $C_{\ze,n}$ such that 
$$P\left(d_{n}-d_{n+1} \geq C_{\ze,n}\right) = 1,$$
is $C_{\ze,n} = g_\det- g_\ran$.  \label{lem:carryfreeonestep}
\end{lem}

\begin{proof}
\textbf{Achievability:} The achievability follows naturally by solving the appropriate set of linear equations to calculate controls to cancel the bits of the state.\\

\noindent \textbf{Converse:} To show the converse, we must show that for any $x[n]$ and for $u[n]$ that depends on $x[n]$ and its history we cannot beat $g_\det-g_{\ran}$,
\begin{align*}
P\left(d_{n} - d_{n+1} < g_\det - g_\ran+1\right) > 0.
\end{align*}

Consider any $x[n] = z^{d_{n}} + x_{d_{n}-1} z^{d_{n}-1} + \ldots$, with degree $d_{n}$. Let $u[n] = u_{m_{n}}n[]z^{m_{n}} + u_{m_{n}-1}[n] z^{m_{n}-1} + \ldots$ be any control action. The leading coefficients for both the state and the control must be $1$, else we could just reduce the degree, so we have that $u_{m_{n}}[n]=1$.
\begin{align*}
x[n+1] = &(x_{d_{n}}[n]z^{d_{n}} + x_{d_{n}-1}[n] z^{d_{n}-1} + \ldots) + \nonumber \\
&(z^{g_\det}+b _{g_\ran}[n] z^{g_\ran} + b_{g_\ran-1}[n] \cdot z^{g_\ran-1} + \ldots) \cdot (u_{m_{n}}[n]z^{m_{n}} + u_{m_{n}-1}[n]z^{m_{n}-1} + \ldots)
\end{align*}

We recall that $g_\det-g_\ran \geq 0$ by definition of $g_\det$ and $g_\ran$.
First we consider the case where $g_\det + m_{n} > d_{n}.$ Then, the degree at time $n+1$ is given as $d_{n+1} = g_\det + m_{n}$. So 
\begin{align*}
d_{n} - d_{n+1} = d_{n} - (g_\det + m_{n}) \leq 0 < g_\det - g_\ran + 1. 
\end{align*}

Next, assume $g_\det + m_{n} < d_{n}.$ Then, we must have that the degree of the state does not change after the control is applied, i.e. $d_{n+1} = d_{n}$. So we have that
\begin{align*}
d_{n} - d_{n+1} = 0 < g_\det - g_\ran + 1. 
\end{align*}

Finally, we consider the case when $g_\det + m_{n} = d_{n}.$ To calculate $d_{n+1}$, we first consider the coefficient of $z^{d_{n} - g_\ran + g_\det}$ in $x[n+1]$ as below:
\begin{align*}
x_{d_{n} - g_\ran +g_\det}[n] + b_{g_\ran}[n] \cdot u_{m_{n}}[n] + b_{g_\det}[n] \cdot u_{g_\ran + d_{n} - 2g_\det}[n] 
\end{align*}
Recall that that $g_\det + m_{n} = d_{n}$, and $u_{m_{n}}[n] = 1$, and all coefficients of $b$ between $z^{g_\det}$ and $z^{g_\ran}$ are zero, which gives the second and third terms above.

Now consider the term below. Recall $b_{g_{\det}}[n]=1$.
\begin{align}
&x_{d_{n} - g_\ran +g_\det}[n] + b_{g_\ran}[n] \cdot 1+ 1 \cdot u_{g_\ran + d_{n} - 2g_\det}[n] \nonumber \\
= &x_{d_{n} - g_\ran +g_\det}[n] + b_{g_\ran}[n] + u_{g_\ran + d_{n} - 2g_\det}[n] 
\end{align}
Since here $b_{g_\ran}[n]$ is a Bernoulli$-(\frac{1}{2})$, this term will be zero exactly with probability $\frac{1}{2}$. Hence, with probability $\frac{1}{2}$, $d_{n+1} \geq d_{n} - g_\ran + g_\det$. Hence, with probability $\frac{1}{2}$, $d_{n}-d_{n+1} \leq g_\ran - g_\det$. Thus $\P\left(d_{n}-d_{n+1} < 1+ g_\ran - g_\det\right) \geq \frac{1}{2}$, which gives the converse.

\end{proof}
Lemma~\ref{lem:carryfreeonestep} is the key ingredient that gives Theorem~\ref{thm:cfzecapacity}. This proof follows easily since the lemma decouples the controls at different time steps. The proofs of the real-valued notions of control capacity were inspired by this structure. The lemma frees us from considering time-varying or state-history dependent control strategies, which generally makes this style of converse difficult.

\begin{proof}[Proof of Thm.~\ref{thm:cfzecapacity}]
The lemma above bounds the decrease in degree of the state at any given time $n$, regardless of the control $u[n](z)$ or the state of the system $x[n](z)$. We have that
\begin{align*}
\frac{1}{n}  (d_{0} - d_{n}) = \frac{1}{n} \sum_{i=0}^{n-1}\left(d_{i} - d_{i+1}\right)
\end{align*}
Now, we know from Lemma~\ref{lem:carryfreeonestep} that, 
\begin{align*}
\P\left( \frac{1}{n} \sum_{i=0}^{n-1}\left(d_{i} - d_{i+1}\right) > \sum_{i=0}^{n-1} \left(g_\det - g_\ran\right)\right)= 1.
\end{align*}
Hence, we must have
\begin{align*}
C_{\ze}(\mathcal{S}) =\frac{1}{n} \sum_{i=0}^{n-1} \left(g_\det - g_\ran\right) = g_\det - g_\ran,
\end{align*}
which concludes the proof.
\end{proof}

\subsection{Stability in expectation}
Zero-error control capacity considers stability of the carry-free system with probability $1$. A weaker notion of stability is ``stability in expectation.'' This parallels the Shannon notion of capacity for real-valued systems. Since the definitions and theorems for this notion of stability are very similar to that of zero-error stability in the earlier section, we only state them here and omit details and proofs, which can be found in~\cite{gireejaBeast}.
\begin{defn}
The system~\eqref{eq:cfwithb1} is stablizable in expectation if there exists a control strategy $u[\cdot](z)$ such that for some $M< \infty$ we have that that $\limsup_{n\rightarrow\infty}\bigE\left[d_{n}\right] < M$. 
\end{defn}
We can also define a related notion of control capacity.
\begin{defn}
The ``Shannon'' control capacity of the system $\mathcal{S}^{\CF}$ from~\eqref{eq:cfwithb1}, $C_{\sh}(\mathcal{S})$, is defined as $$\max_{u[0](z), \cdots, u[n](z)}\lim_{n\rightarrow\infty}\frac{1}{n} \bigE[d_{0} - d_{n}].$$ 
\end{defn}

We can connect the stability of system $\mathcal{S}^{\CF}$ to $\mathcal{S}^{\CF}_{a}$ as in zero-error case through the following theorem.
\begin{thm}
Consider the system $\S^{\CF}_{a}$ from~\eqref{eq:cfwithb} and the affiliated system $\mathcal{S}^{\CF}$ from~\eqref{eq:cfwithb1}, such that the actuation channels, i.e. the $b[n](z)$ are drawn identically in both cases.  
$\S^{\CF}_{a}$ is stabilizable in expectation if and only if $C_{\sh}(\mathcal{S}) \geq g_{a}$. 
\label{thm:cfshmeaning}
\end{thm}
\begin{proof}
The proof follows naturally from the definition.
\end{proof}
The last theorem in this section explicitly calculates the carry-free Shannon control capacity in terms of the parameters of the random gain in the carry-free model. 
\begin{thm}
The Shannon control capacity for the system $\mathcal{S}$ is given by 
\begin{align*}
C_{\sh}\left(\mathcal{S}^{CF}\right) = g_\det - g_\ran + 1.
\end{align*} 
\label{thm:shcfcapacity}
\end{thm}
This proof is similar to the proof of Theorem~\ref{thm:cfzecapacity}. The extra bit of capacity is gained because the control only needs to cancel it ``in expectation''. Consider the bit at level $d_{n} + g_{\ran}-1$. The controller can only cancel this bit with probability $\frac{1}{2}$ (see Figure~\ref{fig:cfccnoa}. Thus, the bit at level $d_{n} + g_{\ran} - i$ for $g_{\ran} > -i$ is set to zero by the controller with probability $2^{-i}$, and summing the geometric terms gives the result.

\subsection{Carry-free model with side information}
Carry-free models easily allow us to illustrate how bits of side-information impact the system. We can count the non-random bits in the output and ``compute'' the value of a piece of side-information. We illustrate this here with an example.
\begin{figure}[htbp]
\begin{center}
\subfigure[]{
\includegraphics[width=.35\textwidth]{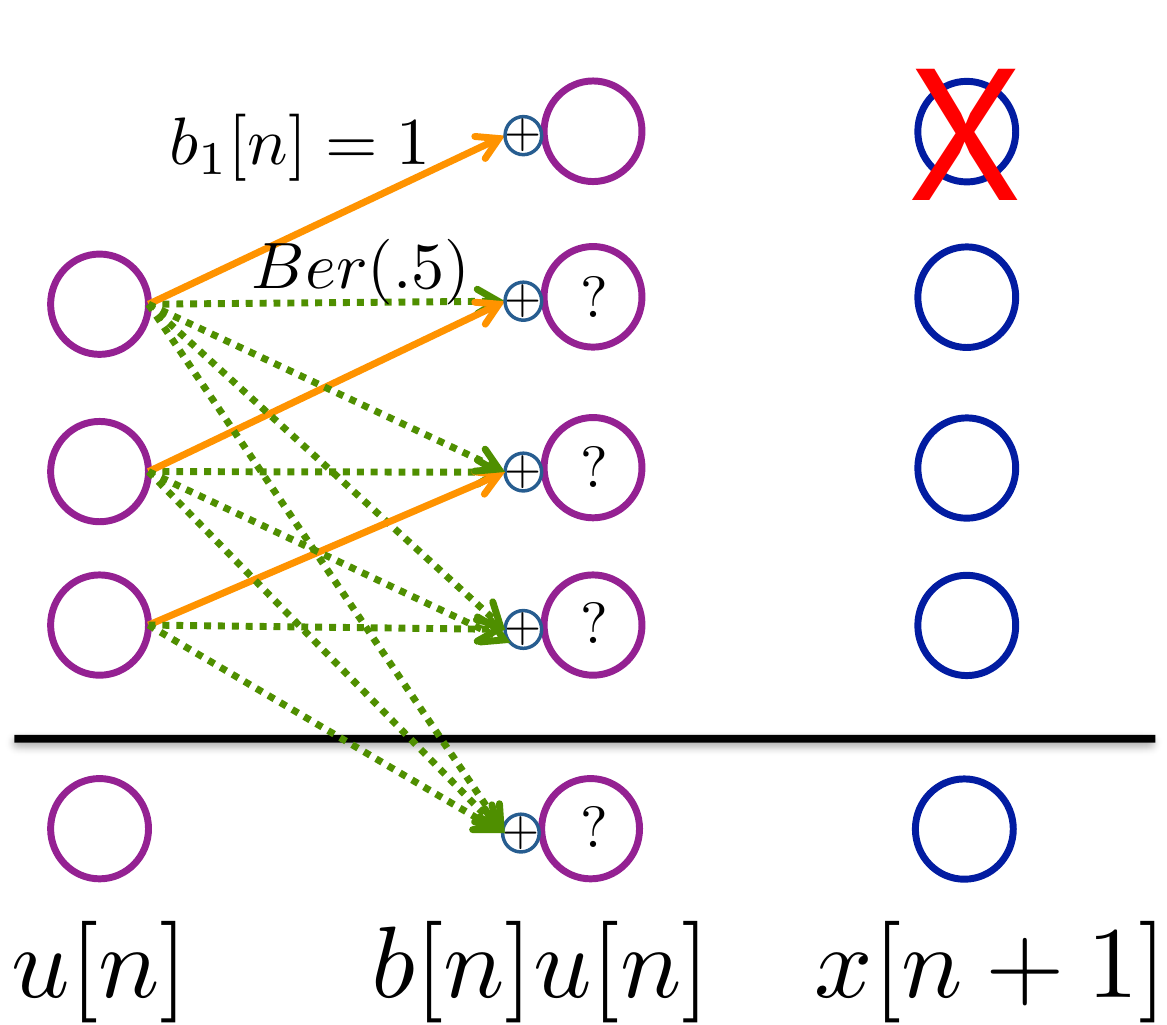}
}
\subfigure[]{\includegraphics[width=.35\textwidth]{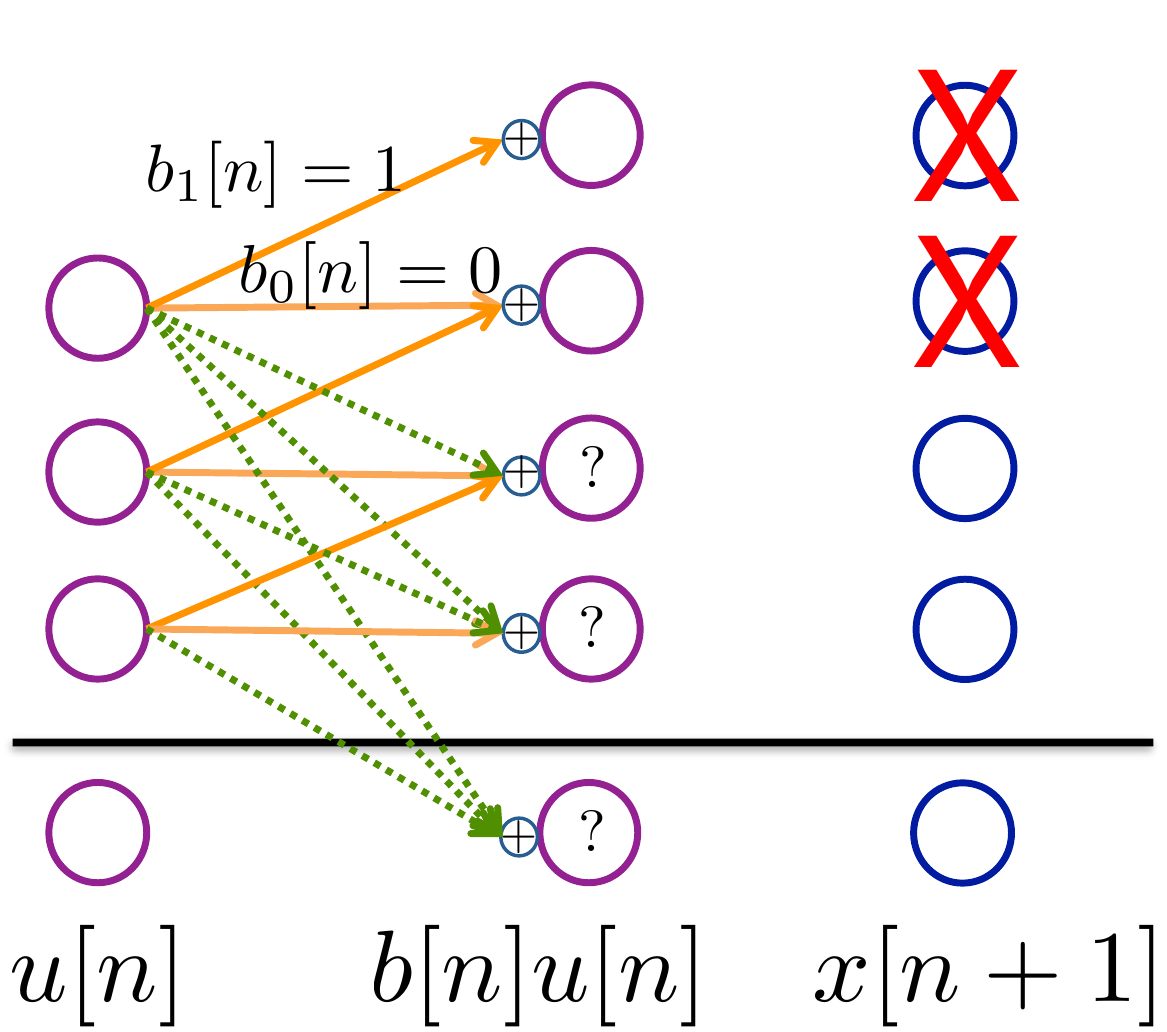}
}
\caption{This system has the highest deterministic link at level $g_{det} = 1$ and the highest unknown link at $g_{ran} = 0$. Bits $b_{-1}[n], b_{-2}[n], \cdots$ are all random Bernoulli-$(\frac{1}{2})$. As a result the controller can only influence the top bits of the control going in, and can only cancel one bit of the state. If one extra bit $b_{0}$ were known to the controller, it could cancel a second bit of the state.}
\label{fig:cfuncertain}
\end{center}
\end{figure}
In Figure~\ref{fig:cfuncertain}, we consider a simple bit-level carry-free model that is the counterpart of system~\eqref{eq:sys1}. Say the control gain $B[n]$ has one deterministic bit, so that $g_{det}=1$, but all lower bits are random Bernoulli~$-\left(\frac{1}{2}\right)$ bits. Thus, the controller can only cancel $1$ bit of the state each time with probability $1$ and the system has zero-error control capacity of $1$. If two bits (the most significant bit and the bit after that) of $B[n]$ were deterministic, the zero-error control capacity would be $2$. For instance,  if the value of the bit at level 0, i.e.~$b_{0}$, were also known, then we could tolerate a growth through $a$ of two bits at a time. We can think of this as the value of the side information $b_{0}$ for this problem. 

\subsection{A carry-free counterexample for the value of side information}
In the portfolio theory literature, it is known that the maximum increase in doubling rate due to side information $Z$ for a set of stocks distributed as $Q$ is upper bounded by $I\left(Q;Z\right)$. In this case, a bit buys a bit. It is tempting to conjecture a similar result in the case of control systems, however this turns out not to be true. 

\begin{figure}[htbp]
\begin{center}
\subfigure[]{
\includegraphics[width=.35\textwidth]{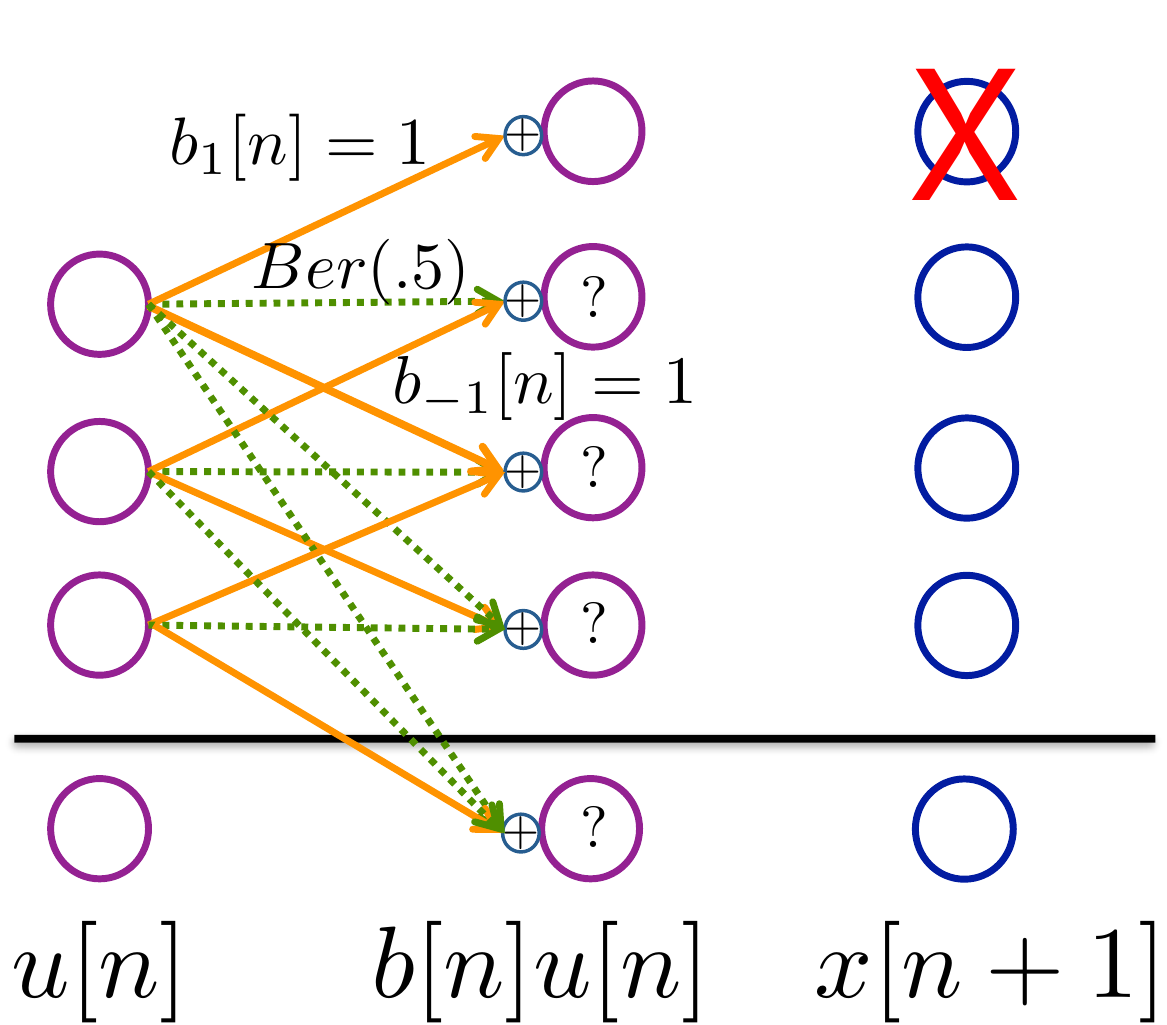}
}
\subfigure[]{
\includegraphics[width=.35\textwidth]{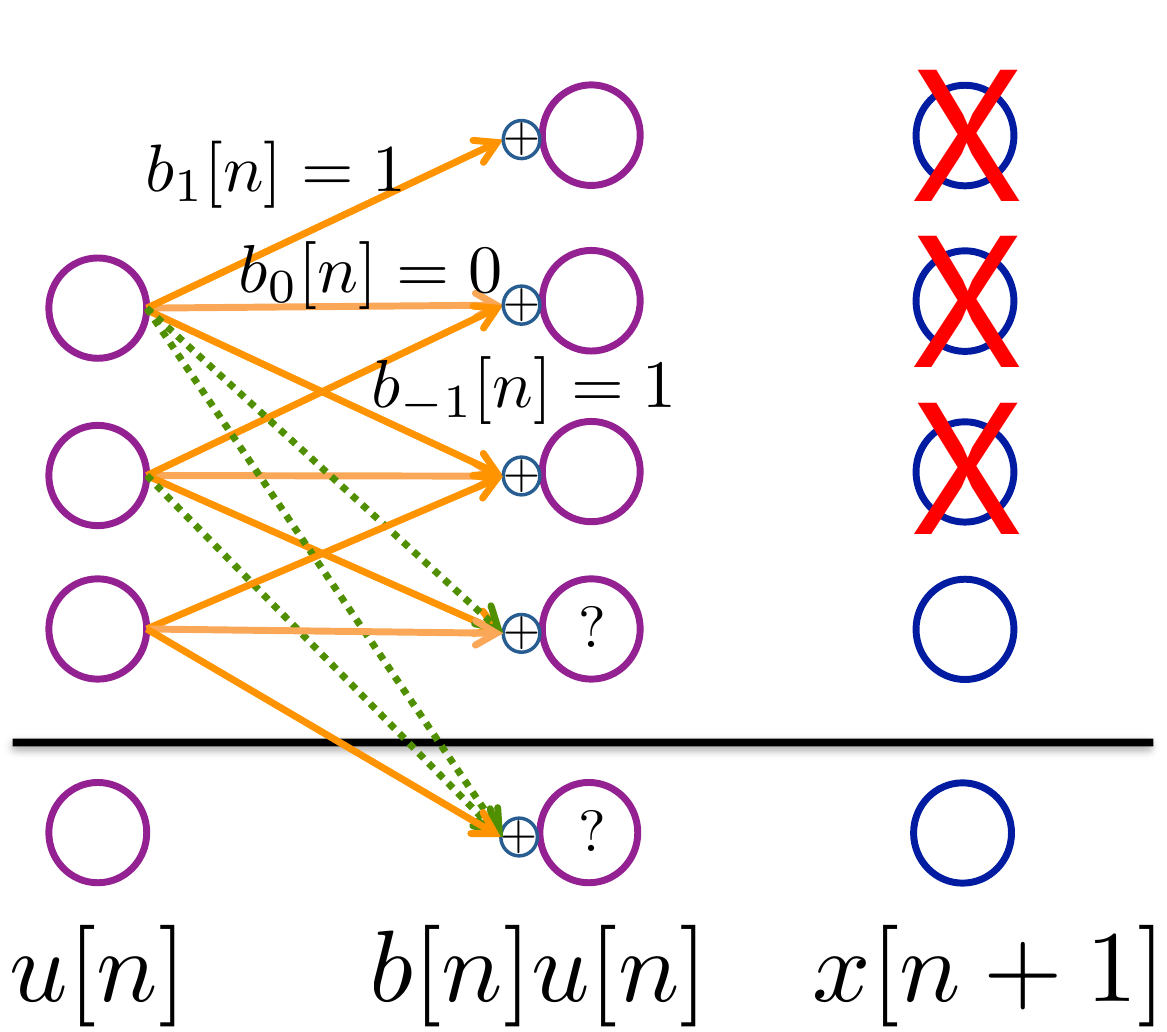}
}
\caption{Consider the following gain for the controller in (a): $b_{1}[n]=1, b_{-1}[n]=1$ are deterministically known, but all other links are Bernoulli-$(\frac{1}{2})$. Only a gain of $\log a = 1$ can be tolerated in this case. Now, say side information regarding the value of $b_{0}[n]$ is received as in (b). This suddenly buys the controller not just one, but two bits of growth.}
\label{fig:counterexample}
\end{center}
\end{figure}

To see a counter example, consider the carry-free model in Figure~\ref{fig:counterexample}. Here $u[n]$ is the control action, and $b[n]u[n] = z[n]$ is the control effect. In Figure~\ref{fig:counterexample}(a) the uncertainty in $b_{0}[n]$ does not allow the controller to utilize the knowledge that $b_{-1}[n] = 1$. However, one bit of information $b_{0}[n]$ in Figure~\ref{fig:counterexample}(b), lets the controller buy two bits of gain in the tolerable growth rate as explained in the caption. 

This carry-free model represents a real system where $B[n]$ is drawn from a mixture of disjoint uniform distributions, as in Figure~\ref{fig:dist1}(a). The first most significant bit and the third most significant bit are known, but the second most significant bit is not known. The first bit tells us whether $B[n]$ comes from group A-B or group C-D. The third bit only discriminates to the level that is shown in Figure~\ref{fig:dist1}(b), i.e. the controller only knows that $B[n]$ belongs to one of the two orange boxes. So the variance of the distribution isn't actually lowered by much due to the side information. The side information containing the second bit finally lowers the variance, as in Figure~\ref{fig:dist1}(c). Thus, the gain this one-bit of information provides is worth more than a bit.

\begin{figure}[htbp]
\begin{center}
\subfigure[]{
\includegraphics[width=.3\textwidth]{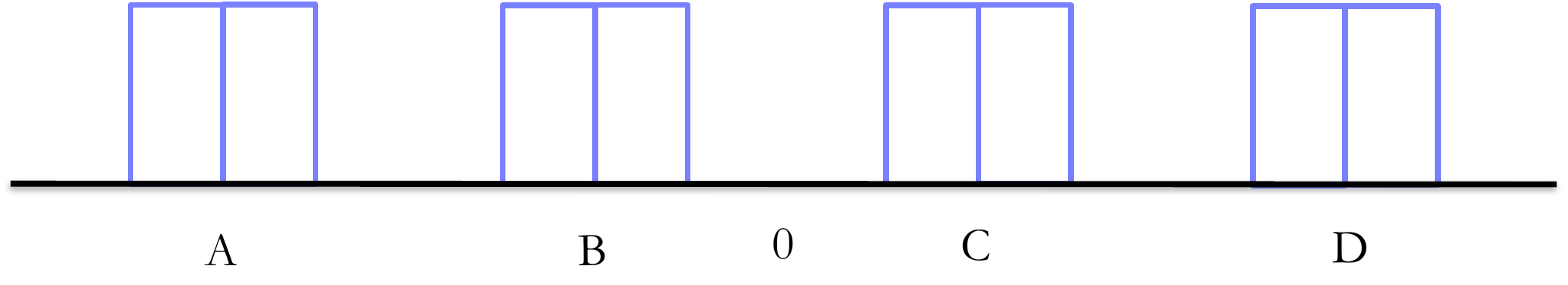}
}
\subfigure[]{
\includegraphics[width=.3\textwidth]{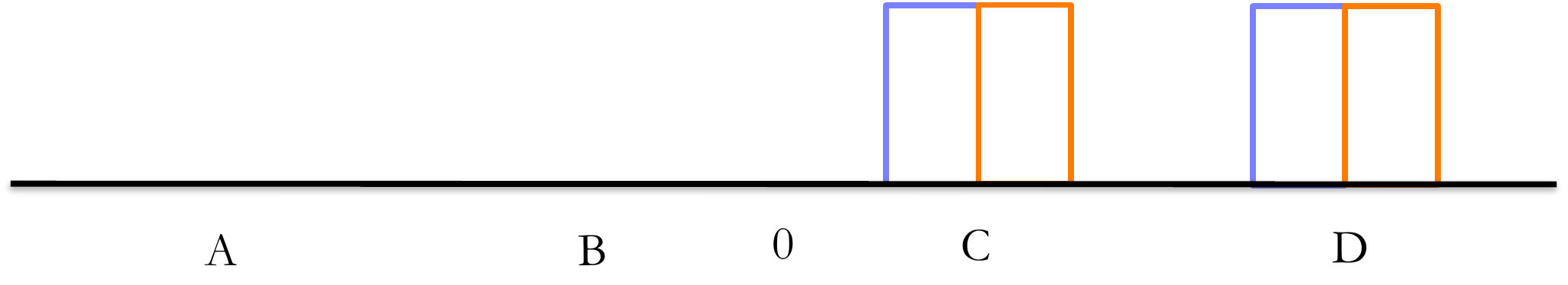}
}
\subfigure[]{
\includegraphics[width=.3\textwidth]{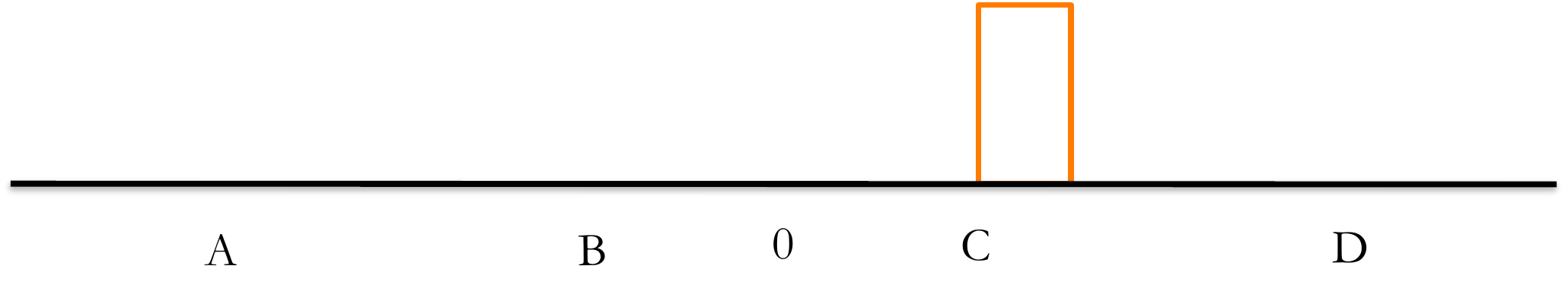}
}
\caption{(a) The distribution of $B[n]$ for the carry-free model. (b) The first and the third bit together tell the controller that $B[n]$ comes from one of the orange parts of the distribution. Since there are two orange sections that are far away, the effective variance of $B[n]$ is not reduced. (c) Once the second bit is also known, the fact that the controller already knew the third bit becomes useful.}
\label{fig:dist1}
\end{center}
\end{figure}

\subsection{A communication aside} 
\begin{figure}[htbp]
\begin{center}
\includegraphics[width=.3\textwidth]{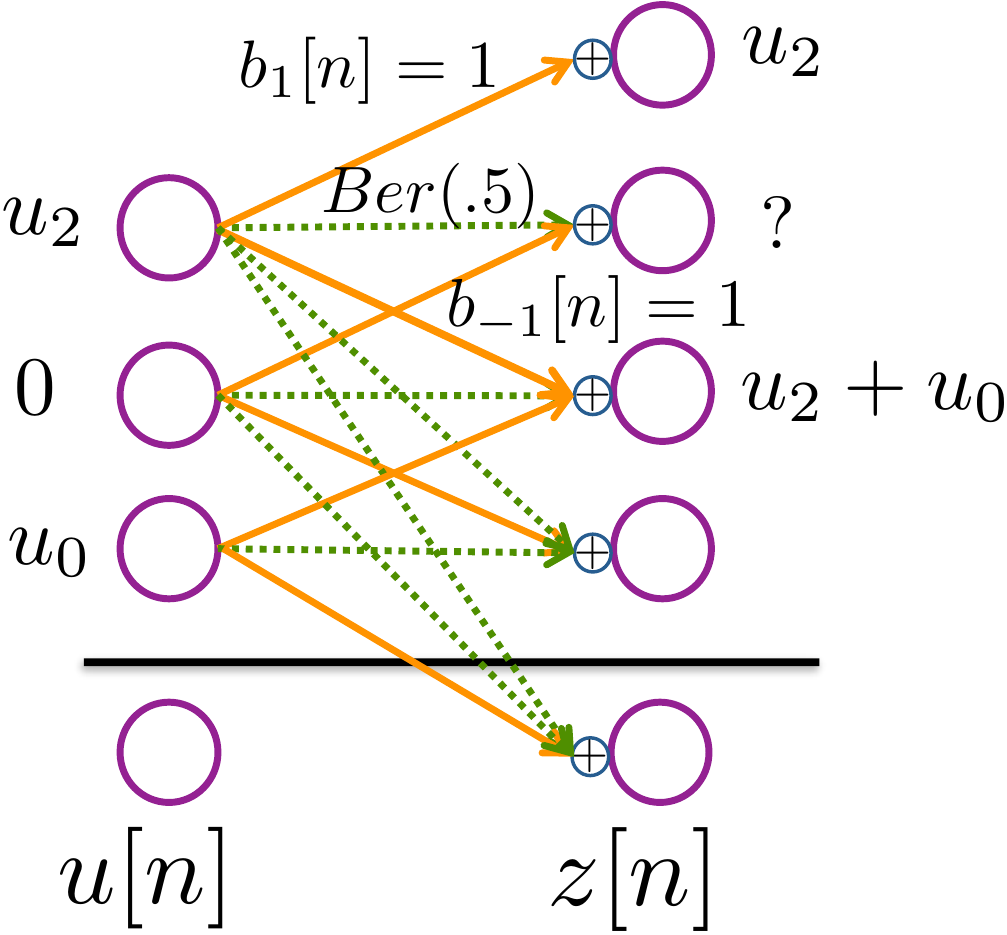}
\caption{With two known bits on the gain, the decoder can decode two additional bits about the message $u[n]$ from the received signal $z[n]$. However, these bits are decoded at specific positions. It is not possible to get information out of the received signal position $z_{2}$.}
\label{fig:commCE}
\end{center}
\end{figure}

If the problem was that of pure communication, we could still decode two bits of information about $u[n]$ from $b[n]u[n] = z[n]$. See Figure~\ref{fig:commCE}. Let $u_{2}$ and $u_{0}$ be the information carrying bits, and set $u_{1}= 0$ to zero. Then $z_{3} = u_{2}$ and $z_{1} = u_{2} + u_{0}$. With two equations and two unknowns, both $u_{2}[n]$ and $u_{0}[n]$ can be recovered at the decoder. In the control problem, this is not possible because of the contamination that is introduced by $b_{0}$ at $z_{2}$. While communication systems can choose which bits contain relevant information, control systems do not have that flexibility. A bit at a predetermined position must be cancelled or moved by the control action. 

In the case of portfolio theory, it is possible to hedge across uncertainty in the system and get ``partial-credit'' for uncertain quantities. This is not possible in control systems since it is not possible to hedge a control signal in the same way one can hedge a bet.  

\section{Zero-error control capacity Appendix} \label{sec:appendixzeroerror}

{\begin{restatable}{lem}{lemboundedlem}
\begin{align}
\underset{d \in \mathbb{R}}{\min}~\underset{B \in [b_{1}, b_{2}]}{\max} \bigr| 1 + B\cdot d \bigr| =
\begin{cases}
\frac{|b_{2}-b_{1}|}{|b_{2}+b_{1}|} ~\textrm{if}~~0 \not\in [b_{1}, b_{2}]\\
1~~\textrm{if}~~0 \in [b_{1},b_{2}]. \label{eq:lem1}
\end{cases}
\end{align}
Thus, if $B$ is a random variable with essential infimum $b_1$ and essential supremum $b_2$, we know that $\forall d, \forall \epsilon > 0$,
\begin{align}
P\left( \log \biggr| 1 + B\cdot d \biggr| < \log\biggr| \frac{b_{2}-b_{1}}{b_{1}+b_{2}}\biggr| -\epsilon \right) < 1~\textrm{if}~0 \not\in [b_{1}, b_{2}].
\end{align}
and 
\begin{align}
P\left( \log \bigr| 1 + B\cdot d \bigr|< -\epsilon \right) < 1~\textrm{if}~0 \in [b_{1},b_{2}].
\end{align}
\label{lem:boundedlem}
\end{restatable}}
\begin{proof}
First, we consider the case where $0 < b_{1} < b_{2}$. The case where $b_{1} < b_{2} < 0$ follows similarly. To show that the bound in~\eqref{eq:lem1} is achievable,  choose $d^{*}= -\left(\frac{b_{1}+b_{2}}{2}\right)^{-1}$. We claim this is the minimizing $d$.
Then, the maximum value of $\bigr|1 + B\cdot d^{*}\bigr|$ is attained when $B$ is realized as $B=b_{2}$ or $B = b_1$.
\begin{align*}
\left|1 - b_{2}\cdot \left(\frac{b_{1}+b_{2}}{2}\right)^{-1}\right| = \left|1 - b_{1}\cdot \left(\frac{b_{1}+b_{2}}{2}\right)^{-1}\right|  = \frac{b_{2}-b_{1}}{b_{1}+b_{2}}
\end{align*}
Now, suppose $\exists d \neq d^{*} \in \mathbb{R}$ that such that $\forall B \in [b_{1}, b_{2}]$, 
\begin{align*}
\left|1 + B \cdot d\right| < \left|\frac{b_{2}-b_{1}}{b_{1}+b_{2}} \right| = \frac{b_{2}-b_{1}}{b_{1}+b_{2}}.
\end{align*}
This implies for $B=b_{1}$ that $\left|1 + b_{1} \cdot d\right| < \frac{b_{2}-b_{1}}{b_{1}+b_{2}}.$. This can be written as:
\begin{align*}
& \frac{b_{1}-b_{2}}{b_{1}+b_{2}} < 1 + b_{1} \cdot d < \frac{b_{2}-b_{1}}{b_{1}+b_{2}} \\
\Rightarrow&  d < \frac{-2}{b_{1}+b_{2}}.
\end{align*}
Further, we also have for $B=b_{2}$ that $|1 + b_{2} \cdot d| < \frac{b_{2}-b_{1}}{b_{1}+b_{2}}.$ This gives
\begin{align*}
& \frac{b_{1}-b_{2}}{b_{1}+b_{2}} < 1 + b_{2} \cdot d < \frac{b_{2}-b_{1}}{b_{1}+b_{2}} \\
\Rightarrow&  d > \frac{-2}{b_{1}+b_{2}},
\end{align*}
which gives us a contradiction. Hence, $\forall d \in \mathbb{R}$, $\exists B \in [b_{1}, b_{2}]$, 
\begin{align*}
\left|1 + B \cdot d\right| \geq \biggr|\frac{b_{2}-b_{1}}{b_{1}+b_{2}} \biggr|. 
\end{align*}

Given the minmax of $|1 + B \cdot d|$ is  $\biggr|\frac{b_{2}-b_{1}}{b_{1}+b_{2}}\biggr|$, then for any bounded random variable $B$ with essential supremum $b_2$ and essential infimum $b_1$, we know by the definition of essential suprema and infima that: 
\begin{align}
\P\left( \log \biggr| 1 + B \cdot d \biggr| < \log\biggr| \frac{b_{2}-b_{1}}{b_{1}+b_{2}}\biggr| -\epsilon \right) < 1.
\end{align}

Now, we consider the case where $b_{1} \leq 0 < b_{2}.$ Here, choose $d^{*} = 0$ to show that the bound is achievable. Then $\left|1+ Bd\right| = 1$ for all $B$.

Suppose there exists $ d \neq d^{*} \in \mathbb{R}$ such that $\forall B \in [b_{1}, b_{2}]$,
$|1 + B \cdot d | < 1$.
This implies that for $B=b_{1}$: 
\begin{align*}
& -1 < 1 + b_{1} \cdot d < 1 ~\Rightarrow -2  < b_{1} \cdot d <  0
\end{align*}
Thus, $b_{1}$ cannot be equal to zero. 
If $b_{1} < 0$, then  since $b_{1} \cdot d  < 0$, we must have $d > 0$. Now,  we also know that 
\begin{align*}
|1 + b_{2} \cdot d| < 1~\Rightarrow& -1 < 1 + b_{2} \cdot d < 1 \\
\Rightarrow& -2 < b_{2} \cdot d < 0 
\end{align*}
$b_{2} \cdot d < 0$ implies $d < 0$,  which again gives us a contradiction. 

Hence, for all $d \in \mathbb{R}$, there exists $B \in \{b_{1}, b_{2}\}$ such that $|1 + B \cdot d| \geq 1.$ 

Consequently, for any bounded random variable $B$ with essential supremum $b_2$ and essential infimum $b_1$ we have that: 
\begin{align}
\P\left( \log \bigr| 1 + B[n]\cdot d \biggr|< -\epsilon \right) < 1.
\end{align}

\end{proof}


\lemzeroerroronestep*

\begin{proof}
\textbf{Achievability:}\\
The case where $X[n+1] = 0$ is trivial, so we focus on the case where $\left|X[n+1]\right| > 0$. Further, we see that if $0 \in [b_{1}, b_{2}]$, choosing $U[n] = 0$ gives the result. 

Let $0 \not\in [b_{1}, b_{2}]$. We would like to show that $\P\left(-\log \biggr|\frac{X[n+1]}{x[n]} \biggr| \geq \log \biggr|\frac{b_{1}+b_{2}}{b_{2}-b_{1}} \biggr| \right) = 1$. For any $x[n]\neq0$, choose $U[n] = - \left(\frac{b_{1}+b_{2}}{2}\right)^{-1} x[n]$. Then,
\begin{align*}
X[n+1] &= x[n] - B[n] \left(\frac{b_{1}+b_{2}}{2}\right)^{-1} x[n]\\
&= x[n] \left( 1 - \frac{2 B[n]}{b_{1}+b_{2}}\right).
\end{align*}
Now, note that the following inequalities are equivalent.
\begin{align*}
&-\log \biggr|\frac{X[n+1]}{x[n]} \biggr| \geq \log \biggr|\frac{b_{1}+b_{2}}{b_{2}-b_{2}}\biggr|\\
\Longleftrightarrow &\log \biggr|\frac{X[n+1]}{x[n]} \biggr| \leq \log  \biggr|\frac{b_{2}-b_{1}}{b_{1}+b_{2}}\biggr|\\
\Longleftrightarrow &\biggr| 1 - \frac{2 B[n]}{b_{1}+b_{2}}\biggr| \leq~  \biggr|\frac{b_{2}-b_{1}}{b_{1}+b_{2}}\biggr|.
\end{align*}

Hence, the event $\mathcal{E}_{1} = \left\{-\log \biggr|\frac{X[n+1]}{x[n]} \biggr| \geq \log \biggr|\frac{b_{1}+b_{2}}{b_{2}-b_{1}} \biggr| \right\}$, is identical to the event $\mathcal{E}_{2} =\left\{ \biggr| 1 - \frac{2 B[n]}{b_{1}+b_{2}}\biggr| \leq~  \biggr|\frac{b_{2}-b_{1}}{b_{1}+b_{2}}\biggr| \right\}$.

But we know from Lemma~\ref{lem:boundedlem} that the maximum of $\biggr| 1 - \frac{2 B[n]}{b_{1}+b_{2}}\biggr|$ is $\biggr|\frac{b_{2}-b_{1}}{b_{1}+b_{2}}\biggr|$. Hence, the event $\mathcal{E}_{2}$ occurs with probability $1$, which implies $\mathcal{E}_{1}$ occurs with probability $1$, which proves achievability.

\noindent \textbf{Converse:} Here, we must show $\forall x[n]\neq 0, \forall \epsilon > 0$, for any $U[n]$,
\begin{align*}
\P\left(- \log \biggr| \frac{X[n+1]}{x[n]} \biggr| > \log\biggr| \frac{b_{1}+b_{2}}{b_{2}-b_{1}}\biggr| +\epsilon \right) < 1.
\end{align*}

First, consider the case $0 < b_{1} < b_{2}$ or $b_{1} < b_{2} < 0.$ Then,
\begin{align*}
\P\left(- \log \biggr| \frac{X[n+1]}{x[n]} \biggr| > \log\frac{|b_{2}+b_{1}|}{|b_{2}-b_{1}|} +\epsilon \right) 
= \P\left( \log \biggr| \frac{x[n] + B[n]U[n]}{x[n]} \biggr| < \log\biggr| \frac{b_{2}-b_{1}}{b_{1}+b_{2}}\biggr| -\epsilon \right) 
\end{align*}
Let $\frac{U[n]}{x[n]} = d$, for any strategy $U[n].$ Then the above display must equal:
\begin{align*}
&=\P\left( \log \biggr| 1 + B[n]\cdot d \biggr| < \log\biggr| \frac{b_{2}-b_{1}}{b_{1}+b_{2}}\biggr| -\epsilon \right) \\
&< ~1.~(\textrm{by Lemma~\ref{lem:boundedlem}})
\end{align*}

Next if $b_{1} \leq 0 < b_{2}$  or $b_{1} < 0 \leq b_{2}$ then
\begin{align*}
&\P\left(- \log \biggr| \frac{X[n+1]}{x[n]} \biggr| > \epsilon  \right)
= \P\left( \log \biggr| \frac{X[n] + B[n]U[n]}{X[n]} \biggr|< -\epsilon \right)
\end{align*}
Let $\frac{U[n]}{x[n]} = d$, for any strategy $U[n].$ Then the above display must equal:
\begin{align*}
&=\P\left( \log \bigr| 1 + B[n]\cdot d \biggr|< -\epsilon \right)
< ~1,
\end{align*}
which is less than 1 by Lemma~\ref{lem:boundedlem}. This concludes the proof.
\end{proof}
\vspace{-8mm}
\label{sec:appendix1}

\bibliographystyle{IEEEtran}
\bibliography{ccjournal}

\end{document}